\PassOptionsToPackage{table}{xcolor}

\documentclass[review=false]{jfp-epi} 

\setcopyright{none}

\usepackage[english]{babel}
\usepackage{listings}
\usepackage{url}
\usepackage{csquotes}
\usepackage{thmtools}
\usepackage{marginnote}
\usepackage{stmaryrd}
\usepackage{placeins}
\usepackage{tabularx}
\usepackage{mathpartir}
\usepackage{ifthen}
\usepackage{enumerate}
\usepackage[inline]{enumitem}
\usepackage{fancyvrb}

\definecolor{GrayBgColor}{rgb}{0.9, 0.9, 0.9}
\definecolor{GrayFgColor}{rgb}{0.4, 0.4, 0.4}
\definecolor{StringColor}{rgb}{0.0, 0.38039, 0.141176}
\usepackage{pgfplots, pgfplotstable}
\usetikzlibrary{patterns}

\AtEndPreamble{%
  \usepackage{cleveref}
  \newtheorem{property}[theorem]{Property}

}

\renewcommand{\theequation}{\arabic{equation}}

\newcommand{\eqwalizer}{eqWAlizer}
\newcommand{\gradualizer}{Gradualizer}
\newcommand{\ety}{Etylizer}
\newcommand{\dialyzer}{Dialyzer}
\newcommand\MinErl{\ensuremath{\lambda_{\textsf{erl}}}}

\newcommand*{\medcup}{\mathop{\scalebox{0.8}{\ensuremath{\bigcup}}}}%
\newcommand*{\medwedge}{\mathop{\scalebox{0.8}{\ensuremath{\bigwedge}}}}%
\newcommand*{\medvee}{\mathop{\scalebox{0.8}{\ensuremath{\bigvee}}}}%
\newcommand{\doublewedge}{%
  \mathbin{
    \mathchoice{\wedge\mkern-15mu\wedge}
               {\wedge\mkern-15mu\wedge}
               {\wedge\mkern-9mu\wedge}
               {\wedge\mkern-8mu\wedge}
    }
}
\newcommand{\bigdoublewedge}{%
  \mathop{
    \mathchoice{\bigwedge\mkern-15mu\bigwedge}
               {\bigwedge\mkern-15mu\bigwedge}
               {\bigwedge\mkern-12.5mu\bigwedge}
               {\bigwedge\mkern-11mu\bigwedge}
    }
}

\newcommand{\figureboxSingle}[1]
        {\fbox{\begin{minipage}{\columnwidth} #1 \end{minipage}}}

\newcommand{\boxfigSingle}[3]
           {\begin{figure}\figureboxSingle{#3}\caption{\label{#1}#2}\end{figure}}

\newcommand\Multi[2][!*NEVER USED ARGUMENT*!]{%
  \ifthenelse{\equal{#1}{!*NEVER USED ARGUMENT*!}}{#2_{i \in I}}{#2_{#1}}%
}
\newcommand{\kw}[1]{\ensuremath{\mathtt{#1}}}
\newcommand\Name[1]{\ensuremath{\mathtt{#1}}}

\newcommand\Rule[1]{\TirName{#1}}
\def\RuleForm#1{{\setlength{\fboxrule}{0.5pt}\fbox{\normalsize \ensuremath{#1}}}}
\newcommand\RuleSection[2]{%
  \begin{tabularx}{\textwidth}{Xr}%
    #1\hfill & {\small \it #2}%
  \end{tabularx}%
}

\makeatother
\lstdefinelanguage{erlang}{%
   morekeywords={after,and,case,catch,div,end,exit,export,if,import,module,of,or,%
     receive,rem,-spec,throw,when,fun,opaque%
   },
   otherkeywords={-spec,-type},%
   sensitive,%
   basicstyle=\lst@ifdisplaystyle\scriptsize\linespread{0.9}\footnotesize\fi\ttfamily,
   keywordstyle=\color{black}\bf\ttfamily,
   numberstyle=\tiny\color{GrayFgColor}\sffamily\raisebox{0.6pt},
   morecomment=[l]{\%},%
   morestring=[b]",%
  }[keywords,comments,strings]%
\makeatletter

\lstdefinestyle{InlineDefault}{
  basicstyle=\small\ttfamily, 
}

\makeatletter
\let\lstinline@orig\lstinline
\DeclareRobustCommand{\lstinline}{%
  \@ifnextchar[{\lstinline@with}{\lstinline@without}%
}
\newcommand{\lstinline@with}[1]{%
  \lstinline@orig[style=InlineDefault,#1]%
}
\newcommand{\lstinline@without}{%
  \lstinline@orig[style=InlineDefault]%
}
\makeatother

\newcommand\IsSubty{\leq}

\newcommand\ValMatches{\mathbin{\#}}
\newcommand\TyEquiv{\simeq}
\newcommand\MoreGeneral\sqsubseteq
\newcommand\Inter{\wedge}
\newcommand\Union{\vee}
\newcommand\cupdisjoint{\mathbin{\stackrel{.}{\cup}}}
\newcommand\medcupdisjoint{\mathop{\stackrel{.}{\medcup}}}

\newcommand\InterBig{\medwedge}
\newcommand\UnionBig{\medvee}
\newcommand\WithoutTy\setminus

\newcommand\Pairty[2]{#1 \times #2}

\newcommand\ErlTop{\top}
\newcommand\ErlBot{\bot}
\newcommand\Neg{\neg}
\newcommand\Int{\Name{int}}
\newcommand\Float{\Name{float}}

\newcommand\AnyPair{\Name{pair}}
\newcommand\AnyFun{\Name{fun}}

\newcommand\Polyty{\sigma}
\newcommand\Monoty{t}
\newcommand\MonotyAlt{u}
\newcommand\Tyvar{\alpha}
\newcommand\TyvarAux{\beta}
\newcommand\TyvarSet{A}
\newcommand\TyScm[1]{\forall #1\,.\,}
\newcommand\AllConst{\mathbb{K}}
\newcommand\AllConstty{\mathbb{S}}
\newcommand\Constty{\iota}
\newcommand\Basety{b}

\newcommand\AllVar{\mathbb{X}}
\newcommand\AlltyVar{\mathbb{V}}
\newcommand\GuardTy{\mathit{gt}}
\newcommand\Gt{\GuardTy}

\newcommand\LetrecSym{\kw{letrec}}
\newcommand\Letrec[2]{\LetrecSym~#1~\kw{in}~#2}
\newcommand\Case[1]{\kw{case}~#1~\kw{of}~}

\newcommand\DefSym{\mathit{def}}
\newcommand\Prog{\mathit{prog}}
\newcommand\Wildcard{\_}
\newcommand\EqSynSym{\mathbin{\texttt{=}}}
\newcommand\Def[3]{#1 \mathop{\texttt{:}} #2 \EqSynSym #3}
\newcommand\DefNt[2]{#1 \EqSynSym #2}
\newcommand\Abs[1]{\lambda #1 \texttt{.}}
\newcommand\App[1]{#1\,}
\newcommand\Pair[2]{\texttt{(}#1\texttt{,} #2\texttt{)}}
\newcommand\When{\mathrel{\kw{when}}}
\newcommand\Cls{\mathit{cls}}
\newcommand\Pg{\mathit{pg}}

\newcommand\CaptSym{\scalebox{0.8}{\textsf{\normalshape \$}}}
\newcommand\Capt[1]{\CaptSym{}#1}

\newcommand\WithGuard[2]{#1 \When #2}
\newcommand\PatCls[2]{#1 \to #2}
\newcommand\Is[2]{\Name{is}_{#1}(#2)}
\newcommand\Const{\kappa}
\newcommand\GuardOracle{\kw{oracle}}
\newcommand\GuardAndSym{\mathop{\kw{and}}}
\newcommand\GuardAnd[2]{#1 \GuardAndSym #2}
\newcommand\GuardTrue{\kw{true}}

\newcommand\EmptyEnv{\emptyset}
\newcommand\Dom[1]{\mathsf{dom}(#1)}
\newcommand\Forall[1]{(\forall #1)}
\newcommand\FreeSym[1]{\mathsf{fv_{#1}}}
\newcommand\Free[1]{\mathsf{fv}(#1)}
\newcommand\FreeP[1]{\mathsf{fv_p}(#1)}
\newcommand\FreeE[1]{\mathsf{fv_e}(#1)}
\newcommand\FreeG[1]{\mathsf{fv_g}(#1)}
\newcommand\FreePg[1]{\mathsf{fv_{pg}}(#1)}
\newcommand\PatBound[1]{\mathsf{pbv}(#1)}
\newcommand\FreeTyVarsName{\mathsf{ftv}}
\newcommand\FreeTyVars[1]{\FreeTyVarsName(#1)}
\newcommand\TyVars[1]{\FreeTyVars{#1}}
\newcommand\TyVarsFunc[3]{\FreeTyVarsName_{#1}(#2,#3)}
\newcommand\EnvsSym{\mathsf{envs}}
\newcommand\Envs[1]{\EnvsSym(#1)}

\newcommand\Angle[1]{\langle#1\rangle}
\newcommand\MetaPair[2]{\Angle{#1; #2}}
\newcommand\Tysubst{\theta}
\newcommand\TysubstStar{\Tysubst^\star}
\newcommand\TysubstTild{\tilde{\Tysubst}}
\newcommand\TysubstTop{\Tysubst^\top}

\newcommand\Valsubst{\eta}
\newcommand\GenSym{\mathsf{gen}}
\newcommand\Gen[1]{\GenSym(#1)}
\newcommand\EquivSym{\mathsf{equiv}}
\newcommand\Equiv[1]{\EquivSym(#1)}

\newcommand\ReduceSym{\leadsto}
\newcommand\Reduce[3]{#1 \vdash #2 \ReduceSym #3}
\newcommand\ReduceBase[3]{#1 \vdash #2 \ReduceSym_{\mathsf{b}} #3}
\newcommand\PatSubst[2]{\GenPatSubst{#1}{\FunEnv}{#2}}
\newcommand\GenPatSubst[3]{#1 \mathbin{/_{#2}} #3}
\newcommand\FunEnv{\Delta}
\newcommand\PatSubstFail{\Omega}
\newcommand\EvalCtx{\mathcal E}
\newcommand\Hole{\scalebox{0.8}{\ensuremath{\square}}}

\newcommand\ValEq{\sim}
\newcommand\SubstUnion\oplus
\newcommand\Plug[2]{#1[#2]}

\newcommand\EvalGuard[1]{\GenEvalGuard{\FunEnv}{#1}}
\newcommand\GenEvalGuard[2]{\mathcal{G}_{#1}(#2)}
\newcommand\True{\mathit{true}}
\newcommand\False{\mathit{false}}

\newcommand{\Accord}{\bowtie}

\newcommand\Venv{\Gamma}
\newcommand\VenvStar{\Venv^\star}
\newcommand\Senv{\Sigma}
\newcommand\Bothenv[2]{#1 \cdot #2}

\newcommand\Ok{\mathsf{ok}}
\newcommand\ExpSchemaTy[4]{\Bothenv{#1}{#2} \Turns #3 : #4}
\newcommand\ExpTy[3]{\ExpSchemaTy{\Senv}{#1}{#2}{#3}}

\newcommand\DefOk[2]{#1 \Turns #2 ~\Ok}
\newcommand\ProgTy[2]{\Turns #1 : #2}

\newcommand\Turns[1][!*NEVER USED ARGUMENT*!]{%
  \ifthenelse{\equal{#1}{!*NEVER USED ARGUMENT*!}}{\vdash}{\vdash_{\mathsf{#1}}}%
}
\newcommand\DirPatTy{\rho}
\newcommand\DirPatTyUp{\scalebox{0.6}{\ensuremath{\uparrow}}}
\newcommand\DirPatTyDown{\scalebox{0.6}{\ensuremath{\downarrow}}}
\newcommand\PgTy[2]{\GenPatTy{#1}{}{#2}}
\newcommand\PgUpperTy[1]{\PgTy{#1}{\DirPatTyUp}}
\newcommand\PgLowerTy[1]{\PgTy{#1}{\DirPatTyDown}}
\newcommand\PatTy[2]{\GenPatTy{#1}{#2}{\DirPatTy}}

\newcommand\GenPatTy[3]{\llbracket #1 \rrbracket_{{#2}}^{#3}}
\newcommand\PatUpperTy[2]{\GenPatTy{#1}{#2}{\DirPatTyUp}}
\newcommand\PatLowerTy[2]{\GenPatTy{#1}{#2}{\DirPatTyDown}}
\newcommand\SafeSym{\operatorname{\mathsf{safe}}}
\newcommand\SafeDown[2]{\SafeSym^{\DirPatTyDown}(#1, #2)}
\newcommand\SafeUp[1]{\SafeSym^{\DirPatTyUp}(#1)}

\newcommand\TyOfConstSym{\mathsf{ty}}
\newcommand\TyOfConst[1]{\TyOfConstSym(#1)}
\newcommand\TyOfGt[1]{\TyOfConstSym(#1)}
\newcommand\TyOf[1]{\TyOfConstSym(#1)}
\newcommand\PatEnv[2]{#1\negmedspace\sslash\negmedspace#2}
\newcommand\ProjL[1]{\pi_1(#1)}
\newcommand\ProjR[1]{\pi_2(#1)}
\newcommand\ProjI[1]{\pi_i(#1)}

\newcommand\InterEnv{\doublewedge}
\newcommand\BigInterEnv{\bigdoublewedge}

\newcommand\Inst[1]{\InstSym(#1)}
\newcommand\InstSym{\mathsf{inst}}
\newcommand\Env[1]{\mathsf{env}(#1)}

\newcommand\showfreshtyvars{}
\newcommand\IfShowFresh[1]{\ifnum\pdfstrcmp{\showfreshtyvars}{-fresh}=0 #1 \else \fi}
\newcommand\IfShowFreshElse[2]{\ifnum\pdfstrcmp{\showfreshtyvars}{-fresh}=0 #1 \else #2 \fi}
\newcommand\Constr{C}
\newcommand\ConstrAlt{D}
\newcommand\ProgConstr{P}
\newcommand\SiConstr{c}
\newcommand\SiConstrAux{d}
\newcommand\SiConstrAlt{d}
\newcommand\ConstrOr{\mathbin{+}}
\newcommand\ConstrAnd{\mathbin{\&}}
\newcommand\ConstrTrue{\mathtt{True}}
\newcommand\MedConstrAnd{\mathop{\scalebox{1.2}{\&}}}
\newcommand\IsSubtyConstr{\mathbin{\dot{\IsSubty}}}
\newcommand\SubtyConstr[2]{#1 \IsSubtyConstr #2}
\newcommand\DefConstr[2]{\kw{def}~#1~\kw{in}~#2}
\newcommand\LetConstr[3]{\kw{letrec}\,\MetaPair{#1}{#2}~\kw{in}~#3}
\newcommand\CaseConstr[2]{\kw{case}~#1~\kw{of}~#2}
\newcommand\InConstr[3]{#1~\kw{in}~#2~\kw{unless}~#3}
\newcommand\ConstrGen[3]{#1 : #2 \Rightarrow #3}
\newcommand\ConstrGenV[4]{#1 : #2 \mathbin{\Rightarrow_{\IfShowFresh{#4}}} #3}
\newcommand\DefEnvGen[2]{#1 \Rightarrow #2}
\newcommand\DefConstrGen[3]{#1 \Rightarrow \MetaPair{#2}{#3}}
\newcommand\DefConstrGenV[4]{#1 \mathbin{\Rightarrow_{\IfShowFresh{#4}}} \MetaPair{#2}{#3}}
\newcommand\ConstrRew[3]{\ConstrRewV{#1}{#2}{#3}{}}
\newcommand\ConstrRewV[4]{\ConstrRewFull{\Senv}{#1}{#2}{#3}{#4}}
\newcommand\ConstrRewFull[5]{\Bothenv{#1}{#2} \Turns #3 \mathbin{\leadsto_{\IfShowFresh{#5}}} #4}
\newcommand\ConstrRewProgV[3]{\Turns #1 \mathbin{\leadsto_{\IfShowFresh{#3}}} #2}
\newcommand\ConstrRewProg[2]{\ConstrRewProgV{#1}{#2}{}}
\newcommand\PatTyEnvConstr[4]{#1 \sslash #2 \Rightarrow \MetaPair{#3}{#4}}
\newcommand\PatTyEnvConstrV[5]{#1 \sslash #2 \mathbin{\Rightarrow_{\IfShowFresh{#5}}} \MetaPair{#3}{#4}}

\newcommand\TallyBase{\mathsf{tally}}
\newcommand\TallySym{\TallyBase^{\star}}
\newcommand\Tally[2][!*NEVER USED ARGUMENT*!]{%
  \ifthenelse{\equal{#1}{!*NEVER USED ARGUMENT*!}}{\TallySym(#2)}{\TallySym_{#1}(#2)}%
}

\newcommand\SubstSolves[2]{#1 \Vdash #2}

\newcommand\EmptyVarSet{\emptyset}
\newcommand\VarSet[1][A]{#1}

\newcommand\Model{\ensuremath{\mathcal M}}
\newcommand\Filtermap{\mathit{filtermap}}
\newcommand\ListTy[1]{\mathit{list}\,#1}
\newcommand\BoolTy{\mathit{bool}}

\renewcommand\Bot\bot

\newcommand\PowersetFin{\mathcal{P}_{\mathsf{fin}}}
\newcommand\DefinedAs{\triangleq}
\newcommand\RestrictDomain[2]{#1{\restriction_{#2}}}

\makeatletter
\newlength{\CaseIndent}
\newcommand{\listmargincase}{
  \ifnum\@listdepth=1
    \setlength{\leftmargin}{\CaseIndent}%
  \else
    \ifnum\@listdepth=2
      \setlength{\leftmargin}{2\CaseIndent}%
    \else
      \ifnum\@listdepth=3
        \setlength{\leftmargin}{3\CaseIndent}%
      \else
        \ifnum\@listdepth=4
        \setlength{\leftmargin}{4\CaseIndent}%
        \fi
      \fi
    \fi
  \fi}

\newsavebox{\CDName}
\newcommand\CaseFormat[1]{\textit{#1}}
\newenvironment{CaseDistinction}[2][no]{%
  \begin{CaseDistinctionExplicit}{#2}{\ifthenelse{\equal{#1}{repeat}}{ #2}{}}%
}{%
  \end{CaseDistinctionExplicit}%
}
\newenvironment{CaseDistinctionExplicit}[2]
  {\par\noindent\CaseFormat{Case distinction} #1.%
   \sbox{\CDName}{\CaseFormat{End case distinction}#2.}%
   \begin{itemize}[topsep=0pt]{}{\listmargincase\setlength\listparindent\parindent}%
  }
  {\end{itemize}\usebox{\CDName}%
  }
\newcommand{\CdCase}[1]{\item\CaseFormat{Case} #1:}

\newenvironment{igather}{\collect@body\@igather}{\global\@ignoretrue}

\newcommand\@igather[1]{{%
\let\old@label@in@display\label@in@display%
\renewcommand\label@in@display[1]{\stepcounter{equation}\Labelnote{##1}\tag{\theequation}\old@label@in@display{##1}}%
\begin{gather*}
#1
\end{gather*}
}}
\makeatother

\newenvironment{EnumThm}{%
  
  \renewcommand\theenumi{\roman{enumi}}
  \begin{enumerate}[topsep=0pt]
}{\end{enumerate}}

\newcommand\LabelQualifier{}
\newcommand\QualifyLabel[1]{#1::\LabelQualifier}
\newcommand\QLabel[1]{\label{\QualifyLabel{#1}}}
\newcommand\QRef[1]{\eqref{\QualifyLabel{#1}}}

\makeatletter
\newcommand{\TextLabel}[1]{\refstepcounter{equation}\label{\QualifyLabel{#1}}\Labelnote{\QualifyLabel{#1}}\textup{\tagform@{\theequation}}}
\makeatother

\renewenvironment{proof}[1][**empty**]{%
        \setcounter{equation}{0}\par\addvspace{6pt}\noindent%
        \ifthenelse{\equal{#1}{**empty**}}{\textbf{Proof.}}{\textbf{#1.}}%
        \hskip5.5pt%
        }{\par\addvspace{6pt}}

\newcommand\ReasonAbove[2]{\overset{#2}{#1}}

\makeatletter
\def\smallunderbrace#1{\mathop{\vtop{\m@th\ialign{##\crcr
   $\hfil\displaystyle{#1}\hfil$\crcr
   \noalign{\kern3\p@\nointerlineskip}%
   \tiny\upbracefill\crcr\noalign{\kern3\p@}}}}\limits}
\makeatother
\newcommand\BraceBelow[2]{\smallunderbrace{#1}_{#2}}

\newcommand\IH{I.H.\@}
\newcommand\Wlog{W.l.o.g.\@}
\newcommand\WlogLower{w.l.o.g.\@}

\newcommand\dstar{{\star\star}}
\newcommand\estar{e^{\star}}

\newcommand\AndText{~\textrm{and}~}

\newcommand\Labelnote[1]{\marginnote{\footnotesize \textcolor{StringColor}{\StrBefore{#1}{::}}}}
\renewcommand\Labelnote[1]{} 

\begin{document}

\title{Set-Theoretic Types for Erlang: Theory, Implementation, and Evaluation}

\author{Albert Schimpf}
\affiliation{%
  \institution{RPTU University Kaiserslautern-Landau}
  \country{Germany}
  \authoremail{albert.schimpf@rptu.de}
}

\author{Stefan Wehr}
\affiliation{%
  \institution{Offenburg University of Applied Sciences}
  \country{Germany}
  \authoremail{stefan.wehr@hs-offenburg.de}
}

\author{Annette Bieniusa}
\affiliation{%
  \institution{RPTU University Kaiserslautern-Landau}
  \country{Germany}
  \authoremail{bieniusa@rptu.de}
}

\begin{abstract}
  Erlang's dynamic typing discipline can lead to runtime errors that
  persist even after process restarts. Some of these runtime errors
  could be prevented through static type checking.
  While Erlang provides a type specification
  language, the compiler does not enforce these types, thereby limiting their
  role to documentation purposes. Type checking Erlang code is challenging
  due to language features such as dynamic type tests, subtyping,
  equi-recursive types, polymorphism, intersection types in
  signatures, and untagged union types.

  This work presents a set-theoretic type system for Erlang which
  captures the core features of Erlang's existing type language. The
  formal type system guarantees type soundness, and ensures that type checking
  remains decidable. Additionally, an implementation of a type checker
  is provided, supporting all features of the Erlang type language and
  most term-level language constructs. A case study with modules from
  Erlang's standard library, an external project, and the type checker itself demonstrates
  its effectiveness in verifying real-world Erlang code.
\end{abstract}

\maketitle

\section{Introduction}
\label{sec:intro}

Erlang is a functional language designed for highly concurrent, distributed systems that provide soft real-time guarantees and are resilient to failures. Since its introduction in the 1980s~\citep{conf/hopl/Armstrong07}, it has been widely adopted across industries, including telecommunications, messaging, finance, and gaming.

In production, handling errors safely is critical. Erlang's most frequent and promoted approach is to let faulty processes crash and restart. While effective for many failures, some errors --- such as type mismatches --- persist regardless of restarts. Static type systems prevent such issues before execution~\citep{TAPL}, but Erlang, being dynamically typed, detects them only at runtime.
While Erlang provides a type specification language for function signatures~\citep{erlang-typespec, 10.1145/1292520.1292523}, the compiler does not enforce these specifications, limiting them to (sometimes incorrect) documentation rather than verification.

Static type checking in Erlang is challenging due to features such as dynamic type tests, subtyping, equi-recursive types~\citep{conf/pldi/CraryHP99,conf/lics/AbadiF96}, polymorphism, intersection types in signatures, and untagged union types~\citep{Pierce1992}.
Several efforts explored the introduction of static type checking for Erlang~\citep{dialyzer,gradualizer,conf/icfp/MarlowW97,conf/erlang/ValliappanH18,conf/erlang/RajendrakumarB21,conf/coordination/MostrousV11,whatsapp-whatsappeqwalizer-nodate}.
Traditional subtyping, defined by syntactic deduction rules, often struggles to accommodate intersection and union types, as distributivity laws do not always hold and least upper/greatest lower bounds are not generally expressible~\citep{Pierce1992,conf/tacs/Damm94,journals/pacmpl/MuehlboeckT18}.

Semantic subtyping~\citep{journals/jacm/FrischCB08,conf/icfp/CastagnaX11,conf/popl/Castagna0XILP14,conf/popl/Castagna0XA15} provides an alternative approach by interpreting types as sets, thereby reducing subtyping to set inclusion. This approach naturally accommodates intersection, union, and negation types but requires a carefully designed type model and sophisticated subtyping algorithms.

This work presents a set-theoretic type system for Erlang, formalized as both a declarative and an algorithmic system. Unlike state-of-the-art Erlang type checkers~\citep{whatsapp-whatsappeqwalizer-nodate,svenningsson-josefsgradualizer-2024} and static analyzers~\citep{dialyzer}, our system is founded on a sound and expressive theory, supporting a wide range of language features, with non-regular polymorphic types being the only notable exception.

Our work builds upon and supersedes a prior publication~\citep{conf/ifl/SchimpfWB22} both theoretically and practically.
On the theoretical side, we formally establish type soundness and prove the decidability of type checking.
On the practical side, we implemented a prototype type checker, which handles a wide range of Erlang language features and demonstrates its usefulness with a real-world case study.

\paragraph{Contributions}
We present the following contributions:
\begin{itemize}
\item We formalize a type system for Erlang based on set-theoretic types (\Cref{sec:formal-type-system}).
\item We establish the soundness of the type system (\Cref{sec:soundness} and \Cref{sec:appendix-type-safety-declarative-system}) and prove the decidability of type checking within our formalism (\Cref{sec:algor-typing-rules} and \Cref{sec:appendix-algorithmic}).
\item We develop and implement the \ety{} type checker in Erlang. \ety{} supports all features of Erlang's
  type language and a comprehensive subset of the term-level language constructs to enable its application in real-world projects (\Cref{sec:implementation}).
\item We empirically evaluate \ety{} by applying it to modules from the Erlang standard library (\Cref{sec:evaluation-erlang-standard-library}), a third-party FOSS project and \ety{}'s own codebase (\Cref{sec:evaluation-project}).
\item With \ety{} we present the first fully independent implementation of a set-theoretic type system, distinct from and not derived from the CDuce implementation~\citep{cduce}.
\end{itemize}

\paragraph{On reproducibility}
Prior work on set-theoretic types has been conducted by a single
research group, with all existing implementations supporting
corecursive types built upon the CDuce framework~\citep{cduce}.  Our work is, to our best knowledge,
the first in decades to theoretically and practically validate this prior
research.  By independently developing a complete type checker that
integrates all aspects required for a real-world codebase, we also provide
a new foundation for further research in the field — making the whole
greater than the sum of its parts.

\paragraph{Roadmap} \Cref{sec:examples} introduces our type system using examples. \Cref{sec:formal-type-system} presents
the formal type system and a typing algorithm, together with the metatheoretical results. \Cref{sec:implementation} discusses
the implementation and a real-world case study. \Cref{sec:related-work} gives insights into related work
and \Cref{sec:conclusion} presents a conclusion. All proofs can be found in \Cref{sec:appendix-type-safety-declarative-system} and
\Cref{sec:appendix-algorithmic}, and \Cref{sec:empirical-patterns,sec:empirical-dynamic-functions} present the results of an empirical evaluation
on the use of type tests and value-dependent types in real-world Erlang code.

\section{A Primer on Erlang and Set-Theoretic Types}
\label{sec:examples}

To introduce the principles of set-theoretic types, we start by
discussing selected aspects of the Erlang (type) language that are
difficult to address in rule-based subtyping systems, but are a good fit for
set-theoretic types and semantic subtyping.  
The discussion builds upon and extends examples presented in a previous article \citep{conf/ifl/SchimpfWB22}, providing a foundation for understanding the benefits and applications of set-theoretic types in the context of Erlang.

\subsection{Singleton, Union, and Intersection types}
\label{sec:singl-union-inters}

As a first example, consider the function \lstinline|last_day_of_the_month/2| from the Erlang standard library that determines the number of days in a month.\footnote{\url{https://www.erlang.org/doc/man/calendar.html\#type-ldom}}

\begin{lstlisting}[language=Erlang,numbers=left]
last_day_of_the_month(_, 4) -> 30;
last_day_of_the_month(_, 6) -> 30;
last_day_of_the_month(_, 9) -> 30;
last_day_of_the_month(_,11) -> 30;
last_day_of_the_month(Y, 2) ->
   case is_leap_year(Y) of
      true -> 29;
      _    -> 28
   end;
last_day_of_the_month(_, M) -> 31.
\end{lstlisting}
   


In Erlang, a function consists of at least one function clause.
Each clause consists of the function name and the
arguments enclosed in parentheses, followed by the
arrow symbol~\lstinline|->| and the function
body. Clauses are separated by semicolons, and the last clause is terminated
with a period. Each function argument is a pattern; in the example above, the
patterns are constants (\lstinline{4}, \lstinline{6}, \ldots),
wildcards (\lstinline{_}), and variables (\lstinline{Y}). Constants match
only their exact value, whereas wildcards and variables match everything.

Type specifications are optional in Erlang.  They can be added to any
function with the \lstinline|-spec| annotation, using
the Erlang type specification language~\citep{erlang-typespec}.
The Erlang standard library provides the following type specification for the function:\footnote{The standard library actually introduces some type synonyms that we directly inline here.}

\begin{lstlisting}[language=Erlang,numbers=none]
-spec last_day_of_the_month(non_neg_integer(), 1..12) -> 28 | 29 | 30 | 31
\end{lstlisting}
Here, the type \lstinline{1..12} denotes the range of integers between 1 and 12.

Though the type specification is correct,
we can provide a more precise type specification for the function.
To refine a type for \lstinline|last_day_of_month|, we notice that 
the first parameter (the year) can be any integer greater or equal to zero, while
the second parameter (the month) is an integer between 1 and 12 such that
values from the set
$\{4,6,9,11\}$ are mapped to $30$, value 2 is mapped to either 29 or 28, and values in
$\{1, \dots, 12\} \setminus \{4,6,9,11,2\} = \{1,3,5,7,8,10,12\}$ are
mapped to 31.

The most precise type for the above function
--- without allowing arithmetic operations on type level --- is shown in \Cref{fig:spec-last-day}.

\begin{figure}[t]
\begin{lstlisting}[language=Erlang,numbers=none]
-spec last_day_of_the_month
    (non_neg_integer(), 2) -> 28 | 29
  ; (non_neg_integer(), 4 | 6 | 9 | 11) -> 30
  ; (non_neg_integer(), 1 | 3 | 5 | 7 | 8 | 10 | 12) -> 31.
\end{lstlisting}
\caption{Precise type specification for \lstinline|last_day_of_the_month|.}
\label{fig:spec-last-day}
\end{figure}

This type specification assigns three function types to \lstinline{last_day_of_month},
meaning the function is actually a function of three different types.
The symbol \lstinline{;} denotes \emph{intersection}
in Erlang's type language. Other features of the type language used in the example are
\emph{singleton types} such as \lstinline{2},
the type \lstinline{non_neg_integer()} representing the set of numbers
$\{0, 1, 2, \ldots\}$, the \emph{union} $t_1\,\mathtt{|}\,t_2$ of two types,
and function types $\mathtt{(}t_1,\ldots,t_n\texttt{)\,->\,}t$.
Erlang allows intersection types only at the top-level of a type specification, but not
nested inside a type. Union binds tighter than the function arrow, which in
turns binds tighter than intersection.
All these type operators are right associative.

Our static type checker \ety{} handles both type specifications
and can successfully check the function with either of
the two specifications. Hence, programmers are able to trade
readability with preciseness, without compromising type safety.

\subsection{Type Tests, Occurrence Typing, and Exhaustiveness}
\label{sec:type-tests-exha}

A common idiom of Erlang programs is the use of dynamic type tests to determine execution paths based on the runtime type of a value.  
Similar to the function clause guards used in the function \lstinline|last_day_of_month| above, we can employ guards in \lstinline|case| clauses to execute different branches depending on the type of some expression.  
Guards are a restricted subset of valid Erlang expressions, and their evaluation must be side-effect-free. 
To accurately type check such code, the type system must be able to refine the type of the scrutinee in different branches via pattern matching and type tests~(guards).

\Cref{fig:inter-function} shows an example illustrating a function with a type specification that is an intersection of two distinct function clauses.

\begin{figure}[t]
\begin{lstlisting}[language=Erlang,numbers=none]
-spec inter(integer()) -> integer();
  ({atom(), atom() | integer()} | {integer(), atom()}) -> atom().
inter(X) ->
  case X of
    _ when is_integer(X) -> X + 1;
    {Y, Z} when is_atom(Y) -> Y;
    {_, Z} -> Z
  end.
\end{lstlisting}
\caption{The \lstinline|inter| function illustrating occurrence typing with guards and patterns.}
\label{fig:inter-function}
\end{figure}

The specification establishes a disjoint input-output relationship: 
integer inputs yield integer outputs, 
while inputs comprising either form of the specified tuple variants yield atom outputs.

Occurrence typing is used to verify that each branch of the \lstinline|case| expression adheres to the type signature. 
In the first branch, the guard \lstinline|is_integer(X)| refines the type of \lstinline|X| from the full input union down to \lstinline|integer()|, making the operation \lstinline|X + 1| type safe. 
In the second branch, the pattern \lstinline|{Y, Z}| combined with the guard \lstinline|is_atom(Y)| selects the part of the input type where \lstinline|X| is a tuple whose first element is an atom, i.e., \lstinline|{atom(), atom() |$\mid$\lstinline| integer()}|.  
Consequently, the returned value \lstinline|Y| is known to be of type \lstinline|atom()|.

Finally, the third branch is only reached for tuples not captured by the previous clause. 
At this point, \lstinline|X|'s type is refined to \lstinline|{integer(), atom()}|, which is the remaining part of the input type. Therefore, the variable \lstinline|Z| is known to have type \lstinline|atom()|, making the return type correct. 
Furthermore, the type system can verify that the \lstinline|case| expression is \emph{exhaustive}, as the three branches collectively cover all possible types specified for the input \lstinline|X|.

In the \ety{} codebase, we leverage these capabilities to process the Erlang abstract syntax tree, which involves large case distinctions with up to 25 clauses. 
We can guarantee that all possible type representations are handled without relying on a catch-all clause. 
To our knowledge, no other system currently supports exhaustiveness checking for Erlang while also integrating guards with occurrence typing, type reconstruction, and the full range of set-theoretic operations (union, intersection, and negation).

\subsection{Recursive Data Types \& Polymorphism}
\label{sec:recursive-data-types}

Erlang supports custom type definitions and polymorphic functions through type variables. 
Consider the standard library's \lstinline|flatten| function from the \lstinline|lists| module:

\begin{lstlisting}[language=Erlang,numbers=none]
flatten([H|T], Tail) when is_list(H) -> flatten(H, flatten(T, Tail));
flatten([H|T], Tail) -> [H|flatten(T, Tail)];
flatten([], Tail) -> Tail.
\end{lstlisting}

The Erlang function \lstinline|flatten| recursively flattens nested lists, combining all elements into a single flat list and appending a given \lstinline|Tail|.
It is supplied with the following type specification:

\begin{lstlisting}[language=Erlang,numbers=none]
-spec flatten(DeepList, Tail) -> [term()] when
      DeepList :: [term() | DeepList], Tail :: [term()].
\end{lstlisting}

The \lstinline|when| clause introduces type constraints, restricting the allowed types for the type variables used in the signature. 
In this case, it specifies that \lstinline|DeepList| is a deep list, a recursive type of a list where each list element is either a term (\lstinline|term()|) or a deep list (\lstinline|DeepList|). 
Note that the vertical bar \lstinline{|} in type contexts denotes type union, distinct from the list constructor \lstinline{[H | T]} used in expression contexts.
The \lstinline|Tail| is constrained to be a regular list of terms. 
The recursive definition \lstinline|[term() |$\mid$ \lstinline|DeepList]| in the constraint attempts to describe a structure that may be nested to an arbitrary depth.

This specification has three significant limitations. 
First, even though the first argument input type is specified recursively, the type simplifies to \lstinline{[term() | DeepList] = [term()]}.
Second, the type loses the input-output relationship. 
The output type \lstinline|[term()]| is overly broad and doesn't reflect that the result contains element types which are related in some way to the input types. 
This is a problem when type checking an application of the function.
This can be addressed with polymorphism using type variables, which allow expressing relationships between input and output types. 
To capture this relation, we can define a custom recursive type:

\begin{lstlisting}[language=Erlang,numbers=none]
-type deepList(A) :: [A | deepList(A)].
-spec flatten(deepList(A), [A]) -> [A].
\end{lstlisting}

Type definitions introduce named aliases for type expressions, improving readability and enabling reuse. 
Third, this specification still fails to capture the semantic constraint that elements themselves are not lists in the flattened result. 
Consider the invocation \lstinline|do_flatten(H, do_flatten(T, Tail))| in the first function clause. 
Here, the list constructor \lstinline{[H|T]} decomposes a list into head and tail, while the guard \lstinline|is_list(H)| checks if the head is itself a list.
The type of \lstinline|H| is refined to be a list, but \lstinline|do_flatten| expects a \lstinline|deepList| as its first argument.
The implementation fails to type check with this generic type specification.

To type check \lstinline|flatten|, we need to make the specification more precise.
Erlang's type language cannot express set differences, yet.
We therefore need to use a custom set-theoretic type negation to express the type, as shown in \Cref{fig:deep-list-type}.

\begin{figure}[t]
\begin{lstlisting}[language=Erlang,numbers=none]
-type deep_list(A) :: [etylizer:without(A, list()) | deep_list(A)].
\end{lstlisting}
\caption{Precise recursive type for \lstinline|flatten| using set-theoretic negation.}
\label{fig:deep-list-type}
\end{figure}

With this refined type, the function type checks correctly.
Notably, all previous type system developments for Erlang have failed to successfully verify Erlang's \lstinline|flatten/2| function.

As an example combining recursive, polymorphic, and intersection types, consider the \lstinline|filtermap/2| function from Erlang's standard library, which has the following specification:

\begin{lstlisting}
-spec filtermap( fun((Elem) -> boolean() | {'true', Value}), [Elem] ) ->
   [Elem | Value] when Elem :: term(), Value :: term().
filtermap(F, [Hd|Tail]) ->
    case F(Hd) of
        true -> [Hd|filtermap(F, Tail)];
        {true,Val} -> [Val|filtermap(F, Tail)];
        false -> filtermap(F, Tail)
    end;
filtermap(F, []) when is_function(F, 1) -> [].
\end{lstlisting}

As the name and signature suggest, \lstinline|filtermap/2| either filters elements based on a Boolean function, or filters and applies an additional transformation on the filtered values (mapping from \lstinline|T| to  \lstinline|U|, or yields a list mixing transformed and unfiltered values).
Our type checker can efficiently decide the type safety of this higher-order, polymorphic, and recursive function. 
We argue that its implementation could be expressed with a more precise type:

\begin{lstlisting}
-spec filtermap(fun((T) -> boolean()), [T]) -> [T]
             ; (fun((T) -> {true, U} | false), [T]) -> [U]
             ; (fun((T) -> {true, U} | boolean()), [T]) -> [T | U].
\end{lstlisting}

Our system successfully also verifies this type specification.

Certain classes of advanced polymorphic data types are not supported, however.
Consider the following example of a \emph{nested datatype} \citep{conf/mpc/BirdM98, journals/jfp/Hinze00}:

\begin{lstlisting}[language=Erlang,numbers=none]
-type perfect(A) :: A | perfect({A, A}).
\end{lstlisting}

Essentially, \lstinline|perfect(A)| describes complete binary trees where the leaf nodes are of type \lstinline|A|.
The interesting aspect of the example is that on the right side of the equation, \lstinline|perfect({A,A})| is differently parametrized from the left side of the definition.
Our type system cannot handle such polymorphic recursion on types.

Additionally, our system does not support rank-2 polymorphism, which would be required to type functions that take polymorphic functions as arguments. 
As discussed in \Cref{sec:evaluation-project}, this limitation affects certain higher-order programming patterns, though they occur infrequently in practice.

\subsection{Type Reconstruction for Local Functions}
\label{sec:type-reconstr-local}

In Erlang, one-time-use helper functions are typically defined as anonymous functions. 
However, since Erlang does not support type annotations for anonymous functions, these functions must be refactored into separate named functions to enable type checking, which can make their type checking more complicated.
A notable example are computations on the internal representation of types for multi-arity functions in the implementation of \ety{}; \Cref{fig:intersect-function} shows the intersection operation for set-theoretic types represented in DNF.

\begin{figure}[t]
\begin{lstlisting}[language=Erlang,numbers=none]
-spec dnf_intersect(dnf(), dnf()) -> dnf().
-spec intersect(
	{dnf(), #{integer() => dnf()}},
	{dnf(), #{integer() => dnf()}}
	) -> {dnf(), #{integer() => dnf()}}.
intersect({DefaultT1, T1}, {DefaultT2, T2}) ->
  AllKeys = maps:keys(T1) ++ maps:keys(T2),
  IntersectKey = fun(Arity) ->
    dnf_intersect(
      maps:get(Arity, T1, DefaultT1),
      maps:get(Arity, T2, DefaultT2)
      ) end,
  L = lists:map(fun(Key) -> {Key, IntersectKey(Key)} end, AllKeys),
  {
    dnf_intersect(DefaultT1, DefaultT2),
    maps:from_list(L)
  }.
\end{lstlisting}
\caption{Intersection of multi-arity function types in DNF representation from the \ety{} codebase, with a reconstructed type for the anonymous function \lstinline|IntersectKey|.}
\label{fig:intersect-function}
\end{figure}

The representation consists of a tuple containing a default representation for all unspecified arities, along with an explicit mapping from 
arity to function representation.
The function \lstinline|dnf_intersect| implements the intersection for two function spaces, correctly handling both explicitly defined arities and the default case.
Within this function, \lstinline|IntersectKey| is an anonymous function without a type annotation, used to compute the intersection for a given arity in multi-arity functions.
\ety{} successfully reconstructs its type and verifies the function for type safety.

%

\subsection{Value-Dependent Types}

The crux of the difficulty of typing Erlang is that types sometimes depend on values. 
The standard library specifications highlight this challenge:

\begin{lstlisting}
-spec element(N :: pos_integer(), tuple()) -> term().
-spec ukeysort(N :: pos_integer(), [Tuple]) -> [Tuple] when Tuple :: tuple().
\end{lstlisting}

Both functions access tuple elements dynamically, \lstinline|element/2| for a single tuple and \lstinline|ukeysort/2| for sorting a list of tuples based on a dynamic index. 
Consider their usage in \lstinline|from_list/1| from Erlang's standard library and \lstinline|loc_exp/1| from \ety{}'s codebase:

\begin{lstlisting}[language=Erlang,numbers=none]
-spec from_list([{Key, Value}]) -> orddict(Key, Value).
from_list([]) -> [];
from_list([{_,_}] = Pair) -> Pair;
from_list(Pairs) -> lists:ukeysort(1, reverse_pairs(Pairs, [])).

-spec loc_exp(exp()) -> loc().
loc_exp(X) -> element(2, X).
\end{lstlisting}

The function \lstinline|from_list/1| relies on \lstinline|lists:ukeysort/2| to sort the tuple list for a specific tuple position, while \lstinline|loc_exp/1| extracts the second element from a tuple of unknown size. 
Both require knowing specific indices to determine precise types.

In our implementation, we handle these cases through \emph{type overlays} that provide specialized specifications for common patterns:

\begin{lstlisting}[language=Erlang,numbers=none]
% Overlay for lists:ukeysort with constant index
-spec 'lists:ukeysort'(1, [{K, V}]) -> [{K, V}].

% Overlay for element with constant indices  
-spec 'erlang:element'(2, {_A, B}) -> B;
  (2, {_A, B, _C}) -> B;
  (2, {_A, B, _C, _D}) -> B;
  ...;
  (_, tuple()) -> term().
\end{lstlisting}

These overlays effectively express pattern matching at the type level, verifying that inputs have sufficient size and correct structure for the given indices. 
Our empirical study~(\Cref{sec:empirical-dynamic-functions}) shows that despite the dynamic nature of value-dependent operations in Erlang, most uses (76\%) employ constant parameters, making them amenable to static typing with our system's features.
While general value-dependent types would be needed for any possible application of \lstinline|element/2| and \lstinline|ukeysort/2|, our combination of singleton and intersection types adequately addresses the common cases found in practice.

~\newline

The examples illustrate that set-theoretic types align naturally with
Erlang's type language.  Our experience suggests that while some
existing type signatures may require refinement, all Erlang language
constructs can be accommodated to accurately reflect the language's
semantics.  In particular, type checks and guarded expressions
commonly used in Erlang for function clauses and \lstinline|case|
expressions can be statically verified, preventing common mistakes
that typically lead to semantic errors.

\section{Formal Type System}
\label{sec:formal-type-system}

This section presents a formal foundation for our
type system for Erlang. It defines the language \MinErl{},
a core calculus designed to capture the
essential ingredients of sequential programming
in Erlang. It extends the $\lambda$-calculus with
constants, pairs, pattern matching with type tests, and top-level,
mutually recursive functions.
The type system of
\MinErl{} is based on subtyping, polymorphic set-theoretic
types~\citep{journals/jacm/FrischCB08,conf/icfp/CastagnaX11}, and
occurrence typing~\citep{conf/popl/Tobin-HochstadtF08}.
Its formulation is influenced
by the work of Castagna and colleagues on using set-theoretic types
to model polymorphic variants~\citep{conf/icfp/CastagnaP016}.

\subsection{Syntax}

\boxfigSingle{f:syntax}{Syntax of \MinErl{}}{
	\small
	\[
	\begin{array}{c}
          \begin{array}{r@{\quad}l@{\qquad}r@{\quad}l}
	    \textrm{Type variables}       & \Tyvar, \TyvarAux \in \AlltyVar&
            \textrm{Expression variables} &  x, y \in \AllVar \\
	    \textrm{Finite sets of type variables} & \TyvarSet \in \PowersetFin(\AlltyVar) &
	    \textrm{Constants (ints, floats)} & \Const \in \AllConst \supseteq \AllConstty\\
	    \textrm{Singleton types (ints)} & \Constty \in \AllConstty \\
          \end{array}\\
	\end{array}
	\]

	\[
          \begin{array}[t]{@{}l@{~}r}
            \begin{array}[t]{r@{~~}r@{~}c@{~}l@{~}l}
		\textrm{Type schemes} &
		\Polyty & ::= & \TyScm{\TyvarSet} \Monoty
		\\
		\textrm{Mono types} &
		t, u  & ::= &
		\Monoty \Union \Monoty
                \mid \Neg\Monoty  \mid
		\Monoty \to \Monoty\\
                \textrm{(coinductive)}&&&
		\mid \Pairty{\Monoty}{\Monoty}
		\mid \Tyvar
		\mid \Basety\\
		\textrm{Base types} & \Basety & ::= & \Constty \mid \Int \mid \Float\\[\smallskipamount]
              \textrm{Abbreviations} &
              \Monoty_1 \Inter \Monoty_1 & \DefinedAs & \Neg (\Neg\Monoty_1 \Union \Neg\Monoty_2)\\
              & \Monoty_1 \WithoutTy \Monoty_2 & \DefinedAs & \Monoty_1 \Inter \Neg \Monoty_2 \\
              & \ErlTop &\DefinedAs&  \Int \Union \Float \\
                                     &&& {}\Union (\Pairty{\ErlTop}{\ErlTop})\\
                                     &&& {}\Union (\ErlBot \to \ErlTop) \\
              & \ErlBot & \DefinedAs &  \Neg \ErlTop
            \end{array}
            &
            \begin{array}[t]{r@{~~}r@{~}r@{~}l@{~}l}
		\textrm{Expressions}
		& e & ::= & x
		\mid \Const
		\mid \Abs{x} e
		\mid \App{e}{e}
                \mid \Pair{e}{e} \\
                &&& \mid \Case{e}{\Multi{\Cls}} \\
                \textrm{Values}
                & v  & ::= & \Const \mid \Abs{x} e \mid \Pair{v}{v} \\
		\textrm{Pattern clauses}
		& \Cls & ::= & \PatCls{\Pg}{e}                \\
		\textrm{Guarded patterns}
		& \Pg & ::= & \WithGuard{p}{g}                \\
		\textrm{Patterns}
		& p & ::= & v \mid \Wildcard \mid x \mid \Capt{x} \mid \Pair{p}{p}                 \\
		\textrm{Guards}
		& g & ::= & \Is{\GuardTy}{x} \mid \Is{\GuardTy}{v} \mid \GuardOracle \\
                &&& \mid \GuardAnd{g}{g} \mid \GuardTrue \\
		\textrm{Guard types}
		& \GuardTy & ::= & \Int \mid \Float \mid \AnyPair \mid \AnyFun \\
		\textrm{Definitions}
		& \DefSym  & ::= & \Def{x}{\Polyty}{\Abs{y}{e}} \mid \DefNt{x}{\Abs{y}{e}} \\
		& \FunEnv & ::= & \Multi{\DefSym} \\
		\textrm{Programs}
		& \Prog & ::= &
		\Letrec{\FunEnv}{e} \\
            \end{array}
          \end{array}
	\]
}

\Cref{f:syntax} defines the syntax of \MinErl{}.
We use $\Tyvar, \TyvarAux \in \AlltyVar$ for type variables;
$x, y \in \AllVar$ for expression variables;
$\Const \in \AllConst$ for constants;
and $\Constty \in \AllConstty$ for singleton types.
We assume that the sets $\AlltyVar, \AllVar, \AllConst$, and $\AllConstty$ are countably infinite.
The set of constants $\AllConst$ contains ints and floats, whereas the set of singleton types
$\AllConstty \subseteq \AllConst$ contains only ints.
Metavariable $\TyvarSet$ ranges over finite sets of type variables.

The notation $\mathfrak{s}_{i \in I}$ for some syntactic
construct $\mathfrak{s}$ and some index set
$I = \{1, \ldots, n\}$ is short for the
sequence $\mathfrak{s}_1 \dots \mathfrak{s}_n$.

\subsubsection{Types}

A type scheme $\Polyty = \forall A . t$ quantifies a monomorphic type $\Monoty$ over
a finite set of type variables $A$.
At some places, we identify the type scheme $\TyScm{\emptyset}{\Monoty}$
with the monomorphic type $\Monoty$.

Monomorphic types $\Monoty$ encompass set-theoretic connectives union
$t_1 \Union t_2$ and negation $\Neg t$ (to be explained shortly),
as well as type constructors for functions $t_1 \to t_2$,
pairs $\Pairty{t_1}{t_2}$, type variables $\alpha$, and base types $b$.
Base types are singleton types $\Constty$ for individual ints, and types for
the whole set of ints and floats.

Set-theoretic types require somewhat more explanation.
We only sketch the main ideas here and refer to the literature
for
details~\citep{journals/jacm/FrischCB08,conf/icfp/CastagnaX11,conf/popl/Castagna0XILP14,conf/popl/Castagna0XA15,conf/icfp/CastagnaP016}.
With set-theoretic types, we interpret types as subsets of some model
\Model{}. Subtyping between types then becomes set inclusion.
Intuitively, \Model{} could be the set of all well-typed values of the
programming language under investigation; a type $t$ then should
denote the set of all values of this type. For technical reasons, such
a definition of \Model{} does not work well in practice. Fortunately,
there are alternative definitions that allow us to recover
the intuitive model~\citep{conf/ppdp/CastagnaF05,journals/jacm/FrischCB08}.

In general, the definition of a model for set-theoretic types is
tied to a specific programming language, in particular to its mode of
evaluation (strict or not) and to the available type constructors.
However, Erlang is close enough to the
languages considered in the aforementioned literature, so the model
and algorithms developed in the work of \cite{conf/popl/Castagna0XILP14,conf/popl/Castagna0XA15,conf/icfp/CastagnaP016}
directly transfer to our setting.

Coming back to the syntax of monomorphic types,
$\Monoty_1 \Union \Monoty_2$ denotes the union of the sets associated with $\Monoty_1$
and $\Monoty_2$, whereas $\Neg \Monoty$ denotes the complement of $\Monoty$ with respect to
the full model.
Intersection $\Monoty_1 \Inter \Monoty_2$, set difference $t_1 \WithoutTy t_2$, as well
as the top type $\ErlTop$ and bottom type $\ErlBot$
are defined as abbreviations, see \Cref{f:syntax}.
The top type $\ErlTop$ represents the full model; that is,
the union of $\Int$, $\Float$, the top type for pairs
$\Pairty{\ErlTop}{\ErlTop}$, and the top type for functions $\ErlBot \to \ErlTop$.
The bottom type $\ErlBot$, representing the empty set, is defined as the complement
of $\ErlTop$.
We employ the convention that the type connectives $\Union, \Inter, \Neg$ and
the type constructors $\to, \Pairty{}{}$ associate to the right,
while the (ascending) order of precedence is $\Union, \Inter, \Neg, \to, \Pairty{}{}$.

To allow for recursive types, the definition of monomorphic types $t$
is to be read coinductively. Hence, a monomorphic type is a potentially
infinite tree.
For example, the type for lists can be defined as the type
fulfilling the equation
$\ListTy{\Tyvar} = 0 \Union (\Pairty{\Tyvar}{\ListTy{\Tyvar}})$. 
Here, the singleton $0$ denotes the empty list. 
The notation \enquote{$\ListTy{\Tyvar} =$} is not part
of the type syntax, but allows us to write the infinite tree
$0 \Union \Pairty{\Tyvar}{(0 \Union \Pairty{\Tyvar}{\ldots})}$
with a finite description.

We require the potentially infinite trees
to be \emph{regular} and \emph{contractive}~\citep{conf/icfp/CastagnaP016, TAPL}.
Regularity requires that a tree contains only finitely many distinct subtrees.
This condition is crucial for establishing decidability of the subtyping algorithm.
Contractiveness states that every infinite branch has infinitely many occurrences
of the type constructors $\to$ and $\Pairty{}{}$. This condition rules out
nonsensical types fulfilling equations such as $t = t \Union t$ or $t = \Neg t$.

Erlang's type syntax supports intersection types only in limited fashion; that is,
only as the outermost connective of a type specification for a top-level function. Further,
Erlang provides no syntax for negation types.
However, our formalization supports full set-theoretic types because
several typing rules (see \ref*{sec:static-semantics}) require their expressiveness.

\subsubsection{Expressions, definitions, and programs}

Expressions $e$ include variables $x$, constants $\Const$ for ints and floats,
lambda abstractions $\lambda x.e$, function application $e_1\,e_2$, pairs $\Pair{e_1}{e_2}$,
as well as \kw{case}-expressions for pattern matching.
By convention,
function application associates to the left and
the body of a lambda abstraction extends as far to the right as possible.
Values $v$ are constants, lambda abstractions, and pairs of values.

The expression $\Case{e}{\Multi{\Cls}}$ matches scrutiny $e$
against pattern clauses $\Cls_i$. A pattern clause has the form
$\PatCls{\WithGuard{p}{g}}{e'}$, where $p$ is a pattern, $g$ a guard, and
$e'$ the body expression.
Patterns consist of values $v$, wildcards $\Wildcard$,
binding variables $x$, bound variables $\Capt{x}$, and pair patterns $\Pair{p}{p}$.
Values and variables in patterns require some further explanation.

In Erlang, pattern variables
have a somewhat unusual scoping rule. If the variable is already
bound in the outer scope, it acts like a constant pattern and only matches values equal
to the variable's value.
If the variable is not bound in the outer scope, it always matches and
binds the value being matched.

In our calculus, we handle these two variants with explicit constructs.
A bound pattern variable
$\Capt{x}$ requires $x$ to be in scope; matching against $\Capt{x}$ matches
against the value $x$ is bound to.
A binding pattern variable $x$ always matches and binds the value to $x$, even
if $x$ is already defined in the outer scope. The scope of (the inner) $x$ then comprises the guard
and the body of the corresponding pattern clause.

The semantics intended by bound pattern variables motivates the inclusion
of values (and not only constants) into the pattern syntax. During
evaluation, a bound pattern variable $\Capt x$ should be substituted
by some value $v$.

A guard $g$ further constrains a match. It is essentially a conjunction of $\GuardTrue$,
$\GuardOracle$, and type tests against
a guard type $\GuardTy$. Type tests, written $\Is{\GuardTy}{x}$ or $\Is{\GuardTy}{v}$,
correspond to guard functions in Erlang such as
\lstinline|is_integer|.\footnote{%
  We support only variables and values in guard tests because an empirical evaluation of
  real-world Erlang code showed that nearly all types tests are on variables (\Cref{sec:empirical-patterns}).
  Other sorts of type tests need to be modeled as \GuardOracle.
  }
Similar to Erlang, guard types are limited
to $\Int$, $\Float$, $\AnyPair$, and $\AnyFun$.
Guards also include $\GuardOracle$ for modeling guard expressions that cannot
be analyzed statically (e.g.\ \verb|tuple_size| in Erlang).
At some places, we identify
a pattern $p$ with the guarded pattern $\WithGuard{p}{\GuardTrue}$.

A top-level definition $\DefSym$ binds a variable to a function, optionally
equipped with a type annotation $\Polyty$. A sequence of definitions is denoted by $\FunEnv$,
assuming that $\FunEnv$ contains no duplicate definitions for the same variable.
We write $\Dom{\FunEnv}$ for the set of variables defined in $\FunEnv$, and $\FunEnv(x)$
to access the right-hand side of $x$'s definition in $\FunEnv$, assuming $x \in \Dom{\FunEnv}$.
A program $\Prog$ is then a sequence of
mutually recursive function definitions $\FunEnv$ and a main expression $e$.

\subsubsection{Miscellaneous}

We write $\FreeTyVars{t}$ to denote the set of type variables
free in type $t$. The coinductive definition of $\FreeTyVars{t}$ can be
found in \Cref{sec:appendix-type-safety-declarative-system}, \Cref{def:free-tyvars}.
We write $\Free{e}$ and $\Free{\Pg}$ to denote the set of expression variables free in
expression $e$ and guarded pattern $\Pg$, respectively.
These definitions are also contained in
\Cref{sec:appendix-type-safety-declarative-system}, \Cref{def:free-exp}.
We consider all syntactic constructs equal up to renaming of bound expression
and type variables.

\subsection{Dynamic Semantics}

\boxfigSingle{f:dynamic-semantics}{Dynamic semantics of \MinErl}{
  \[
  \begin{array}{rr@{~}r@{~}ll}
    \textrm{Evaluation contexts}
       & \EvalCtx & ::= & \Hole \mid \App{\EvalCtx}e \mid \App{v}{\EvalCtx} \mid
                          \Pair{\EvalCtx}{e} \mid \Pair{v}{\EvalCtx}
                          \mid \Case{\EvalCtx}{\Cls_{i \in I}} \\
    \textrm{Value substitutions} & \Valsubst & ::= & [v_i/x_i]_{i \in I}
  \end{array}
  \]
  \RuleSection{
    \RuleForm{
      \ReduceBase{\FunEnv}{e}{e} \qquad
      \Reduce{\FunEnv}{e}{e}
    }
  }{
    Reductions
  }
  \begin{mathpar}
    \inferrule[red-abs]{}{
      \ReduceBase{\FunEnv}{(\App{\Abs{x}{e})}{v}}{e[v/x]}
    }

    \inferrule[red-var]{
      \FunEnv(x) = v
    }{
      \ReduceBase{\FunEnv}{x}{v}
    }

    \inferrule[red-case]{
      \PatSubst{v}{\Pg_j} = \Valsubst \\
      \Forall{i < j}~\PatSubst{v}{\Pg_i} = \PatSubstFail
    }{
      \ReduceBase{\FunEnv}{\Case{v}{\Multi{(\Pg_i \to e_i)}}}{e_j\Valsubst}
    }

    \inferrule[red-context]{
      \ReduceBase{\FunEnv}{e_1}{e_2}
    }{
      \Reduce{\FunEnv}{\Plug{\EvalCtx}{e_1}}{\Plug{\EvalCtx}{e_2}}
    }
  \end{mathpar}

\RuleSection{
    \RuleForm{
      \Reduce{}{\Prog}{\Prog}
    }
  }{
    Program reductions
  }
  \begin{mathpar}
    \inferrule[red-prog]{
      \Reduce{\FunEnv}{e}{e'}
    }{
      \Reduce{}{\Letrec{\FunEnv}{e}}{\Letrec{\FunEnv}{e'}}
    }
  \end{mathpar}

}

The formalization of the dynamic semantics spreads out over
two figures.
\Cref{f:dynamic-semantics} defines the reduction relation and
\Cref{f:dynamic-patterns} defines auxiliaries for evaluating
pattern matching.

\Cref{f:dynamic-semantics} defines a standard, call-by-value, small-step reduction relation
for \MinErl{}.
A call-by-value evaluation context $\EvalCtx$ is an expression with a hole $\Hole$ such that
the hole marks the point where the next evaluation step should happen.
We write $\EvalCtx[e]$ to denote the replacement of the hole in $\EvalCtx$ with expression $e$.

Metavariable $\Valsubst = [v_i/x_i]_{i \in I}$ denotes the capture-avoiding
substitution of expression variables $x_i$ with value $v_i$.
We write $e\Valsubst$ for applying substitution $\Valsubst$ to expression $e$
(similar for other syntactic constructs),
$\Dom{\Valsubst}$ for
the domain of $\Valsubst$, and
$\Valsubst(x)$ for accessing the value for variable $x$, implicitly
assuming $x \in \Dom{\Valsubst}$.

Relation $\ReduceBase{\FunEnv}{e}{e'}$ is the basic evaluation relation that
performs a single reduction step from $e$ to $e'$ at the top-level of $e$.
Rule \Rule{red-abs} reduces a lambda expression in the usual way. Rule \Rule{red-var}
requires
$\FunEnv$ to access $\LetrecSym$-bound variables because
substitutions for letrec-bound variables
are performed lazily to allow for mutually recursive functions.
Rule \Rule{red-case} reduces a case expression to the body of the first branch that
matches its guarded pattern.
It uses $\PatSubst{v}{\Pg}$ for matching value $v$ against guarded pattern $\Pg$, yielding
a substitution $\Valsubst$ on success or $\PatSubstFail$ on failure. We explain $\PatSubst{v}{\Pg}$
below (also see \Cref{f:dynamic-patterns}).

Relation $\Reduce{\FunEnv}{e}{e'}$ performs a single reduction step from $e$ to $e'$
guided by an evaluation context, see rule \Rule{red-context}.
Relation $\Reduce{}{\Prog}{\Prog}$ takes one step from program $\Prog$ to $\Prog'$ by
reducing the program's main expression, see rule \Rule{red-prog}.

\boxfigSingle{f:dynamic-patterns}{Dynamic semantics of patterns}{
  \RuleSection{
    \RuleForm{
      \PatSubst{v}{\Pg} = \Valsubst \mid \PatSubstFail \qquad
      \PatSubst{v}{p} = \Valsubst \mid \PatSubstFail
    }
  }{
    Matching of patterns
  }
  \begin{mathpar}
    \inferrule[match-false-pat]{
      \PatSubst{v}{p} = \PatSubstFail
    }{
      \PatSubst{v}{(\WithGuard{p}{g})} = \PatSubstFail
    }

    \inferrule[match-false-guard]{
      \PatSubst{v}{p} = \Valsubst \\
      \EvalGuard{g\Valsubst} = \False
    }{
      \PatSubst{v}{(\WithGuard{p}{g })} = \PatSubstFail
    }

    \inferrule[match-true]{
      \PatSubst{v}{p} = \Valsubst \\
      \EvalGuard{g\Valsubst} = \True
    }{
      \PatSubst{v}{(\WithGuard{p}{g})} = \Valsubst
    }

    \inferrule[match-val]{}{
      \PatSubst{v_1}{v_2} = {
        \begin{cases}
          [\,] & \textrm{if}~ v_1 = v_2\\
          \PatSubstFail & \textrm{otherwise}
        \end{cases}
      }
    }

    \inferrule[match-wild]{}{
      \PatSubst{v}{\Wildcard} = [\,]
    }

    \inferrule[match-var]{}{
      \PatSubst{v}{x} = [v/x]
    }

    \inferrule[match-capture]{}{
      \PatSubst{v}{\Capt{x}} = \PatSubst{v}{\FunEnv(x)}
    }

    \inferrule[match-pair]{}{
      \PatSubst{v}{(p_1, p_2)} = {
        \begin{cases}
          \Valsubst_1 \cup \Valsubst_2 &
          \parbox[t]{8cm}{
            if $v = (v_1, v_2)$ and $\PatSubst{v_i}{p_i} = \Valsubst_i$
            for $i=1,2$ and $\Valsubst_1 \ValEq \Valsubst_2$
          } \\
          \PatSubstFail & \textrm{otherwise}
        \end{cases}
      }
    }
  \end{mathpar}

  \RuleSection{
    \RuleForm{
      \EvalGuard{g}= \True \mid \False \qquad
    }
  }{
    Evaluation of guards
  }
  \begin{mathpar}
    \inferrule[eval-guard-value]{}{
      \EvalGuard{\Is{\GuardTy}{v}} = {
        \begin{cases}
          \True   & \textrm{if}~ v \ValMatches \GuardTy\\
          \False  & \textrm{otherwise}
        \end{cases}
      }
    }

    \inferrule[eval-guard-var]{}{
      \EvalGuard{\Is{\GuardTy}{x}} = {
        \begin{cases}
          \True   & \textrm{if}~ \FunEnv(x) = v \AndText v \ValMatches \GuardTy\\
          \False  & \textrm{otherwise}
        \end{cases}
      }
    }

    \inferrule[eval-guard-oracle]{}{
      \EvalGuard{\GuardOracle} ={} ??
    }

    \inferrule[eval-guard-true]{}{
      \EvalGuard{\GuardTrue} = \True
    }

    \inferrule[eval-guard-and]{}{
      \EvalGuard{\GuardAnd{g_1}{g_2}} = \EvalGuard{g_1} \textrm{~and~} \EvalGuard{g_2}
    }
  \end{mathpar}

  \RuleSection{
    \RuleForm{
      \Valsubst \ValEq \Valsubst \quad
      v \ValMatches \GuardTy
    }
  }{
    Similarity and matching
  }
  \begin{mathpar}
    \inferrule[eq-subst]{
      \Forall{x \in \Dom{\Valsubst_1} \cap \Dom{\Valsubst_2}}~~\Valsubst_1(x) = \Valsubst_2(x)
    }{
      \Valsubst_1 \ValEq \Valsubst_2
    }
    \\

    \inferrule[is-int]{
      \Const~\textrm{is an int}
    }{
      \Const \ValMatches \Int
    }

    \inferrule[is-float]{
      \Const~\textrm{is a float}
    }{
      \Const \ValMatches \Float
    }

    \inferrule[is-pair]{}{
      \Pair{v_1}{v_2} \ValMatches \AnyPair
    }

    \inferrule[is-fun]{}{
      (\Abs{x}{e}) \ValMatches \AnyFun
    }
  \end{mathpar}
}

\Cref{f:dynamic-patterns} defines the functions $\PatSubst{v}{\Pg}$ and $\PatSubst{v}{p}$
to model the semantics of pattern matching.
If value $v$ matches the guarded pattern $\Pg$
or the pattern $p$, then the result is a substitution $\Valsubst$ for
the variables bound by the pattern. Otherwise, the result is
an error $\PatSubstFail$.

The first three rules match a guarded pattern $\WithGuard p g$ against a value. All three rules
first match the pattern. If this succeeds, rules
\Rule{match-false-guard} and \Rule{match-true} use $\EvalGuard{g\Valsubst}$
to check if the guard evaluates to $\True$.

Rule \Rule{match-val} matches a value $v_1$ against a pattern value $v_2$. The match
only succeeds if $v_1$ and $v_2$ are equal (modulo renaming of bound variables).
Rules \Rule{match-wild} and \Rule{match-var} are standard for matching wildcards and
variables. Rule \Rule{match-capture} deals with bound pattern variables $\Capt x$. Such variables
must be bound at the top-level because $\lambda$-bound variables are substituted during evaluation.
Hence, the rule obtains the definition of $x$ in $\FunEnv$, written $\FunEnv(x)$.
Rule \Rule{match-pair} handles pair patterns $\Pair{p_1}{p_2}$. If $p_1$ and $p_2$
both bind the same variable, then the values being matched must be equal, expressed
as $\Valsubst_1 \ValEq \Valsubst_2$.
The notation $\Valsubst_1 \cup \Valsubst_2$ denotes the union of $\Valsubst_1$ and $\Valsubst_2$,
which is well-defined if $\Valsubst_1 \ValEq \Valsubst_2$.
Similarity $\ValEq$
of substitutions is defined in the lower part of \Cref{f:dynamic-patterns}.

The result of evaluating a pattern guard $g$ is written $\EvalGuard{g}$.
Rule \Rule{eval-guard-value} evaluates
a type test $\Is{\GuardTy}{v}$ by checking value $v$
against some guard type $\GuardTy$, written $v \ValMatches \GuardTy$
(see bottom of \Cref{f:dynamic-patterns}).
Rule \Rule{eval-guard-var} handles type tests on variables.
Rule \Rule{eval-guard-oracle} leaves the result of oracles unspecified, written $??$,
to stress the fact that the type system must not depend on its result.
\footnote{While
Erlang's full guard system includes arity-specific function checks (e.g., \texttt{is\_fun/2}),
this complexity is omitted from our formalization as it adds no theoretical insight. In practice, extending our implementation with arity-aware function
guards is straightforward.}

\subsection{Static Semantics}
\label{sec:static-semantics}
This section formalizes a declarative type system for \MinErl{}. We rely on a subtyping relation
between set-theoretic types, written $t_1 \IsSubty t_2$. This relation
has been defined by~\cite{conf/icfp/CastagnaX11} and~\cite{conf/popl/Castagna0XA15}.
Two monomorphic types $t$ and $t'$ are equivalent, written $t \TyEquiv t'$,
iff $t \IsSubty t'$ and $t' \IsSubty t$.

A type environment $\Venv = \{ x_i : \Monoty_i \mid i \in I \}$
is a partial mapping from expression variables to (monomorphic) types. We let $\EmptyEnv$
denote the empty type environment,
$\Dom{\Venv}$ the domain of $\Venv$, and
$\Venv(x)$ the lookup of $x$'s type in $\Venv$, implicitly assuming
$x \in \Dom{\Venv}$.
The extension of a variable environment is written $\Venv, x : \Monoty$
where the new binding $x : \Monoty$ hides a potential binding for $x$ in $\Venv$.
Similarly, in the concatenation $\Venv,\Venv'$ bindings in $\Venv'$ take precedence over
those in $\Venv$. The notation $\Venv \neq \ErlBot$ requires that there is no $x \in \Dom{\Venv}$
with $\Venv(x) = \ErlBot$.

Akin to type environments, we define scheme environments $\Senv = \{ x_i : \Polyty_i \mid i \in I \}$
as partial mappings from expression variables to type schemes; all other definitions and notations are analogous.
Further, we assume that a type scheme $\forall A . t$ only contains meaningful type variables;
i.e., for all type variables $\Tyvar \in \TyvarSet$ we have
$\Monoty[\ErlBot/\Tyvar] \not\TyEquiv \Monoty$.
A formal treatment of meaningful type variables can be found in~\cite{conf/icfp/CastagnaP016}.

\subsubsection{Typing Guarded Patterns}

\boxfigSingle{f:aux-pattern-types}{Environments from guarded patterns}{
  \RuleSection{
    \RuleForm{
      \PatEnv{\Monoty}{\Pg} = \Venv \qquad
      \PatEnv{\Monoty}{p} = \Venv
    }
  }{Pattern environments}
  \begin{mathpar}
    \PatEnv{\Monoty}{(\WithGuard{p}{g})} = (\PatEnv{\Monoty}{p}) \InterEnv \Env{g}

    \PatEnv{\Monoty}{v} =
    \PatEnv{\Monoty}{\Wildcard} =
    \PatEnv{\Monoty}{\Capt{x}} = \EmptyEnv

    \PatEnv{\Monoty}{x} = \{ x : t \}

    \PatEnv{\Monoty}{(p_1, p_2)} =
    (\PatEnv{\ProjL{\Monoty}}{p_1}) \InterEnv (\PatEnv{\ProjR{\Monoty}}{p_2})
  \end{mathpar}

  \RuleSection{
    \RuleForm{
      \Env{g} = \Venv \qquad
      \TyOfGt{\GuardTy} = t
    }
  }{Guard environments and types for guard types}
  \begin{mathpar}
    \Env{\Is{\GuardTy}{x}} = \{ x : \TyOfGt{\GuardTy} \}

    \Env{\Is{\GuardTy}{v}} =
    \Env{\GuardOracle} =
    \Env{\GuardTrue} = \EmptyEnv

    \Env{\GuardAnd{g_1}{g_2}} = \Env{g_1} \InterEnv \Env{g_2}

    \\

    \TyOfGt{\Int} = \Int

    \TyOfGt{\Float} = \Float

    \TyOfGt{\AnyPair} = \Pairty{\ErlTop}{\ErlTop}

    \TyOfGt{\AnyFun} = \ErlBot \to \ErlTop
  \end{mathpar}

  \RuleSection{
    \RuleForm{\Venv \InterEnv{} \Venv = \Venv}
  }{
    Intersection of environments
  }
  \begin{mathpar}
    (\Venv_1 \InterEnv \Venv_2)(x) =
    \begin{cases}
      \Venv_1(x)  & \textrm{if}~ x \in \Dom{\Venv_1} ~\textrm{and}~ x \notin \Dom{\Venv_2} \\
      \Venv_2(x)  & \textrm{if}~ x \notin \Dom{\Venv_1} ~\textrm{and}~ x \in \Dom{\Venv_2} \\
      \Venv_1(x) \Inter \Venv_2(x)  & \textrm{if}~ x \in \Dom{\Venv_1} ~\textrm{and}~ x \in \Dom{\Venv_2}
    \end{cases}
  \end{mathpar}
}

\Cref{f:aux-pattern-types} contains definitions for computing the type environment
resulting from a guarded pattern:
$\PatEnv{\Monoty}{\Pg}$ and $\PatEnv{\Monoty}{p}$ compute
an environment $\Venv$ for the variables bound in $\Pg$ and $p$, respectively, when being
matched against a value of type $\Monoty$, and $\Env{g}$ computes
the environment resulting from guard $g$.
The definitions are straightforward except for three situations.
First, if two environments
$\Venv_1$ and $\Venv_2$
contain a type for the same variable, we intersect the two types. Intersection
for environments is written
$\Venv_1 \InterEnv \Venv_2$ and defined at the bottom of \Cref{f:aux-pattern-types}.
Second, $\TyOfGt{\GuardTy}$ yields the type $t$ corresponding to guard type $\GuardTy$. In its definition,
$\Pairty{\ErlTop}{\ErlTop}$ and $\ErlBot \to \ErlTop$ are the supertypes of all pair and function
types, respectively.
Third, the definition of $\PatEnv{\Monoty}{(p_1, p_2)}$
uses $\ProjL{\Monoty}$ and $\ProjR{\Monoty}$
to project on the left and right component of a pair type $\Monoty \IsSubty \Pairty{\ErlTop}{\ErlTop}$
(see \Cref{sec:appendix-type-safety-declarative-system}, \Cref{prop:projection}).

\boxfigSingle{f:aux-pattern-types2}{Typing guarded patterns}{
  \[
  \begin{array}{rr@{~}r@{~}ll}
    \textrm{Direction of pattern types} & \DirPatTy & ::= & {\DirPatTyUp} \mid {\DirPatTyDown}
  \end{array}
  \]
  \RuleSection{
    \RuleForm{
      \PgTy{\Pg}{\DirPatTy} = \Monoty
    }
  }{
    Potential and accepting types of guarded patterns
  }
  \begin{mathpar}
    \PgUpperTy{\WithGuard{p}{g}} = {
      \begin{cases}
        \PatUpperTy{p}{\Env{g}} & \textrm{if}~ \SafeUp{g} ~\textrm{and}~\Env{g} \neq \ErlBot\\
        \ErlBot & \textrm{otherwise}
      \end{cases}
    }
    \\
    \PgLowerTy{\WithGuard{p}{g}} = {
      \begin{cases}
        \PatLowerTy{p}{\Env{g}}
        & \textrm{if}~ \SafeDown{g}{\PatBound{p}}
          \AndText \Env{g} \neq \ErlBot
          \AndText p~\mathrm{linear} \\
        \ErlBot & \textrm{otherwise}
      \end{cases}
    }
  \end{mathpar}

  \RuleSection{
    \RuleForm{
      \PatTy{p}{\Venv} = \Monoty
    }
  }{
    Potential and accepting types of patterns
  }
  \begin{mathpar}
    \PatUpperTy{v}{\Venv} = {
      \begin{cases}
        \TyOfConst{\Const} & \textrm{if}~ v = \Const\\
        \ErlTop            & \textrm{otherwise}
      \end{cases}
    }

    \PatUpperTy{\Capt{x}}{\Venv} = \ErlTop

    \PatLowerTy{v}{\Venv} = {
      \begin{cases}
        \Const & \textrm{if}~ v = \Const \AndText \TyOfConst{\Const} = \Const\\
        \ErlBot     & \textrm{otherwise}
      \end{cases}}

    \PatLowerTy{\Capt{x}}{\Venv} = \ErlBot

    \PatTy{\Wildcard}{\Venv} = \ErlTop

    \PatTy{x}{\Venv} = {
      \begin{cases}
        \Venv(x) & \textrm{if}~ x \in \Dom{\Venv}\\
        \ErlTop     & \textrm{otherwise}
      \end{cases}
    }

    \PatTy{\Pair{p_1}{p_2}}{\Venv} = \Pairty{\PatTy{p_1}{\Venv}}{\PatTy{p_2}{\Venv}}
  \end{mathpar}

  \RuleSection{
    \RuleForm{
      \SafeUp{g} \quad
      \SafeDown{g}{X}
    }
  }{
    Safety of guards
  }
  \begin{mathpar}
    \SafeUp{\Is{\GuardTy}{x}} = \SafeUp{\GuardOracle} = \SafeUp{\GuardTrue} = \True

    \SafeUp{\Is{\GuardTy}{v}} = {
      \begin{cases}
        \True  & \textrm{if}~v \ValMatches \GuardTy \\
        \False   & \textrm{otherwise}
      \end{cases}
    }

    \SafeUp{\GuardAnd{g_1}{g_2}} = \SafeUp{g_1} \AndText \SafeUp{g_2}

    \\

    \SafeDown{\Is{\GuardTy}{x}}{X} = {
      \begin{cases}
        \True  & \textrm{if}~ x \in X \\
        \False   & \textrm{otherwise}
      \end{cases}
    }

    \SafeDown{\Is{\GuardTy}{v}}{X} = {
      \begin{cases}
        \True  & \textrm{if}~ v \ValMatches \GuardTy \\
        \False   & \textrm{otherwise}
      \end{cases}
    }

    \SafeDown{\GuardTrue}{X} = \True

    \SafeDown{\GuardOracle}{X} = \False

    \SafeDown{\GuardAnd{g_1}{g_2}}{X} =
    \SafeDown{g_1}{X} \AndText \SafeDown{g_2}{X}
  \end{mathpar}

  \RuleSection{
    \RuleForm{
      \PatBound{p} \in \PowersetFin(\AllVar) \quad
      \TyOfConst{\Const} = b
    }
  }{
    Variables bound by patterns, types of constants
  }
  \begin{mathpar}
    \PatBound{p} =
    \begin{cases}
      \{x\}  & \textrm{if}~ p = x \\
      \PatBound{p_1} \cup \PatBound{p_2} & \textrm{if}~ p = (p_1, p_2)\\
      \emptyset & \textrm{otherwise}
    \end{cases}

    \TyOfConst{\Const} =
    \begin{cases}
      \Const  & \textrm{if}~\Const~\textrm{is an int}\\
      \Float  & \textrm{otherwise}
    \end{cases}
  \end{mathpar}
}


\Cref{f:aux-pattern-types2} defines two types for a guarded pattern $\Pg$:
\begin{itemize}
\item $\PgUpperTy{\Pg}$ is the \emph{potential type}:
  if some expression matches $\Pg$,
  then the expression must be of this type.
\item $\PgLowerTy{\Pg}$ is the \emph{accepting type}:
  if some expression has this type, then $\Pg$ definitely matches the expression.
\end{itemize}

Before explaining these two kinds of types, we require some more
auxiliaries. First, defined at the bottom of \Cref{f:aux-pattern-types2},
$\PatBound{p}$ denotes the set of expression variables bound by
pattern $p$ and $\TyOfConst{\Const}$ denotes the type of some constant $\Const$.

Second, $\PatUpperTy{p}{\Venv}$ and $\PatLowerTy{p}{\Venv}$ denote the potential and accepting type, respectively,
of some pattern $p$ with respect to type environment $\Venv$. In practice, $\Venv$ is the type
environment resulting from the guard of the pattern.

The potential type of some pattern $p$ has the same intuition as for guarded patterns: if an expression
matches $p$, then the matching expression has type $\PatUpperTy{p}{\Venv}$. Thus,
if $p$ is a (numeric) constant $\Const$, than a matching expression must be of type $\TyOfConst{\Const}$.
If $p$ is a function value or a captured variable, we overapproximate its potential type with the top type $\ErlTop$.
Note that pattern matching on functions is not very useful in practice, but we include it for completeness.

For the accepting type, if some expression is of type $\PatLowerTy{p}{\Venv}$, then
$p$ definitely matches the expression. Thus, if $p$ is an int constant $\Const$, then the accepting
type is $\Const$. For float constants, there are no singleton types, so we cannot
give a precise accepting type. In this case, as well as for functions and captured variables,
we conservatively use the bottom type $\ErlBot$ as the accepting type.

Wildcard patterns always match, so both the potential and the accepting type are $\ErlTop$. Variable
patterns also match unconditionally; if the variable is bound in $\Venv$, we use the type from $\Venv$
as the potential and accepting type, otherwise $\ErlTop$. For pair types, we recurse.

The third and last auxiliary is the predicate $\SafeSym$, which
provides a static approximation on the result of some guard.
\begin{itemize}
\item $\SafeUp{g}$ is $\False$ if $g$ will never succeed, otherwise $\True$. Essentially, the
  outcome $\False$ is only possible if $g$ contains a type test $\Is{\GuardTy}{v}$ such
  that $v \ValMatches \Gt$ does not hold. (See \Cref{f:dynamic-patterns} for the definition of $\ValMatches$).
\item $\SafeDown{g}{X}$ is $\True$ if $g$ will always succeed, otherwise $\False$. Here, $X$
  is the set of variables bound by the pattern $p$ in the pattern $\Pg$. The interesting case in its
  definition is where $g$ is of the form $\Is{\Gt}{x}$. If $x \in X$,
  then the result is $\True$ because $x$ is bound by the pattern belonging to $g$ and
  the accepting type of that pattern restricts
  the type of $x$ such that the guard will succeed. Otherwise $x \notin X$, so we need
  to be conservative and the result is $\False$.
\end{itemize}

Finally, we come back to the potential and accepting type of some guarded
pattern $\Pg = \WithGuard p g$ (see top of \Cref{f:aux-pattern-types2}).
The potential type $\PgUpperTy{\Pg}$ is the potential type of $p$ with respect to environment
$\Env{g}$, provided the guard could possibly match. Otherwise, the potential type is $\ErlBot$.
For guard $g$ to possibly match, we require $\SafeUp{g}$ and $\Env{g} \neq \ErlBot$.
The latter ensures that environment $\Env{g}$ binds all variables to non-empty types;
otherwise, $g$ contains conflicting type tests for some variable $x$,
leading to $\Env{g}(x) = \ErlBot$.

The accepting type $\PgLowerTy{\Pg}$ is the accepting type of $p$ with respect to environment
$\Env{g}$, provided the guard always succeeds for the set of variables bound by $p$.
For this, we require $\SafeDown{g}{\PatBound{p}}$ and $\Env{g} \neq \ErlBot$ and
$p$ must be linear. (A linear pattern contains no duplicate bindings for the same variable.)
Two things are noteworthy here. First, guard $g$ is marked as safe, even if it contains a type test
$\Is{\Gt}{x}$ for a variable bound by $p$. This is ok because $\PatLowerTy{p}{\Env{g}}$ then
restricts the type of $x$ to $\TyOfGt{\Gt}$. Second, linearity of $p$ is required because
otherwise values matched against the same variable must be equal (see \Cref{f:dynamic-patterns},
rule \Rule{match-equal}). However, deciding such an equality statically is beyond the scope
of our type system.

We finish the explanation of the typing rules for guarded patterns with three examples.
Consider the following case-expression, where $x$ and $y$ are variables of unknown type.
$$
\Case{(x, y)}{\PatCls{\WithGuard{(1,z)}{\Is{\Int}{z}}}{\ldots}}
$$
Define $g_1 := \Is{\Int}{z}$. The environment of the guard is $\Env{g_1} = \{z : \Int\}$.
Further, we have $\SafeUp{g_1} = \True = \SafeDown{g_1}{\{z\}}$. Thus, the potential
and the accepting types are equal:
\[
\PgTy{\WithGuard{(1, z)}{g_1}}{\rho} =
\PatTy{(1, z)}{\{z: \Int\}} = \Pairty{1}{\Int} \qquad (\textrm{for}~\rho = \DirPatTyUp,\DirPatTyDown)
\]

We now change the guard to a type test on $y$:
$$
\Case{(x, y)}{\PatCls{\WithGuard{(1,z)}{\Is{\Int}{y}}}{\ldots}}
$$
Define $g_2 := \Is{\Int}{y}$. The environment of the guard is $\Env{g_2} = \{y : \Int\}$.
Further, $\SafeUp{g_2} = \True$ and $\SafeDown{g_2}{\{z\}} = \False$ because there is
a type test on a variable not bound by the pattern. Thus, we get
\[
  \begin{array}{l@{~}l@{~}l}
    \PgUpperTy{\WithGuard{(1, z)}{g_2}} &=
    \PatUpperTy{(1, z)}{\{y: \Int\}} &= \Pairty{1}{\ErlTop} \\
    \PgLowerTy{\WithGuard{(1, z)}{g_2}} &= \ErlBot
  \end{array}
\]
Our type system does not provide more precise types for such situations because
an empirical evaluation demonstrates that type tests in real-world Erlang programs
are almost always on variables bound by the pattern (see \Cref{sec:empirical-patterns}).

For the last example, we again change the guard. It now includes $\GuardOracle$, representing
a part of the guard that cannot be analyzed statically.
$$
\Case{(x, y)}{\PatCls{\WithGuard{(1,z)}{\GuardAnd{\Is{\Int}{z}}{\GuardOracle}}}{\ldots}}
$$
Define $g_3 := \GuardAnd{\Is{\Int}{z}}{\GuardOracle} $. The environment is $\Env{g_3} = \{z : \Int\}$,
and we have
$\SafeUp{g_3} = \True$ and $\SafeDown{g_3}{\{z\}} = \False$. Thus
\[
  \begin{array}{l@{~}l@{~}l}
    \PgUpperTy{\WithGuard{(1, z)}{g_3}} &=
    \PatUpperTy{(1, z)}{\{z: \Int\}} &= \Pairty{1}{\Int} \\
    \PgLowerTy{\WithGuard{(1, z)}{g_3}} &= \ErlBot
  \end{array}
\]

All expressions matching the pattern and passing the guard must hence be of type $\Pairty{1}{\Int}$
(potential type).
But due to the oracle, we cannot statically guarantee whether any expression of this type will
actually pass the guard (accepting type $\ErlBot$).

\subsubsection{Typing Expressions and Programs}

\boxfigSingle{f:decl-typing}{Typing rules for \MinErl}{
  \RuleSection{
    \RuleForm{
      \ExpTy{\Venv}{e}{\Monoty}
    }
  }{
    Typing expressions
  }
  \begin{mathpar}
    \inferrule[d-var]{
      \Venv(x) = \Monoty\\
      x \notin \Dom{\Senv}
    }{
      \ExpTy{\Venv}{x}{\Monoty}
    }

    \inferrule[d-var-poly]{
      \Monoty \in \Inst{\Senv(x)}\\
      x \notin \Dom{\Venv}
    }{
      \ExpTy{\Venv}{x}{\Monoty}
    }

    \inferrule[d-const]{}{
      \ExpTy{\Venv}{\Const}{\TyOfConst{\Const}}
    }

    \inferrule[d-abs]{
      \ExpTy{\Venv, x:\Monoty_1}{e}{\Monoty_2}
    }{
      \ExpTy{\Venv}{\Abs{x} e}{\Monoty_1 \to \Monoty_2}
    }

    \inferrule[d-app]{
      \ExpTy{\Venv}{e_1}{\Monoty' \to \Monoty}\\
      \ExpTy{\Venv}{e_2}{\Monoty'}
    }{
      \ExpTy{\Venv}{\App{e_1}{e_2}}{\Monoty}
    }

    \inferrule[d-pair]{
      \ExpTy{\Venv}{e_i}{\Monoty_i}\quad(i=1,2)
    }{
      \ExpTy{\Venv}{\Pair{e_1}{e_2}}{\Pairty{\Monoty_1}{\Monoty_2}}
    }

    \inferrule[d-case]{
      \ExpTy{\Venv}{e}{\Monoty}\\
      \Monoty \IsSubty \UnionBig\nolimits_{i \in I} \PgLowerTy{\Pg_i}\\\\
      \Forall{i \in I}\quad
      \Free{\Pg_i} \subseteq \Dom{\Senv} \cup \Dom{\Venv}\\
      \Monoty_i = (\Monoty \WithoutTy \UnionBig\nolimits_{j < i} \PgLowerTy{\Pg_j})
      \Inter \PgUpperTy{\Pg_i}\\
      \Venv_i = \Venv \InterEnv (\PatEnv{\Monoty_i}{\Pg_i})
      \\\\
      {
        \begin{cases}
          \Monoty_i' = \ErlBot ~\textrm{and}~
          \Free{e_i} \subseteq \Dom{\Senv} \cup \Dom{\Venv_i}
          & \textrm{if}~ \Monoty_i \IsSubty \ErlBot
          \\
          \ExpTy{\Venv_i}{e_i}{\Monoty_i'}
          & \textrm{otherwise}
        \end{cases}
      }
    }{
      \ExpTy{\Venv}{\Case{e}\Multi{(\Pg_i \to e_i)}}{\UnionBig\nolimits_{i \in I}{\Monoty_i'}}
    }

    \inferrule[d-sub]{
      \ExpTy{\Venv}{e}{\Monoty'}\\\\
      \Monoty' \IsSubty \Monoty
    }{
      \ExpTy{\Venv}{e}{\Monoty}
    }
  \end{mathpar}

  \RuleSection{
    \RuleForm{
      \DefOk{\Senv}{\DefSym}
    }
  }{
    Typing definitions
  }
  \begin{mathpar}
    \inferrule[d-def-no-annot]{
      \Senv(x) = \Monoty\\
      \ExpTy{\EmptyEnv}{\Abs{y}{e}}{t}
    }{
      \DefOk{\Senv}{\DefNt{x}{\Abs{y}{e}}}
    }

    \inferrule[d-def-annot]{
      \Senv(x) = \Polyty \\
      \FreeTyVars{\Polyty} = \emptyset\\
      \Polyty = \TyScm{\TyvarSet}{\InterBig\nolimits_{i \in I}{(\Monoty_i' \to \Monoty_i)}}\\
      \Forall{i \in I}\quad
      \ExpSchemaTy{\Senv}{\{y:\Monoty_i'\}}{e}{\Monoty_i}
    }{
      \DefOk{\Senv}{\Def{x}{\Polyty}{\Abs{y}{e}}}
    }

  \end{mathpar}

  \RuleSection{
    \RuleForm{
      \Envs{\DefSym} = \MetaPair{\Senv}{\Senv}
    }
  }{
    Environments from definitions
  }
  \begin{mathpar}
    \inferrule[d-envs-no-annot]{
      \TyvarSet = \FreeTyVars{\Monoty}
    }{
      \Envs{\DefNt{x}{e}} = \MetaPair{\{x : \Monoty\}}{\{x : \TyScm{\TyvarSet}\Monoty\}}
    }

    \inferrule[d-envs-annot]{
      \Polyty = \TyScm{\TyvarSet}{t}\\
      \TyvarSet = \FreeTyVars{\Monoty}
    }{
      \Envs{\Def{x}{\Polyty}{e}} = \MetaPair{\{x : \Polyty\}}{\{x : \Polyty\}}
    }
  \end{mathpar}

  \RuleSection{
    \RuleForm{
      \ProgTy{\Prog}{\Monoty}
    }
  }{
    Typing programs
  }
  \begin{mathpar}
    \inferrule[d-prog]{
      \Forall{i \in I}~\Envs{\DefSym_i} = \MetaPair{\Senv_i}{\Senv_i'}\\
      \DefOk{\Senv_{i \in I}}{\DefSym_i}\\
      \ExpSchemaTy{\Senv'_{i \in I}}{\EmptyEnv}{e}{\Monoty}
    }{
      \ProgTy{\Letrec{\Multi{\DefSym}}{e}}{\Monoty}
    }
  \end{mathpar}
}

\Cref{f:decl-typing} shows the typing rules for \MinErl{}.
Judgment $\ExpTy{\Venv}{e}{\Monoty}$ asserts that under environments $\Senv$ and $\Venv$ expression $e$
has type $\Monoty$. The motivation for having two environments is that $\Senv$ contains polymorphic
bindings from top-level $\kw{letrec}$ definitions, whereas $\Venv$ contains monomorphic bindings for
variables bound by $\lambda$ or pattern matching.

The rules for constants \Rule{d-const}, abstractions \Rule{d-abs}, applications \Rule{d-app},
pairs \Rule{d-pair}, and the subsumption rule \Rule{d-sub} are standard.
There are two rules for variables.
Rule \Rule{d-var} looks up the type of a variable in $\Venv$, whereas rule \Rule{d-var-poly} retrieves
a type schema from $\Senv$ and instantiates it. The two rules require that a variable is contained in either
$\Senv$ or $\Venv$ but never in both. We can always satisfy this requirement by renaming bound variables.
The set of instances of a type schema is defined as follows:
$$
    \Inst{\TyScm{\TyvarSet}{\Monoty}} =
    \{ \Monoty\Tysubst \mid \Dom{\Tysubst} = \TyvarSet \}
$$

Here, $\Tysubst = \Multi{[\Monoty_i/\Tyvar_i]}$ denotes the capture-avoiding type substitution
of type variables $\Tyvar_i$ with types $\Monoty_i$.
We write $\Monoty\Tysubst$ for the application of $\Tysubst$ to some type $\Monoty$,
$\Dom{\Tysubst}$ for its domain, and $\FreeTyVars{\Tysubst}$ for the set of type variables
free in the range of $\Tysubst$.

Rule \Rule{d-case} types \kw{case}-expressions. The first premise assigns type $t$ to the scrutiny $e$.
The second premise checks that the branches of the case are exhaustive: $t$ has to be a subtype
of the union $\UnionBig\nolimits_{i \in I}\PgLowerTy{\Pg_i}$ of the accepting types of all branches.

The remaining premises have to hold for all branches $i \in I$.
The restriction on $\Free{\Pg_i}$ ensures that all variables used in $\Pg_i$ are defined somewhere.
The type $t_i$ is
the input type $t_i$ for branch $i$; that is, the type we can assume for the scrutiny if the $i$-th branch
matches. To compute $t_i$,
we subtract the accepting types $\UnionBig\nolimits_{j < i} \PgLowerTy{\Pg_j}$
of all preceding branches from $t$ (the full type of the scrutiny),
as these branches could not have been matched when matching the $i$-th branch. Then
we intersect with the potential type $\PgUpperTy{\Pg_i}$
of the branch itself because matching
implies that $t_i$ is a subtype of the potential type.

$\Venv_i$ is the environment for typing the body of the $i$-th branch. It is the intersection of $\Venv$ with the
environment obtained by matching $\Pg_i$ against a value of type $t_i$, written
$\Venv \InterEnv \PatEnv{\Monoty_i}{\Pg_i}$.
We intersect the two environments
so that the type of a variable already in $\Venv$ can be further refined by pattern matching.
See \Cref{f:aux-pattern-types} for the definition of $\InterEnv$ and $\PatEnv{t_i}{\Pg_i}$.


The last premise checks the body of the $i$-th branch by assigning it the output type $t_i'$.
The type of the whole \kw{case}-expression is then the union of all output types.
If input type $t_i$ is bottom,
then the branch can never match, so the output type $t_i'$ is also bottom.
When typing a function against an intersection type,
skipping unmatched branches is essential. Otherwise, type errors in some branch
would be reported, although the branch can never match for a specific part of the intersection,
see below for an example.\footnote{%
  Our implementation reports an error if it detects that a branch does not match for
  all variants of an intersection or if there is no intersection at all.}
If input type $t_i$ is not bottom, then we type the body $e_i$ of the branch in the refined environment
$\Venv_i$.

The auxiliary judgment $\DefOk{\Senv}{\DefSym}$ checks correctness of definition $\DefSym$
under scheme environment $\Senv$. Type reconstruction with intersection types is undecidable in
general~\citep{journals/tcs/Bakel95},
so our system never guesses such types. Consequently, rule \Rule{d-def-no-annot} for definitions
without type annotations simply checks that right-hand side of the definition
has the expected type from the environment. Rule \Rule{d-def-annot} supports
intersection types in annotations by checking the function body against each type
of the intersection. As $\Senv$ only contains bindings for top-level functions,
the type scheme $\Polyty$ must be closed.

Rule \Rule{d-prog} type checks whole programs. For each definition $\DefSym_i$, it computes
two scheme environments $\MetaPair{\Senv_i}{\Senv_i'}$. For definitions with type annotations
both environments are identical. But for definitions without type annotations, $\Senv_i$ contains
only a monomorphic type because
type reconstruction with polymorphic recursion is undecidable~\citep{journals/toplas/Henglein93}.
The potentially monomorphic types $\Senv_{i \in I}$ are then used to type the definitions,
whereas the generalization $\Senv'_{i \in I}$ is used to type the main expression $e$.

\subsubsection{Example}
\label{ex:formal-filtermap}
We conclude the explanation of the typing rules of \MinErl{} by sketching how typing
proceeds for the \lstinline{filtermap} example from \Cref{sec:recursive-data-types}.
We suggested the following precise type for \lstinline{filtermap}:
\begin{lstlisting}[language=Erlang,numbers=none]
-spec filtermap(fun((T) -> boolean()), [T]) -> [T]
    ; (fun((T) -> {true, U} | false), [T]) -> [U]
    ; (fun((T) -> {true, U} | boolean()), [T]) -> [T | U].
\end{lstlisting}


We now show an encoding of \lstinline{filtermap} in \MinErl{}.
Lists are represented as nested pairs with the singleton int
$0$ as terminator. The type for lists then
satisfies the recursive equation $\ListTy{\Tyvar} = 0 \Union (\Pairty{\Tyvar}{\ListTy{\Tyvar}})$.
Integers $0$ and $1$ act as booleans false and true, respectively, so $\BoolTy$ abbreviates the type $0 \Union 1$.

\noindent
\scalebox{0.9}{%
$$
\begin{array}[t]{@{}l@{}}
\kw{letrec}\\\phantom{x}
  \Filtermap \\\phantom{x}
  \begin{array}[t]{@{~}r@{~}l@{}}
    :& \forall \Tyvar, \TyvarAux .
      \begin{array}[t]{@{}l@{}l@{}}
      (      & (\alpha \to \BoolTy) \to \ListTy{\alpha} \to \ListTy{\alpha} \\
      ~\Inter~ & (\Tyvar \to ((\Pairty{1}{\TyvarAux}) \Union 0)) \to \ListTy{\Tyvar} \to \ListTy{\TyvarAux} \\
      ~\Inter~ & (\Tyvar \to ((\Pairty{1}{\TyvarAux}) \Union \BoolTy)) \to \ListTy{\Tyvar} \to \ListTy{(\alpha \Union \beta)})
      \end{array}
  \\
  =& \lambda f\, l.\,
  \kw{case}~{l}~\kw{of}~
  {
  \begin{array}[t]{@{}l@{}}
    0 \to 0\\
    (x, l') \to~\\\quad
    {
    \begin{array}[t]{@{}l@{}}
      \kw{case}~{f~x}~\kw{of}~
      {
      \begin{array}[t]{@{}l@{}}
        0 \to \Filtermap~f~l'\\
        1 \to (x, \Filtermap~f~l')\\
        (1, y) \to (y, \Filtermap~f~l')
      \end{array}
      }
    \end{array}
    }
  \end{array}
  }
  \end{array}
\\\kw{in}\,\ldots
\end{array}
$$
}

Type checking \emph{filtermap} verifies its body against each of the three arrow
types in the intersection. Here, we concentrate on the second arrow type:
$$(\Tyvar \to ((\Pairty{1}{\TyvarAux}) \Union 0)) \to \ListTy{\Tyvar} \to \ListTy{\TyvarAux}$$

The initial environment is
$\Venv_0 = \{ f : \Tyvar \to ((\Pairty{1}{\TyvarAux}) \Union 0), l : \ListTy{\Tyvar} \}$.
Checking the first clause of the outer case-expression is straightforward.
After performing pattern matching for the second clause, we
arrive at the environment
$\Venv = \Venv_0 \cup \{x : \Tyvar, l' : \ListTy{\Tyvar} \}$.

We show how rule \Rule{d-case} type checks the inner case-expression
in this environment. Scrutiny $e = f~x$ has type $t = (\Pairty{1}{\TyvarAux}) \Union 0$.
The potential and accepting types of the three patterns of the inner case are:
\[
\begin{array}{r@{~}l@{~}l}
\PgLowerTy{0} & = \PgUpperTy{0} & = 0 \\
\PgLowerTy{1} & = \PgUpperTy{1}  & = 1 \\
\PgLowerTy{(1, y)} & = \PgUpperTy{(1, y)} & = \Pairty{1}{\ErlTop}
\end{array}
\]
Thus, $t \IsSubty 0 \Union 1 \Union \Pairty{1}{\ErlTop}$,
so the branches of the case are exhaustive. The input types of the branches are as follows:
\[
\begin{array}{r@{~}l@{~}l@{~}l}
  t_1 & = t \Inter 0 & = 0\\
  t_2 & = (t \WithoutTy 0) \Inter 1 & = \ErlBot\\
  t_3 & = (t \WithoutTy (0 \Union 1)) \Inter (\Pairty{1}{\ErlTop})
      & = (\Pairty{1}{\beta}) \Inter (\Pairty{1}{\ErlTop}))
      & = \Pairty{1}{\TyvarAux}
\end{array}
\]

We see that $t_2 = \ErlBot$, which reflects our intuition that the second branch
can never be taken when type checking against the second variant of the intersection.
Thus, output type $t_2' = \ErlBot$. For the two other output types, we have
$t_1' = \ListTy{\TyvarAux} = t_3'$. Here, the third branch is checked under the extended
environment $\Venv,y:\TyvarAux$.

Finally, we have as the type for the whole case-expression
$$t_1' \Union t_2' \Union t_3' = \ListTy{\TyvarAux} \Union \ErlBot \Union \ListTy{\TyvarAux} = \ListTy{\TyvarAux}$$
as required.


\subsection{Type Soundness}
\label{sec:soundness}

Type soundness is an important property of any type system.
Our system builds on the type system for polymorphic variants by~\citet{conf/icfp/CastagnaP016},
which has been proven to be sound.
We prove type soundness using the common syntactic approach of~\citet{journals/iandc/WrightF94}.
We refer to \Cref{sec:appendix-type-safety-declarative-system} for the full proofs.

\begin{theorem}[Progress for programs]
  Let $\Prog = \Letrec{\FunEnv}{e}$ be a program such that $\ProgTy{\Prog}{\Monoty}$
  holds for some $\Monoty$.
  Then, either $e$ is a value or there exists an expression $e'$
  such that $\Reduce{}{\Letrec{\FunEnv}{e}}{\Letrec{\FunEnv}{e'}}$.
\end{theorem}

\begin{theorem}[Subject reduction for programs]
  Let $\Prog$ be a program and $\Monoty$ be a type such that $\ProgTy{\Prog}{\Monoty}$.
  If $\Reduce{}{\Prog}{\Prog'}$ for some program $\Prog'$, then $\ProgTy{\Prog'}{\Monoty}$.
\end{theorem}

\begin{corollary}[Type soundness]
Let $\Prog = \Letrec{\FunEnv}{e}$ be a program such that $\ProgTy{\Prog}{\Monoty}$
holds for some $\Monoty$.
Then, either $\Prog$ diverges or reduces to $\Letrec{\FunEnv}{v}$ for some value $v$ such that
$\ProgTy{\Letrec{\FunEnv}{v}}{\Monoty}$.
\end{corollary}

\subsection{Algorithmic Typing Rules}
\label{sec:algor-typing-rules}

The typing rules in \Cref{f:decl-typing} do not lend themselves to a typing algorithm
because the rules are not syntax-directed. The culprit is rule \Rule{d-sub}, which can be
applied on any expression form to lift its type to some supertype.
Thus, we reformulate the typing rules
as a syntax-directed set of constraint-generation rules. The resulting constraints are
then simplified, yielding a set of subtyping constraints.
Finally, the tally algorithm by~\cite{conf/popl/Castagna0XA15} solves
the set of subtyping constraints.

\subsubsection{Constraints}
\renewcommand\showfreshtyvars{}
\boxfigSingle{f:constr-gen\showfreshtyvars}{
  Constraint generation\IfShowFresh{ with tracking of fresh type variables}
}{
  \small
  \[
    \begin{array}{r@{~~}r@{~}r@{~~}l}
      \textrm{constraints} &
      \Constr, \ConstrAlt
      & ::= &
              \SubtyConstr{\Monoty}{\Monoty} \mid
              \SubtyConstr{x}{\Monoty} \mid
              \DefConstr{\Venv}{\Constr} \mid
              \Constr \ConstrAnd \Constr\\
              &&&\mid \CaseConstr{\Constr}{\Multi{(\InConstr{\Venv_i}{\Constr_i}{\Constr_i})}}
      \\
      \textrm{program constraints} &
      \ProgConstr & ::= & \LetConstr{\Constr}{\Senv}{\Constr}
    \end{array}
  \]
  \RuleSection{
    \RuleForm{\ConstrGenV{e}{\Monoty}{\Constr}{\VarSet}}
  }{
    Constraint generation for expressions
  }
  \begin{mathpar}
    \inferrule[c-var]{}{
      \ConstrGenV{x}{\Monoty}{\SubtyConstr{x}{\Monoty}}{\EmptyVarSet}
    }

    \inferrule[c-const]{}{
      \ConstrGenV{\Const}{\Monoty}{\SubtyConstr{\TyOfConst{\Const}}{\Monoty}}{\EmptyVarSet}
    }

    \inferrule[c-abs]{
      \ConstrGenV{e}{\TyvarAux}{\Constr}{\TyvarSet'}\\
      \Tyvar, \TyvarAux ~\textrm{fresh}\\
      \IfShowFresh{\TyvarSet = \TyvarSet' \cupdisjoint \{\Tyvar, \TyvarAux\}}
    }{
      \ConstrGenV{\Abs{x}{e}}
      {\Monoty}
      {(\DefConstr{\{x : \Tyvar\}}{\Constr}) \ConstrAnd
                                       \SubtyConstr{\Tyvar \to \TyvarAux}{\Monoty}}
      {\TyvarSet}
    }

    \inferrule[c-app]{
      \ConstrGenV{e_1}{\Tyvar \to \TyvarAux}{\Constr_1}{A_1}\\
      \ConstrGenV{e_2}{\Tyvar}{\Constr_2}{A_2}\\\\
      \Tyvar, \TyvarAux ~\textrm{fresh} \\
      \IfShowFresh{\TyvarSet = \TyvarSet_1 \cupdisjoint \TyvarSet_2 \cupdisjoint \{\Tyvar, \TyvarAux\}}
    }{
      \ConstrGenV{\App{e_1}{e_2}}{\Monoty}
      {\Constr_1 \ConstrAnd \Constr_2 \ConstrAnd \SubtyConstr{\TyvarAux}{\Monoty}}{\TyvarSet}
    }

    \inferrule[c-pair]{
      \ConstrGenV{e_1}{\Tyvar_1}{\Constr_1}{A_1}\\
      \ConstrGenV{e_2}{\Tyvar_2}{\Constr_2}{A_2}\\\\
      \Tyvar_1, \Tyvar_2 ~\textrm{fresh}\\
      \IfShowFresh{\TyvarSet = \TyvarSet_1 \cupdisjoint \TyvarSet_2 \cupdisjoint \{\Tyvar_1, \Tyvar_2\}}
    }{
      \ConstrGenV{\Pair{e_1}{e_2}}{\Monoty}{
        \Constr_1 \ConstrAnd \Constr_2 \ConstrAnd \SubtyConstr{\Pairty{\Tyvar_1}{\Tyvar_2}}{\Monoty}
      }{\TyvarSet}
    }

    \inferrule[c-case]{
      \ConstrGenV{e}{\Tyvar}{\Constr}{\TyvarSet'}\\
      \ConstrAlt = \Constr \ConstrAnd (\SubtyConstr{\Tyvar}{\UnionBig\nolimits_{i \in I}{\PgLowerTy{\Pg_i}}})
          \ConstrAnd \MedConstrAnd\nolimits_{i \in I}\Constr_i \\
      \Tyvar,\TyvarAux ~\textrm{fresh}\\
      \IfShowFresh{\TyvarSet = \{\Tyvar, \TyvarAux\} \cupdisjoint \TyvarSet' \cupdisjoint
      \medcupdisjoint\nolimits_{i \in I} (\TyvarSet_i \cupdisjoint \TyvarSet_i')}\\
      \Forall{i \in I}\quad
      \Monoty_i = (\Tyvar \WithoutTy \UnionBig\nolimits_{j < i}{\PgLowerTy{\Pg_j}}) \Inter \PgUpperTy{\Pg_i}\\
      \PatTyEnvConstrV{\Monoty_i}{\Pg_i}{\Constr_i}{\Venv_i}{\TyvarSet_i}\\
      \ConstrGenV{e_i}{\TyvarAux}{\ConstrAlt_i}{\TyvarSet_i'} \\
      \MonotyAlt_i = \UnionBig\nolimits_{j < i}{\PgLowerTy{\Pg_j}} \Union \Neg \PgUpperTy{\Pg_i}
    }{
      {
        \begin{array}{@{}l@{}}
          \ConstrGenV{
            \Case{e}{\Multi{(\Pg_i \to e_i)}}
          }{
            \Monoty\\
          }{
          \big(\CaseConstr{\ConstrAlt}{\Multi{(\InConstr{\Venv_i}{\ConstrAlt_i}{\SubtyConstr{\Tyvar}{\MonotyAlt_i}})}}\big) \ConstrAnd
          \SubtyConstr{\TyvarAux}{\Monoty}
          }{\TyvarSet}
        \end{array}
      }
    }
  \end{mathpar}

  \RuleSection{
    \RuleForm{
      \PatTyEnvConstrV{\Monoty}{\Pg}{\Constr}{\Venv}{\TyvarSet}\qquad
      \PatTyEnvConstrV{\Monoty}{p}{\Constr}{\Venv}{\TyvarSet}
    }
  }{
    Environments for patterns
  }
  \begin{mathpar}
    \inferrule[c-pg-env]{
      \PatTyEnvConstrV{\Monoty}{p}{\Constr}{\Venv}{\TyvarSet}\\
      C' = \MedConstrAnd\nolimits_{x \in \Free{\WithGuard{p\,}{\,g}}} (\SubtyConstr{x}{\ErlTop})
    }{
      \PatTyEnvConstrV{\Monoty}{(\WithGuard{p}{g})}{\Constr \ConstrAnd C'}{\Venv \InterEnv \Env{g}}{\TyvarSet}
    }

    \inferrule[c-val-env]{}{
      \PatTyEnvConstrV{\Monoty}{v}{\ConstrTrue}{\EmptyEnv}{\emptyset}
    }

    \inferrule[c-wild-env]{}{
      \PatTyEnvConstrV{\Monoty}{\Wildcard}{\ConstrTrue}{\EmptyEnv}{\emptyset}
    }

    \inferrule[c-capture-env]{}{
      \PatTyEnvConstrV{\Monoty}{\Capt x}{\ConstrTrue}{\EmptyEnv}{\emptyset}
    }

    \inferrule[c-var-env]{}{
      \PatTyEnvConstrV{\Monoty}{x}{\ConstrTrue}{\{x : \Monoty \}}{\emptyset}
    }

    \inferrule[c-pair-env]{
      \Tyvar_1,\Tyvar_2 ~\textrm{fresh}\\
      \IfShowFresh{\TyvarSet = \TyvarSet_1 \cupdisjoint \TyvarSet_2 \cupdisjoint \{\Tyvar_1, \Tyvar_2\}}\\
      \PatTyEnvConstrV{\Tyvar_1}{p_1}{\Constr_1}{\Venv_1}{\TyvarSet_1} \\
      \PatTyEnvConstrV{\Tyvar_2}{p_2}{\Constr_2}{\Venv_2}{\TyvarSet_2}
    }{
      \PatTyEnvConstrV{\Monoty}{\Pair{p_1}{p_2}}{
        \Constr_1 \ConstrAnd \Constr_2 \ConstrAnd \SubtyConstr{\Monoty}{(\Pairty{\Tyvar_1}{\Tyvar_2})}
      }{
        \Venv_1 \InterEnv \Venv_2
      }{\TyvarSet}
    }
  \end{mathpar}

  \RuleSection{
    \RuleForm{\DefConstrGenV{\DefSym}{\Constr}{\Senv}{\TyvarSet}}
  }{
    Constraint generation for definitions
  }
  \begin{mathpar}
    \inferrule[c-def-annot]{
      \FreeTyVars{\Polyty} = \emptyset\\
      \Polyty = \TyScm{\TyvarSet'}{\InterBig\nolimits_{i \in I}{(\Monoty_i' \to \Monoty_i)}}\\\\
      \Forall{i \in I}\quad
      \ConstrGenV{e}{\Monoty_i}{\Constr_i}{\TyvarSet_i}\\
      \IfShowFresh{\TyvarSet = A' \cupdisjoint \medcupdisjoint\nolimits_{i \in I}\TyvarSet_i}
    }{
      \DefConstrGenV{\Def{x}{\Polyty}{\Abs{y}{e}}}{
        \MedConstrAnd\nolimits_{i \in I} (\DefConstr{ \{y : \Monoty_i'\} }{\Constr_i})
      }{
        \{x : \Polyty \}
      }{\TyvarSet}
    }

    \IfShowFreshElse{
      \inferrule[c-def-no-annot]{
        \Tyvar~\textrm{fresh}\\\\
        \IfShowFresh{\TyvarSet = \TyvarSet' \cupdisjoint \{\Tyvar\}}\\
        \ConstrGenV{\Abs{y}{e}}{\Tyvar}{\Constr}{\TyvarSet'}
      }{
        \DefConstrGenV{\DefNt{x}{\Abs{y}{e}}}{\Constr}{ \{ x : \Tyvar \}}{\TyvarSet}
      }
    }{
      \inferrule[c-def-no-annot]{
        \Tyvar~\textrm{fresh}\\
        \ConstrGenV{\Abs{y}{e}}{\Tyvar}{\Constr}{\TyvarSet'}
      }{
        \DefConstrGenV{\DefNt{x}{\Abs{y}{e}}}{\Constr}{ \{ x : \Tyvar \}}{\TyvarSet}
      }
    }
  \end{mathpar}

  \RuleSection{
    \RuleForm{
      \ConstrGenV{
        \Prog
      }{
        \Monoty
      }{
        \ProgConstr
      }{}
    }
  }{
    Constraint generation for programs
  }
  \begin{mathpar}
    \inferrule[c-prog]{
      \Forall{i \in I}~~ \DefConstrGenV{\DefSym_i}{\Constr_i}{\Senv_i}{\TyvarSet_i}\\
      \Senv = \Senv_{i \in  I}\\
      \Constr = \MedConstrAnd\nolimits_{i \in I}\Constr_i\\
      \ConstrGenV{e}{\Monoty}{\Constr'}{\TyvarSet'}
    }{
      \ConstrGenV{
        \Letrec{\Multi{\DefSym}}{e}
      }{
        \Monoty
      }{
        \LetConstr{\Constr}{\Senv}{\Constr'}
      }{}
    }
  \end{mathpar}
}

\Cref{f:constr-gen} defines the syntax of constraints $C$ and program constraints $P$.
The syntax of constraints is mostly standard \cite[Chapter 10]{PierceAdvancedTopics2024}.
The subtyping constraint $\SubtyConstr{t_1}{t_2}$ demands type $t_1$ to be a subtype
of type $t_2$; the variable constraint $\SubtyConstr{x}{t}$ requires the type of variable $x$
to be a subtype of $t$; the definition constraint $\DefConstr{\Venv}{C}$ defines the variables
in $\Venv$ for constraint $C$; the conjunction of two constraints $C_1 \ConstrAnd C_2$ requires
both constraints $C_1$ and $C_2$ to be satisfied.
The constraint resulting from a $\kw{case}$ expression requires further explanation. Such a constraint
has the form
$\CaseConstr{\Constr}{\Multi{(\InConstr{\Venv_i}{\Constr_i}{\hat{\Constr_i}})}}$.
Here, $C$ comprises the constraints for the scrutiny, $\Venv_i$ introduces variables resulting
from pattern matching for branch $i$, and $C_i$ denotes the constraints for the
body of the branch. The constraint $\hat{C_i}$ expresses the condition under which
the $i$-th branch can be taken: if $\hat{C_i}$ is satisfiable, then the branch is unreachable.
In other words: either $C_i$ must hold because the branch can be taken, or
$\hat{C_i}$ must hold because the branch is unreachable.

A program constraint $P$ is of the form $\LetConstr{\Constr}{\Senv}{\Constr'}$.
Here, $C$ is the constraint for the top-level functions resulting in schema environment $\Senv$, and $C'$
is the constraint for the main expression.

In some places, we use $\ConstrTrue$ as an abbreviation for the constraint $\SubtyConstr{\bot}{\top}$.

\subsubsection{Constraint Generation}

The judgment $\ConstrGenV{e}{\Monoty}{\Constr}{\TyvarSet}$ generates
constraint $\Constr$ for expression $e$ having type $\Monoty$ (\Cref{f:constr-gen}).
The rules for variables \Rule{c-var}, constants \Rule{c-const}, $\lambda$-expressions \Rule{c-abs},
applications \Rule{c-app}, and pairs \Rule{c-pair}
are straightforward. Rule \Rule{c-case} handles $\kw{case}$-expressions.
In this rule, constraint $C$ is the constraint for scrutiny $e$.
Constraint $D$ collects $C$, as well as the constraint for checking exhaustiveness and all constraints
$C_i$ resulting from matching guarded pattern $\Pg_i$ against $t_i$.
For the $i$-th branch, $t_i$ is the input type, $\Venv_i$ the environment, and $D_i$ the constraint of the body
(cf. $t_i$ and $\Venv_i$ in rule \Rule{d-case}, \Cref{f:decl-typing}).
Further, $u_i$ is the type for values that definitely do not match the $i$-th branch
because either the value has accepting type $\PgLowerTy{\Pg_j}$ and so matches a preceding branch or
the value does not have potential type $\PgUpperTy{\Pg_i}$ and so cannot match the $i$-th branch.
Hence, if $\alpha$ (the type variable for the scrutiny) is a subtype of $u_i$, then the branch
cannot be taken.
The $\kw{unless}$ condition of this
branch then becomes $\SubtyConstr{\alpha}{u_i}$.

The auxiliary judgment $\PatTyEnvConstrV{\Monoty}{\Pg}{\Constr}{\Venv}{\TyvarSet}$
is the constraint analogue of $\PatEnv{\Monoty}{\Pg} = \Venv$ from \Cref{f:aux-pattern-types}.
It computes the environment $\Venv$ resulting from matching guarded pattern $\Pg$ against a value
of type $t$ such that constraints $C$ restrict the freshly introduced type variables.
Rule \Rule{c-pg-env} defines this judgment; the constraints $\SubtyConstr{x}{\ErlTop}$ on
the variables $x \in \Free{\Pg}$ ensure that all $x$ are in scope.

Judgment $\DefConstrGenV{\DefSym}{\Constr}{\Senv}{\TyvarSet}$ generates constraint $C$
and scheme environment $\Senv$ for definition $\DefSym$ (cf.\ $\DefOk{\Senv}{\DefSym}$ and $\Envs{\DefSym}$ from \Cref{f:decl-typing}).
Finally,
$\ConstrGenV{\Prog}{\Monoty}{\ProgConstr}{}$ generates a constraint $P$
from program $\Prog$ with wanted type $t$. Usually, we set $t$ to some fresh type variable.

\subsubsection{Simple Constraints and Constraint Solving}
\label{sec:simple-constr-constr}

A simple constraint is either a subtype constraint $\SubtyConstr{}{}$,
a conjunction $\ConstrAnd$, or a disjunction $\ConstrOr$ (see \Cref{f:constr-rew}).
Disjunctions of constraints
are non-standard. They are required to encode the choice $C_i~\kw{unless}~\hat{C_i}$ in $\kw{case}$ constraints.

\begin{property}[Tallying algorithm]\label{prop:tally}
  There exists a terminating algorithm $\TallyBase$ such that for every simple constraint $c$ not containing
  disjunctions, $\TallyBase(c)$ is a finite set of type substitutions such that each substitution in
  $\TallyBase(c)$ is a solution of $c$.
\end{property}

Intuitively, $\TallyBase$ is to subtyping constraints what syntactic unification is to
equality constraints. However, there is no unique solution to a set of subtyping
constraints. Hence, $\TallyBase$ returns a set of solutions; the set is empty if no
solution exists. We refer to \citet[Property~A.46]{journals/corr/CastagnaP016} for details.

The algorithm $\TallySym$ extends $\TallyBase$ with support for disjunctions.
\begin{definition}[Tallying algorithm for disjunctions]
  Given some simple constraint $c$ possibly containing disjunctions, $\Tally{c}$ is a set of type substitutions.
  We compute $\Tally{c}$ as follows:
  \begin{enumerate}
  \item Bring $c$ into disjunctive normal form $c_1 \ConstrOr \ldots \ConstrOr c_n$
    such that none of the $c_i$ contains disjunctions.
  \item Solve each $c_i$ by the $\TallyBase$ algorithm from \Cref{prop:tally} to obtain
    sets of substitutions $\mathcal S_i$ for $i = 1,\ldots,n$.
  \item Define $\Tally{c} = \cup_{i = 1,\ldots,n} \mathcal S_i$.
  \end{enumerate}
\end{definition}

Each substitution in $\Tally{c}$ is a solution of $c$ that is more general than other solutions.
See \Cref{lem:props-tally} in \Cref{sec:appendix-algorithmic}.

\subsubsection{Constraint Rewriting}
\boxfigSingle{f:constr-rew\showfreshtyvars}{
  Constraint rewriting\IfShowFresh{ with tracking of fresh type variables}
}{
  \[
    \begin{array}{r@{~~}l@{~}r@{~~}l@{\quad}l}
      \textrm{simple constraints} &
         \SiConstr, \SiConstrAlt
                   & ::= & \SubtyConstr{\Monoty}{\Monoty}
                            \mid \SiConstr \ConstrAnd \SiConstr
                            \mid \SiConstr \ConstrOr \SiConstr
    \end{array}
  \]
  \RuleSection{
    \RuleForm{
      \ConstrRewV{\Venv}{\Constr}{\SiConstr}{\VarSet}
    }
  }{
    Constraint rewriting
  }
  \begin{mathpar}
    \inferrule[rc-and]{
      \ConstrRewV{\Venv}{\Constr_1}{\SiConstr_1}{\TyvarSet_1} \\
      \ConstrRewV{\Venv}{\Constr_2}{\SiConstr_2}{\TyvarSet_2} \\
      \IfShowFresh{\TyvarSet = \TyvarSet_1 \cupdisjoint \TyvarSet_2}
    }{
      \ConstrRewV{\Venv}{\Constr_1 \ConstrAnd \Constr_2}{\SiConstr_1 \ConstrAnd \SiConstr_2}{\TyvarSet}
    }

    \inferrule[rc-subty]{}{
      \ConstrRewV{\Venv}{\SubtyConstr{\Monoty}{\Monoty'}}{\SubtyConstr{\Monoty}{\Monoty'}}{\EmptyVarSet}
    }

    \inferrule[rc-var]{
      x \notin \Dom{\Senv}\\
      \Venv(x) = t'
    }{
      \ConstrRewV{\Venv}{\SubtyConstr{x}{\Monoty}}{\SubtyConstr{t'}{\Monoty}}{\EmptyVarSet}
    }

    \inferrule[rc-var-poly]{
      x \notin \Dom{\Venv}\\
      \Senv(x) = \TyScm{\TyvarSet}{\Monoty'}
    }{
      \ConstrRewV{\Venv}{\SubtyConstr{x}{\Monoty}}{\SubtyConstr{\Monoty'}{\Monoty}}{\VarSet}
    }

    \inferrule[rc-def]{
      \ConstrRewV{\Venv, \Venv'}{\Constr}{\SiConstr}{\TyvarSet} \\
    }{
      \ConstrRewV{\Venv}{\DefConstr{\Venv'}{\Constr}}{\SiConstr}{\TyvarSet}
    }

    \inferrule[rc-case]{
      \ConstrRewV{\Venv}{\Constr}{\SiConstr}{\TyvarSet'}\\
      \Forall{i \in I}~ \ConstrRewV{\Venv}{\hat{\Constr}_i}{\hat{\SiConstr}_i}{\hat{\TyvarSet}_i}\\
      \ConstrRewV{\Venv \InterEnv \Venv_i}{\Constr_i}{\SiConstr_i}{\TyvarSet_i}\\
      \IfShowFresh{\TyvarSet = \TyvarSet' \cupdisjoint
      \medcupdisjoint\nolimits_{i \in I}(\TyvarSet_i \cupdisjoint \hat{\TyvarSet_i})}
    }{
      \ConstrRewV{\Venv}{\CaseConstr{\Constr}{\Multi{(\InConstr{\Venv_i}{\Constr_i}{\hat{\Constr}_i})}}}{
        c \ConstrAnd \MedConstrAnd\nolimits_{i \in I}(\hat{\SiConstr}_i \ConstrOr \SiConstr_i)
      }{\TyvarSet}
    }
  \end{mathpar}

  \RuleSection{
    \RuleForm{
      \ConstrRewProgV{\ProgConstr}{\SiConstr}{}
    }
  }{
    Constraint rewriting for program constraints
  }
  \begin{mathpar}
    \inferrule[rc-prog]{
      \ConstrRewV{\EmptyEnv}{\Constr}{\SiConstr}{\TyvarSet}\\
      \Tysubst \in \Tally{\SiConstr}\\
      \ConstrRewFull{\Gen{\Senv\Tysubst}}{\EmptyEnv}{\Constr'}{\SiConstr'}{\TyvarSet'}
    }{
      \ConstrRewProgV{\LetConstr{\Constr}{\Senv}{\Constr'}}{
        \Equiv{\Tysubst} \ConstrAnd \SiConstr'
      }{}
    }
  \end{mathpar}

  \RuleSection{
    \RuleForm{
      \Gen{\Senv} = \Senv
      \qquad
      \Gen{\Polyty} = \Polyty
      \qquad
      \Equiv{\Tysubst} = \SiConstr
    }
  }{
    Auxiliaries
  }

  \begin{mathpar}
    \Gen{ \{ x_i : \Polyty_i \mid i \in I \}} = \{ x_i : \Gen{\Polyty_i} \mid i \in I \}

    \inferrule{
      \TyvarSet = \FreeTyVars{\Monoty}
    }{
      \Gen{\TyScm{\TyvarSet}{\Monoty}} = \TyScm{\TyvarSet}{\Monoty}
    }

    \inferrule{
      \TyvarSet = \FreeTyVars{\Monoty}
    }{
      \Gen{\TyScm{\emptyset}{\Monoty}} = \TyScm{\TyvarSet}{\Monoty}
    }

    \Equiv{\Tysubst} = \MedConstrAnd_{\Tyvar \in \Dom{\Tysubst}} (
      \SubtyConstr{\Tyvar}{\Tyvar \Tysubst} \ConstrAnd \SubtyConstr{\Tyvar \Tysubst}{\Tyvar}
    )
  \end{mathpar}
}

Constraint rewriting, also defined in \Cref{f:constr-rew}, turns a constraint $C$ into a simple constraints $c$.
The constraint rewriting rules for expressions have the form
$\ConstrRew{\Venv}{\Constr}{\SiConstr}$. Conjunction and subtyping constraints rewrite to themselves.
There are two rules \Rule{rc-var} and \Rule{rc-var-poly}
for variable constraints $\SubtyConstr{x}{t}$, depending on the environment containing $x$
(c.f rules \Rule{d-var} and \Rule{d-var-poly} from \Cref{f:decl-typing}).
Rule \Rule{rc-def} rewrites definition constraints under the extended environment.
Finally, rule \Rule{rc-case} rewrites case constraints
$\CaseConstr{\Constr}{\Multi{(\InConstr{\Venv_i}{\Constr_i}{\hat{\Constr}_i})}}$.
For each branch $i$, it rewrites $\hat{C_i}$ and $C_i$ into $c_i$ and
$\hat{c_i}$, respectively. (Rewriting $\hat{C_i}$ does not require the environment $\Venv_i$.)
The disjunction $\hat{c_i} \ConstrOr c_i$
then encodes that only one of $c_i$ and $\hat{c_i}$ must hold.\footnote{
  Our implementation reports an error \enquote{branch unreachable} if a top-level function has a type annotation
  with an intersection type but for
  all variants of the intersection $\hat{c_i}$ holds or if a top-level function has a type
  annotation without an intersection but $\hat{c_i}$ holds.
 }

For programs, we use the judgment $\ConstrRewProgV{\ProgConstr}{\SiConstr}{}$.
Given a program constraint of the form $\LetConstr{\Constr}{\Senv}{\Constr'}$, we first
rewrite $C$ to $c$. Then, we use $\TallySym$ to solve the simple constraint $c$.
For some type substitution $\Tysubst \in \Tally{c}$, we generalize the types
in $\Senv\Tysubst$ and rewrite $C'$ under this environment to $c'$. This step is
non-deterministic because $\Tally{c}$ is a \emph{set} of type substitutions.
The constraint for the whole program therefore contains not only $c'$ but also
$\Equiv{\Tysubst}$ to ensure that a solution is consistent with $\Tysubst$.

\subsubsection{Type Checking Programs}

We now describe how to type check a program $\Prog$.

\begin{enumerate}
\item Generate program constraint $P$ with $\ConstrGenV{\Prog}{\alpha}{\ProgConstr}{}$ for some fresh
  type variable $\alpha$.
\item Rewrite $\ProgConstr$ to simple constraint $\SiConstr$ with $\ConstrRewProgV{\ProgConstr}{\SiConstr}{}$. In this step,
  the premise of rule \Rule{rc-prog} might
  pick some arbitrary $\Tysubst$ from the set of type substitutions returned by $\TallySym$.
\item Solve the resulting simple constraint $c$ using $\TallySym$. If $\Tally{c} \neq \emptyset$, then $\Prog$ is type correct.
  Otherwise, we need to go back to step 2 and pick a different type substitution $\Tysubst$ there.
  If no type substitution is left, $\Prog$ does not type check.
\end{enumerate}

\subsubsection{Soundness and Completeness}

The type checking algorithm just given is
sound and complete with respect to the declarative type system (\Cref{f:decl-typing}).
See \Cref{sec:appendix-algorithmic} for full proofs.
In the following, $\SubstSolves{\Tysubst}{c}$ denotes that type substitution $\Tysubst$ is a solution to
the simple constraint $c$.

\begin{theorem}[Soundness of algorithmic typing for programs]
  \label{lem:soundness-algo-prog}
  Given program $\Prog$, some type $t$, program constraint $P$,
  simple constraint $c$, and some type substitution $\Tysubst$.
  Assume $\ConstrGen{\Prog}{\Monoty}{\ProgConstr}$ and
  $\ConstrRewProg{\ProgConstr}{c}$ and
  $\SubstSolves{\Tysubst}{c}$.
  Then $\ProgTy{\Prog}{t\Tysubst}$.
\end{theorem}

\begin{theorem}[Completeness of algorithmic typing for programs]
  \label{lem:completeness-algo-prog}
  Given program $\Prog$ and some type $t$.
  If $\ProgTy{\Prog}{t}$ then there exists program constraint $P$, simple constraint $c$,
  and a type substitution $\Tysubst$ such that
  $\ConstrGenV{\Prog}{t}{P}{}$ and $\ConstrRewProgV{P}{c}{}$ and $\SubstSolves{\Tysubst}{c}$.
\end{theorem}


\FloatBarrier
\section{Implementation and Evaluation}
\label{sec:implementation}

We validated our approach by implementing \emph{\ety{}}\footnote{
\url{https://github.com/etylizer/etylizer}}, a static type checker for Erlang, written in Erlang.
\ety{}'s implementation builds on the typing algorithm developed in \Cref{sec:formal-type-system} 
extended to handle multi-arity functions,
tuples of arbitrary size, records, bitstrings, lists, maps, and binary pattern matching.
This section evaluates the tool on real-world Erlang code.
For now, we restrict our evaluation to sequential Erlang, as \ety{} does not yet support message-passing primitives or behaviors.

We begin by describing our methodology (\Cref{sec:case-study-methodology}).
Next, we present the main findings of applying \ety{} to
five modules from Erlang's standard library
(\Cref{sec:evaluation-erlang-standard-library}).
Then, we examine the results of running \ety{} on its own code base as well as on an open-source
Erlang project (\Cref{sec:evaluation-project}).
We then discuss the performance of our tool (\Cref{sec:performance}).
Finally, we conclude with a comparison to other type analysis tools for Erlang (\Cref{sec:comparison-to-other-type-checkers}).

\subsection{Case study: Methodology}
\label{sec:case-study-methodology}

We selected code from three sources to evaluate different aspects of \ety{}.
To demonstrate applicability to well-tested, idiomatic Erlang code,
we used \ety{} to type check five modules from Erlang's standard library.
To stress-test for the type system's expressiveness and performance on code written with prior knowledge of the type checker's limitations, we checked 14 modules  from \ety{} itself.
Finally, we type check the open-source JSON library JSONE\footnote{\url{https://github.com/sile/jsone}}
to provide independent benchmarks and tests on a code base that we do not control.

\subsubsection{Module Selection}
\label{sec:module-selection}
Our selection of modules from the Erlang standard library (\lstinline|stdlib-7.0.2|) was guided by two criteria.
First, we targeted purely functional modules implementing data structures as they fall within the scope of our supported feature set.
Second, we limited our selection to modules with fewer than 20 top-level functions lacking type specification, ensuring that the effort required to add the missing annotations would be manageable.

For the JSONE library, we analyzed the entire project.
Its compact size and consistent use of type specifications for all top-level functions made it a suitable case study. 
 
 From our own \ety{} compiler, we selected a self-contained core component that is representative of the tool's functionality. 
 We intentionally excluded parts that rely on language features not yet supported by our type checker, such as try/catch expressions.

In principle, \ety{} can reconstruct types for top-level functions without type specifications, but doing so incurs a substantial performance cost, which we aimed to avoid.
The type system lacks principal types because a set of subtyping constraints may have multiple, mutually incompatible
solutions.
Hence, the type checker would first need to reconstruct the types of all unannotated functions, and then determine whether the
remaining functions are type-safe under any of these solutions, potentially requiring backtracking.

That said,
\ety{} still performs type reconstruction for locally defined functions, since
Erlang provides no syntax for annotating such functions with a type.
To maintain efficiency, the inferred types of local functions are not generalized.
Empirical studies of type checking in Haskell suggest that type generalization for local functions is rarely needed in practice~\citep{journals/jfp/VytiniotisJSS11}.


\subsubsection{Type Checker Execution}
\label{sec:type-check-exec}
For the evaluation, each module was processed by \ety{} in a separate run.
We categorized the outcome for every function:
\begin{itemize}
	\item \textbf{Type safe:} \ety{} successfully verified the function's body against its specification.
	\item \textbf{Type error:} The type checker found an issue. This could either mean that the function can potentially crash while providing a valid input, return a value outside its specified type, or that the type system is not capable to verify type safety for a function that is safe to execute. Instances arising from the type system's limitations are explicitly identified and discussed, with the ideal outcome being that no such cases occur.
	\item \textbf{Pending:} The result is neither a clear pass nor a definitive fail.Type checking the function requires a feature not yet implemented in \ety{}. 
	\item \textbf{Timeout:} Solving subtyping constraints failed to terminate within a predefined time limit; this is sometimes the case for very complex expressions (e.g., long constant lists, large case expressions).
\end{itemize}

\subsubsection{Code Modifications}
\label{sec:code-modifications}

In some cases, we had to adjust the existing code before passing it to \ety{}. 
We made two sorts of modifications: adding missing type specifications, and specifying which function clauses deliberately violate exhaustiveness.
No other kind of code modifications were performed.

\paragraph{Adding type specifications}
Modern Erlang code is often written with types in mind, and it is common practice for exported functions to include type specifications. 
All top-level functions in both the \ety{} code base used for the case study and the JSONE library are already annotated with such a specification.
In contrast, Erlang’s standard library frequently omits specifications for private helper functions.
As discussed in \Cref{sec:module-selection}, we opted to add type specifications to all top-level functions for performance reasons, deferring efficient type reconstruction for top-level functions and gradual typing to future work. 
In total, we added 47 type specifications.


\paragraph{Exhaustiveness Violations}

In the presence of deliberate non-exhaustive patterns, we disable exhaustiveness checks for specific functions rather than adding trivial catch-all clauses. 
This approach is applied in two situations: when a function or \lstinline|case| expression is intentionally non-exhaustive, or when patterns are semantically exhaustive due to data invariants even though they do not syntactically cover all cases.
For example, \lstinline|orddict:fetch/2| from Erlang's standard library performs lookup of a specific key in a list of key-value pairs. 
It is designed to crash when called with an empty list.
Instead of complicating its type specification or adding a redundant clause, we disable checking for exhaustiveness pattern matches with a custom flag only for those functions.
This preserves the original, simple specification and serves as a formal directive for the type checker and as an explicit, inline documentation for developers reading or modifying the module.

\subsection{Evaluation: Erlang Standard Library}
\label{sec:evaluation-erlang-standard-library}

\begin{table}[tp]
	\centering
\begin{tabular}{|c|c|c|c|c|c|c|c|c|}
\hline
\rowcolor{GrayBgColor}\textbf{module} & \textbf{LOC} & \textbf{funs} & \textbf{type safe} & \textbf{type error} & \textbf{pending} & \textbf{avg} & \textbf{max} & \textbf{total}  \\
\hline
ordsets & 265 & 23 & 21 & 2 & 0 & 0,9s & 9s & 20s \\
\hline
queue & 750 & 50  & 47 & 3 & 0 & 1,5s & 39s & 75s \\
\hline
orddict & 300 & 23  & 19 & 4 & 0 & 1,1s & 5s & 26s \\
\hline
proplists & 500 & 35 & 17 & 0 & 18 & 0,9s & 2s & 32s \\
\hline
sets & 1000 & 55 & 36 & 9 & 10 & 1,1s & 8s & 62s \\
\hline
\end{tabular}
\vspace{.05cm}
\caption{Evaluation results for Erlang standard library modules (version 7.0.2).\\
The first three columns give the module name, the lines of
code, and the number of functions in the module.
The next columns state how many functions \ety{} classified as \enquote{type safe}, \enquote{type error}, or
\enquote{pending} (c.f. \Cref{sec:type-check-exec}). The two columns \enquote{avg} and \enquote{max}
display the average and maximum time
required to type check a function. The last columns show the total time required for type checking the module.
}
\label{table:case-study-stdlib}
\end{table}

\Cref{table:case-study-stdlib} summarizes the results of running \ety{} on five modules
from Erlang's standard library.

\paragraph{ordsets.erl}

Erlang's \lstinline|ordsets| are sets where elements are stored in a sorted order.
The module contains 265 lines of code with 23 functions, type checked with an average of about a second of time needed per function.
After adding two missing type annotations, 21 functions were verified for type safety and two issue were encountered. 
Specifically, \lstinline|is_set/1| does not handle its type specification exhaustively, resulting in a runtime crash for certain inputs covered by its specification, while \lstinline|filtermap/2| had a typo in its type specification where the atom \lstinline|boolean| was used instead of the type \lstinline|boolean()|.

\paragraph{queue.erl}

The Erlang module implementing queues consists of approximately 750 lines of code and contains 50 functions.
We disabled exhaustiveness checks for three functions.
For type checking two functions of the queue module, we had to define an overlay (an alternative type specification used by \ety{}) for \lstinline{lists:reverse/1}.
 In Erlang, \lstinline{lists:reverse/1} has type \lstinline{[term()] -> [term()]}.
With the overlay, we added the case \lstinline{nonempty_list(T) -> nonempty_list(T)} so that \ety{} knows that \lstinline{reverse} preserves
non-emptiness of lists and modified the base case to have type \lstinline{[A] -> [A]} to preserve the input-output relationship.

We added 11 missing type annotations for helper functions and 47 out of 50 functions were successfully type checked in under 80 seconds and three issues were found.
The type specification for \lstinline|get/2| is not precise enough to type check its implementation.
 It uses generic \lstinline|list()| and \lstinline|term()| types for its input and output, resulting in transitive type errors for the functions \lstinline|get/1| and \lstinline|head/1| where applying the function \lstinline|get/2| results in a loss of type information. 

\paragraph{orddict.erl}

The \lstinline|orddict| module implements ordered dictionaries as sorted key-value lists, comprising 23 functions across approximately 300 lines of code. 
We successfully type checked 19 functions with an average verification time of about a second per function, after adding two missing type specifications.
We disabled exhaustiveness checks for two functions, \lstinline|fetch/2| and \lstinline|update/3|.
Further, we define the overlay \lstinline|(1, [{K, V}]) -> [{K, V}]| for an invocation of \lstinline|list:ukeysort/2| with a constant index. 

For functions \lstinline|append/3|, \lstinline|append_list/3|, \lstinline|update_counter/3|, and \lstinline|merge/3|, the original type specification was incorrect. 
These are specialized functions designed to operate on instantiations of the \lstinline|orddict| type, but their specifications were written in a way that suggested they could handle any polymorphic dictionary.

\paragraph{proplists.erl}

The \lstinline|proplists| module implements property list manipulation utilities.
The data structure underneath uses lists with tuples of varying arities as elements.
The module comprises 35 functions across approximately 500 lines of code. 
17 functions could be verified for type safety in about 30 seconds, while 18 are categorized as pending due to a limitation of the Erlang type system related to tuple types (c.f. \Cref{sec:tuple-types}).

\paragraph{sets.erl}

The sets module, implementing both version 1 (hash table-based) and version 2 (map-based) set representations, comprises 55 functions across approximately 1,000 lines of code. 
36 functions were successfully verified in about six minutes, while 8 could not be verified due to a limitation of the Erlang type system (again tuple types, c.f. \Cref{sec:tuple-types}), and two functions due to missing support for map generator comprehension in \ety{}, resulting in 10 functions in the pending category.
Further, we found 9 issues.
The functions \lstinline|update_bucket/3|, \lstinline|get_slot/2|, \lstinline|get_bucket/2|, \lstinline|maybe_expand/1|, \lstinline|maybe_expand_segs/1|, \lstinline|maybe_contract/2|, and \lstinline|maybe_contract_segs/1|, were specified to handle both versions of the sets data structure, but their implementations assumed a specific version.
The type checker correctly identified this mismatch. 
Further, the function \lstinline|filtermap/2| produced a type error because it initializes an expression with \lstinline|new([{version, 1}])|, which has type \lstinline|set(none())|.
Ideally, \lstinline|new/0| should have a polymorphic type \lstinline|new() -> set(A)| to serve as an empty set for any element type \lstinline|A|, but Erlang's compiler considers variables that appear only in the return type unbounded and treats this as a compilation error.
Finally, the type errors for \lstinline|is_element/2|, \lstinline|add_element/2|, and \lstinline|del_element/2| were transitive, stemming from imprecise type specifications in functions they call.

\subsubsection{Analysis of Type Errors}

In the following, we discuss the four distinct categories of issues encountered while type checking the standard library modules. Each category represents a different class of software quality concern.

\paragraph{Imprecise Type Specifications}
We identified 17 violations where type specifications were too imprecise to accurately reflect function semantics. A representative example occurs in \lstinline{sets.erl}, where the module supports two internal implementations but the type specification fails to distinguish them:

\begin{lstlisting}
-opaque set(Element) ::
    #set{segs :: segs(Element)}   % record-based
  | #{Element => ?VALUE}.         % map-based

-spec get_slot(set(E), E) -> non_neg_integer().
get_slot(T, Key) ->
  H = erlang:phash(Key, T#set.maxn), ... % Crashes if T is a map
\end{lstlisting}

Here, \lstinline{get_slot/2} expects a \lstinline|#set{}| record but the specification permits both record and map representations.
While the surrounding code might ensure correct usage, \ety{} flags this as a potential runtime error.
This finding suggests that refining the specification to separate the two implementation types would allow the type system to formally verify the intended usage patterns.

Another example of an invalid specification was recently introduced in the function \lstinline|ordsets:filtermap/2|:\footnote{commit \texttt{be1d1f71a18796f083aaa2a852dd01e653083186} in \url{https://github.com/erlang/otp}}

\begin{lstlisting}
-doc(#{since => <<"OTP 27.0">>}).
-spec filtermap(Fun, Ordset1) -> Ordset2 when
      Fun :: fun((Element1 :: T1) -> boolean | ({true, Element2 :: T2})),
      Ordset1 :: ordset(T1),
      Ordset2 :: ordset(T1 | T2).
\end{lstlisting}

Even without examining the function body, the specification reveals a clear issue: the return type of the callback function includes the atom \lstinline|boolean| rather than the type \lstinline|boolean()|. 
This subtle (syntactic) distinction between a specific atom and a boolean type would be flagged by more rigorous tooling, demonstrating how such errors could be caught before inclusion in the standard library.

\paragraph{Implementation bugs}

We identified one function in the \lstinline|ordsets.erl| module containing a genuine bug that could lead to runtime crashes.
An ordered set is represented by an ordered list. The function \lstinline{is_set/1} checks if some value is indeed such an ordered set.

\begin{lstlisting}
-spec is_set(Ordset) -> boolean() when Ordset :: term().
is_set([E|Es]) -> is_set(Es, E);
is_set([]) -> true;
is_set(_) -> false.

-spec is_set(list(term()), term()) -> boolean().
is_set([E2|Es], E1) when E1 < E2 -> is_set(Es, E2);
is_set([|], _) -> false;
is_set([], _) -> true.
\end{lstlisting}

The main function \lstinline|is_set/1| is specified for all terms, and correctly has a catch-all function clause for non-list inputs.
The helper function \lstinline|is_set/2| is specified for lists only.
The issue is that the input to \lstinline|is_set/1| might be a non-empty, improper list.\footnote{An improper list is a list where the tail is not a list.}
For example, calling \lstinline|is_set/1| with argument \lstinline|[1|$\mid$\lstinline|1]| leads \lstinline|is_set/2| to crash.

This bug was not detected by any other type checker, including Dialyzer, Gradualizer, eqWAlizer, and the type inference algorithm of \citet{conf/erlang/ValliappanH18}.

\paragraph{Parametric polymorphism}

A notable class of specification issues involved incorrect parametricity, where type variables suggested broader applicability than the implementation actually supported.
Here is an example from the orddict module, which implements dictionaries with sorted key-value lists.
The function \lstinline{append_list} appends a list of values to the current list of values associated with some key.

\begin{lstlisting}
-spec append_list(Key, ValList, Orddict) -> Orddict 
    when ValList :: [Value], Orddict :: orddict(Key, Value).
append_list(...) -> ...;
append_list(Key, NewList, [{_K,Old}|Dict]) ->	%Key == K
    [{Key,Old ++ NewList}|Dict];
...
\end{lstlisting}

The specification suggests the value type of \lstinline|Orddict| can be any type \lstinline|Value|.
As the \lstinline|++| operator suggests, the implementation actually requires the value to be of type \lstinline|ValList|.
Here is the corrected type specification:

\begin{lstlisting}
-spec append_list(Key, ValList, Orddict) -> Orddict
    when ValList :: [Value], Orddict :: orddict(Key, ValList).
\end{lstlisting}

Here, \lstinline|orddict| has the value type \lstinline|ValList|, whereas before, it was \lstinline|Value|. 
With this specification, \ety{} is able to type check the function.

\paragraph{Idiomatic crash behavior}

Some function applications intentionally implement only the "happy path" using idiomatic Erlang patterns. 
For example:

\begin{lstlisting}
-spec filelib:ensure_dir(term()) -> ok | {error, term()}.

-spec mkdirs(file:filename()) -> ok | {error, string()}.
mkdirs(D) ->
  ok = filelib:ensure_dir(filename:join(D, "temp")).
\end{lstlisting}

Here, the pattern match on \lstinline|ok| explicitly assumes success, treating any error result as a crash-worthy exception. 
By default, we also check for exhaustiveness in \lstinline|A = B| expressions.
 
When \ety{} reports non-exhaustiveness, the following things need to be assessed:
\begin{enumerate}
    \item The developer must check whether the missing cases represent genuine bugs or intentional design.
    \item If it is a bug (missing error handling), this constitutes a type error that should be fixed.
    \item If it is intentional (early crash on errors), the developer can disable exhaustiveness checking for that specific function via a custom flag.
      (Currently, it is not possible to disable exhaustiveness for individual expressions.)
\end{enumerate}

In our case study, we applied this reasoning process and disabled exhaustiveness checking for functions where crash behavior was clearly intentional.

\paragraph{Tuple types} 
\label{sec:tuple-types}

For the \lstinline|proplists.erl| module, our analysis revealed that 18 functions could be supported by our type theory, but resulted in type errors due to the current limitation of tuple type specifications.

A proplist in Erlang is a list such that each element is either a non-empty tuple or an atom \lstinline{A} abbreviating the tuple \lstinline|{A, true}|.
The first element of each tuple serves as the key. (The list might contain other elements, but these are ignored.)
Here is the function \lstinline|lookup/2| that searches for a given key:

\begin{lstlisting}
-spec lookup(term(),  [term()]) -> none | tuple().
lookup(Key, [P | Ps]) ->
    if is_atom(P), P =:= Key -> {Key, true};
       tuple_size(P) >= 1, element(1, P) =:= Key -> P;
       true -> lookup(Key, Ps)
    end;
lookup(_Key, []) -> none.
\end{lstlisting}

The given type specification \lstinline|term()| for \lstinline|P| is too generic to be able to type check the implementation. 
Instead, this pattern requires specifying that a tuple contains elements of the type of the key at the first position, a concept not currently expressible in Erlang's type language.
Similarly, functions in the \lstinline|sets.erl| module categorized as pending could not be verified due to their reliance on processing tuples of varying sizes.
The implementation in version 1 extensively uses tuples for hash segments as dynamic arrays. 
The absence of this feature does not represent a fundamental limitation as set-theoretic types can in principle express such properties.

\subsubsection{More Modules From Erlang's Standard Library}

Checking larger and more complex modules in the standard library requires handling additional language features, most notably \lstinline|try| expressions, as well as support for gradual typing or improved performance of type reconstruction. 
We are currently working on incorporating gradual types, as described by~\cite{castagna2024guardanalysissafeerasure}, and enhancing the performance of type reconstruction.

\subsection{Evaluation: Erlang Projects}

\label{sec:evaluation-project}

\begin{table}[tp]
	\centering
  \small
  \begin{tabular}{|c|c|c|c|c|c|c|c|c|c|c|}
	\hline
  \rowcolor{GrayBgColor}
	\textbf{project} & \textbf{mods} & \textbf{LOC} & \textbf{funs} & \shortstack{\textbf{type}\\\textbf{safe}} & \shortstack{\textbf{type}\\\textbf{error}} & \textbf{pending} & \textbf{timeout} & \textbf{avg} & \textbf{max} & \textbf{total} \\
	\hline
	jsone & 4 & 1045 & 65 & 34 & 9 & 20 & 2 & 8,0s & 189s & 502s \\
	\hline
	etylizer & 14 & 2770 & 231 & 204 & 6 & 13 & 8 & 8,2s & 235s & 1830s \\
	\hline
	stdlib & 5 & 2815 & 186 & 140 & 18 & 28 & 0 & 1,2s & 39s & 215s \\
	\hline
\end{tabular}
\vspace{.05cm}
\caption{Case study: Erlang Projects.}
\label{table:case-study-projects}
\end{table}

\Cref{table:case-study-projects} presents the result of running \ety{} on parts of Etylizer's own code
and the open-source library JSONE. For completeness, the table also summarizes the results for
Erlang's standard library, which we discussed in great detail in the preceding section.

We selected some core modules of \ety{} itself, in total 14 modules with 231 functions across 2770 LOC.
Type checking these modules resulted in 6 type errors; 8 functions lead to a timeout.
For 13 functions \ety{} could not decide type safety because of unsupported expressions not yet implemented (try-catch, nominal types).
Further, we disabled exhaustiveness checking for parts of \ety{} for a total of 3 functions for similar reasons as explained in \Cref{sec:code-modifications}.

From the JSONE project, we include all four modules with 65 functions in total.
Here, we encountered 9 type errors, and 20 are pending. For two functions \ety{} timed out.


\paragraph{Code Modifications}

In addition to the modifications mentioned in \Cref{sec:code-modifications}, we had to massage the code slightly more before passing it to \ety{}.
\begin{enumerate*}[label=(\arabic*)] 
\item The type checker does not support default values for record fields, so we had to pass values for such fields explicitly.
\item One function used heterogeneous lists only to give special treatment to the first element.
  We rewrote the function header to accept the first element as an extra parameter.
\item We had to fix one incorrect type specification.
\item  Retrieving values from the Erlang term storage yields only the \lstinline|term()| type and thus type information is lost. Therefore, we had to add guards for one function to restore the expected type before further processing.
\end{enumerate*}

\subsubsection{Type Errors}

\paragraph{Rank-2 Types}

The evaluation revealed one function in \ety{} that cannot be type checked:
a generic traversal function that recursively transforms arbitrary nested data structures (\Cref{fig:everywhere}).

\begin{figure}[t]
\begin{lstlisting}
% everywhere(F, X) recursively applies F to X.
% F is a function that returns
%   * {rec, X} when X should be traversed recursively,
%   * {ok, X} when X should not be processed any further,
%   * error if a default traversal should take place for lists, tuples,
%     and maps.
% X is the value to be traversed.
-spec everywhere(
    fun((term()) -> {rec, term()} | {ok, term()} | error,
    term()
) -> term().
everywhere(F, T) ->
    TransList =
        fun(L) -> lists:map(fun(X) -> everywhere(F, X) end, L) end,
    case F(T) of
        error ->
            case T of
                X when is_list(X) -> TransList(X);
                X when is_tuple(X) ->
                    list_to_tuple(TransList(tuple_to_list(X)));
                X when is_map(X) ->
                    maps:from_list(TransList(maps:to_list(X)));
                X -> X
            end;
        {ok, X} -> X;
        {rec, X} -> everywhere(F, X)
    end.
\end{lstlisting}
\caption{The \lstinline|everywhere| function requiring rank-2 polymorphism.}
\label{fig:everywhere}
\end{figure}

The function is inspired by the generic \lstinline{everywhere} traversal from
the Haskell library \enquote{scrap your boilerplate}.
Type checking the \lstinline{everywhere} function in Haskell requires rank-2 polymorphism~\citep{DBLP:conf/tldi/LammelJ03}.
Our type system lacks this feature, making this the only type error that could not be resolved through improvements
of the type checker.

In idiomatic Erlang, this pattern is uncommon; developers typically write type-specific traversals for each data structure rather than using generic transformations.

\paragraph{Overlapping Types}

All other type errors in the \ety{} codebase stem from overlapping tagged union types. 
Consider the specification of \lstinline|form()|, where the third element of the tuple serves as a discriminator between union members:

\begin{lstlisting}
-type import() :: {attribute, anno(), import, {atom(), [fun_with_arity()]}}.
-type wild_attr() :: {attribute, anno(), atom(), term()}.
-type form() :: import() | wild_attr() | ...
\end{lstlisting}

To refine the type \lstinline|form()| to an \lstinline|import()| via a case expression, we pattern match against the structure of the \lstinline|import()| type:

\begin{lstlisting}
  case Form of % Form is an element of the Forms list
    F = {attribute, Anno, import, {ModName, Funs}} -> ...
\end{lstlisting}

Thus, we expect that in the matching branch, \lstinline|F| has the type \lstinline|import()|.
However, the \lstinline|wild_attr| type prevents unique matching because its third component \lstinline|atom()| overlaps with the more specific \lstinline|import| atom of the \lstinline|import()| type. 
The same is true for the fourth component of the tuple.

Apart of changing the data structure itself,
an elegant solution would be to extend the type language with type difference.
This would allow excluding the \lstinline|import| atom from the third position of \lstinline|wild_attr|, using a specification like \lstinline|ety:without(atom(), import())|. 
This would resolve the overlap without changing the data structure. 
However, to maintain compatibility with Erlang's standard type language, we refrained from changing the type language.

%

\paragraph{Precise Number Types}

Two functions in the JSONE library required disabling precise number type verification due to limitations in tracking integer ranges through arithmetic operations.

\begin{lstlisting}
-type utc_offset_seconds() :: -86399..86399.
-spec local_offset() -> utc_offset_seconds().
local_offset() ->
  UTC = {{1970, 1, 2}, {0, 0, 0}},
  Local = calendar:universal_time_to_local_time({{1970, 1, 2}, {0, 0, 0}}),
  calendar:datetime_to_gregorian_seconds(Local) 
    - calendar:datetime_to_gregorian_seconds(UTC).
\end{lstlisting}

The subtraction operation loses the precise range information from the input values, preventing verification that the result falls within the specified bounds. 
However, the narrow return type remains valuable for documentation and API clarity. 
\ety{} supports a flag to selectively disable precise number checking for specific functions while preserving the precise specifications, acknowledging that the implementation upholds the range invariant through external guarantees rather than static verification.

\paragraph{Dependent Types}
\label{sec:dependent-types}

Some Erlang operations are dynamic in nature and their precise type would require value-dependent types.
Here are two examples together with their types in Erlang's standard library.
\begin{itemize}
\item \lstinline|erlang:element/2| of type \lstinline|(pos_integer(), tuple()) -> term()| accesses the $N$-th
  element of a tuple. 
 \item \lstinline|erlang:setelement/3| of type \lstinline|(pos_integer(), tuple(), term()) -> tuple()|
   sets the $N$-th element of a tuple.
\end{itemize}

Giving a more precise type would require the length and element types of the tuple to be reflected on the type level.
Since we cannot modify the official type specifications in the standard library, we employ type overlays.
These allow us to refine the types of such functions by overloading them with cases for statically known indices (singleton types) and specific tuple shapes.

\begin{lstlisting}
% Overlay for element
-spec 'erlang:element' (1, {A}) -> A;
  'erlang:element' (2, {_, B}) -> B;  
  'erlang:element' (2, {_, B, _}) -> B;  
  ...; 
  (_, tuple()) -> term().
\end{lstlisting}

These overlays encode the semantic behavior of functions with specialized cases,
effectively expressing pattern matching at the type level. 

To assess the practical impact of value-dependent typing, we performed an empirical analysis across six Erlang projects. 
Our results show that while functions like \lstinline|element/2| and \lstinline|setelement/3| are inherently dynamic, the majority of their usage (76\% of 1820 calls) employs literal constant indices. 
These cases are precisely handled by our type overlays, which provide specifications for each constant index. 
The remaining 24\% of calls use non-constant expressions, representing the genuinely dynamic cases that would require full dependent typing. (See \Cref{sec:empirical-dynamic-functions} for the full study.)

This demonstrates that value-dependent typing, while theoretically challenging, has limited practical impact on verifying typical Erlang code, as the common cases are adequately handled by our approach.

For runtime-dependent indices, our current approach of overloading on singleton types is insufficient. 
A potential future extension to the type language, such as introducing tuple arity polymorphism (e.g., by providing a type like \lstinline|tuple(T)| to denote tuples of any size where all elements are of type \lstinline|T|), would allow us to verify more dynamic cases. 
For example, the \lstinline|hex/1| function from the JSONE library, which indexes into a homogeneous tuple of integers, could then be verified.

\begin{lstlisting}
-spec 'erlang:element' (integer(), tuple(T)) -> T.
-spec hex(byte()) -> 0..16#FFFF.
hex(X) ->
  element(X + 1,
  {16#3030, 16#3031, ...}).
\end{lstlisting} 

Currently, this specific pattern remains unverified and is categorized as pending and left for future work.


\subsection{Performance \& Timeouts}
\label{sec:performance}


\ety{} offers an expressive and powerful type system.
It supports features such as redundancy and exhaustiveness checks for pattern matching, a limited form of type narrowing, polymorphic functions with union and intersection types, as well as type reconstruction for local and top-level functions (even though the real-world case studies
only relied on type reconstruction for local functions).
This allows for fine-grained static guarantees about program behavior and makes already existing Erlang code fragments typeable.

The computational cost of this expressiveness is substantial. 
Subtyping in set-theoretic type systems is EXPTIME-complete~\citep{conf/icfp/CastagnaX11}, 
and while CDuce~\citep{conf/icfp/BenzakenCF03} demonstrates practical performance for type checking, 
\ety{} faces the additional complexity of type reconstruction for locally defined functions.
It is still an open challenge how to significantly improve the performance of such a constraint system with type reconstruction for locally defined functions in the presence of both set-theoretic and polymorphic types.

While the average verification time remains in the low seconds, certain functions require minutes to complete type checking. 
The performance impact is mitigated by caching verification results, ensuring this cost is incurred only once per function during development. 
The distribution shows that the majority of functions verify quickly, with prolonged verification times affecting only a small subset of particularly complex cases.

In total, for 2\% of the functions in the whole case study, \ety{} could not type check a function and exceeded the timeout limit of 5 minutes.
Most of these functions perform recursive transformations on deeply nested, recursively defined types, such as the types for the Erlang AST, which results in case expressions with over 40 branches.
Improving performance for such functions is a key challenge in advancing the state of the art of set-theoretic types further.

\subsection{Comparison to other static type analysis tools}
\label{sec:comparison-to-other-type-checkers}

Several static type checkers and analyzers have been developed for Erlang, including \dialyzer{} \citep{dialyzer}, \eqwalizer{} \citep{whatsapp-whatsappeqwalizer-nodate}, and \cite{gradualizer}. While each tool has distinct characteristics, they exhibit specific limitations in their type systems. For instance, consider distributivity of union and product types:

\begin{lstlisting}
-spec dist_left({ok | err, nil}) -> {ok, nil} | {err, nil}.
dist_left(X) -> X.

-spec dist({ok | err, arg | nil}) -> 
  {ok, arg} | {err, arg} | {ok, nil} | {err, nil}.
dist(X) -> X.
\end{lstlisting}

The distributivity law for product types over unions states that for any types $A$, $B$, and $C$:

\begin{align*}
(A \cup B) \times C &\equiv (A \times C) \cup (B \times C) &\quad \text{(left distributivity)} \\
A \times (B \cup C) &\equiv (A \times B) \cup (A \times C) &\quad \text{(right distributivity)}
\end{align*}


While applying distributivity once in isolation is accepted, nested distributivity causes type errors in both \eqwalizer{} and \gradualizer{} due to their ad-hoc implementations, whereas \ety{} handles distributivity through its formal foundation.
In our example, \lstinline|dist_left| type checks in all three tools, while \lstinline|dist| type checks only with \ety{}.

Similarly, intersection types present challenges:

\begin{lstlisting}
-spec inter (integer()) -> integer(); (atom()) -> atom().
inter(X) -> case X of
  _ when is_integer(X) -> X + 1;
  _ -> X
end.
\end{lstlisting}

Both \eqwalizer{} and \gradualizer{} fail to properly handle occurrence typing in case branches with intersection types, resulting in unexpected type errors, thus failing to type check the function.
In contrast, 
for \ety{} the interaction between cases expressions, guards, and intersection types work as expected, with the limitation of linear patterns.
\gradualizer{} faces additional challenges as it is no longer actively maintained.

\dialyzer{}, by design, omits certain classes of type discrepancies from detection. 
Its success typing approach guarantees the absence of false positives, but at the cost of potentially missing real errors.
As such, in our comparative evaluations, \dialyzer{} does not report several type errors detected by other tools.
However, it is important to note that \dialyzer{} had already been used during the development process of the analyzed codebases, meaning many issues it would have caught had already been corrected. 
This establishes \dialyzer{} as a useful baseline rather than a fair point of comparison for assessing detection capabilities on previously unchecked code.
While \ety{} provides more principled handling of these typing scenarios, its enhanced expressiveness comes with increased computational costs. 
These comparative observations are substantiated by both quantitative analysis \citep{conf/ifl/SchimpfWB22} and qualitative
assessments \citep{conf/erlang/BergerSBW24}.


All code examined in the case study was developed with \dialyzer{} in mind and successfully passes its analysis.
\gradualizer{} proved of limited comparative value,
as it produced a large number of type errors for functions that are, in fact, type correct.

The comparison with \eqwalizer{} reveals more substantive differences: of the 28 issues identified by \ety{}, 14 could not be found by \eqwalizer{}.
Performance characteristics differ substantially as well: \eqwalizer{} delivers near-instantaneous feedback, but at the cost of reduced issue detection. 
This trade-off reflects the fundamental difference between set-theoretic verification and alternative approaches.
\ety{}'s higher computational investment yields stronger guarantees and uncovers type errors that escape faster, less precise analyses.

\section{Related Work}
\label{sec:related-work}

\paragraph{Type checkers for Erlang}

There is a history of work on type systems aimed at formally taming Erlang's type system. 
This effort began with Marlow and Wadler's work~\citeyear{conf/icfp/MarlowW97}; we use their type specification system as a baseline~\citep{conf/hopl/Armstrong07}.
It expands further to \dialyzer{}~\citep{conf/aplas/LindahlS04,dialyzer,conf/ppdp/LindahlS06}, which demonstrated the practical considerations essential for the successful adoption of a static type system within the community. 
Consequently, this may require certain trade-offs.

Several type systems have been developed as research prototypes to study and explore various aspects and approaches within the context of Erlang, such as gradual typing \citep{gradualizer}, type inference in the style Hindley-Milner~\citep{conf/popl/DamasM82}
with overloading for data constructors \citep{conf/erlang/ValliappanH18}, and
bidirectional typing \citep{conf/erlang/RajendrakumarB21, journals/csur/DunfieldK21}.
However, these projects were either limited to a subset of Erlang or required extensive modifications making them impractical for adoption.
Each approach lacked key components necessary to provide a comprehensive solution.
Our system builds upon the formalization of \cite{conf/ifl/SchimpfWB22}, integrating all type aspects that are necessary for practical adoption in real-world use.

In parallel with our work, efforts are underway to introduce a static type system to Elixir, a language closely related to Erlang. This type system is based on set-theoretic types~\citep{journals/programming/CastagnaDV23, castagna2024guardanalysissafeerasure}, similar to \ety{}. It offers features comparable to ours, with a particular focus on gradual typing of strong arrows and a more refined analysis of guarded patterns. Some of these features have already been integrated into Elixir as of version 1.17.\footnote{https://hexdocs.pm/elixir/main/gradual-set-theoretic-types.html}
We view our work as complementary to these efforts. While their approach emphasizes individual features, we aim to assess the feasibility of a comprehensive formal framework and implementation that integrates all aspects. Further, their work develops a new
type language, whereas our type checker builds on the well-established Erlang type language.
A promising direction for future work would be incorporating gradual types
\citep{journals/programming/CastagnaDV23,castagna2024guardanalysissafeerasure} into \ety{}.
This would allow to type check projects where many functions do not carry type specifications.

\paragraph{Set-theoretic types}

In this work presented here, we combine and extend the work based on set-theoretic types~\citep{journals/jacm/FrischCB08,conf/icfp/CastagnaX11,conf/popl/Castagna0XILP14,conf/popl/Castagna0XA15}. 
We extend this body of work by incorporating guarded pattern matching, which enables dynamic conditions on expression types, facilitating type narrowing.
While this method only approximates the set of values captured by a guard, it is sufficient for type narrowing and exhaustiveness checking in pattern matching.

\cite{DBLP:journals/pacmpl/CastagnaLN24} propose a way to leverage the full power of occurrence typing.
However, challenges remain in designing and applying the system in the presence of side effects, and the impact on performance is yet to be fully evaluated in a more comprehensive case study. 
Nonetheless, qualitative and quantitative studies~\citep{conf/erlang/BergerSBW24,conf/ifl/SchimpfWB22} indicate that our approach to type narrowing is sufficient and practical for real-world applications in Erlang.

Similar to \cite{conf/icfp/CastagnaP016} and \cite{DBLP:phd/hal/Petrucciani19}, we perform type checking and type reconstruction
by solving subtyping constraints with the tallying algorithm.
This allows us to infer types for locally defined functions.
We further extend their work with pattern guards and exhaustiveness checking.

\paragraph{Retrofitting static type systems for dynamically-typed languages}

There are various languages that are dynamically typed but were later equipped with a static type system.
For instance, Flow and TypeScript serve as typed extensions for JavaScript~\citep{journals/pacmpl/ChaudhuriVGRL17, conf/popl/RastogiSFBV15}.
Both have been widely adopted in practice, supporting union and intersection types, though negation types are missing.
Similar to Erlang, Python provides a syntax for type annotations\footnote{\url{https://peps.python.org/pep-0484/}}, but does not enforce static type checking.
Tools like Mypy\footnote{\url{https://mypy-lang.org/}} and pyright\footnote{\url{https://github.com/microsoft/pyright}}
attempt to detect type inconsistencies by inferring function parameters in an ad-hoc manner. 
Notably and in contrast to \ety{}, both JavaScript and Python's type systems lack a formal foundation.

MLstruct~\citep{journals/pacmpl/ParreauxC22} supports records, equi-recursive types, and union, intersection, and negation types, with subtyping defined algebraically rather than semantically. 
The language features type inference with principal types. 
\ety{} also includes a limited form of inference, excluding intersection types, but does not guarantee principal types. 
However, several design choices in MLstruct are not well suited for Erlang. 
For example, while Erlang allows arithmetic operators to be overloaded, MLstruct does not support encoding overloading using intersection types. 
Additionally, MLstruct treats a union of a record and a function type as equivalent to the top type, lacking a mechanism to distinguish between the two parts of such a union. 
In contrast, Erlang provides this distinction using the \lstinline|is_function| predicate.

Typed Racket~\citep{conf/popl/Tobin-HochstadtF08} introduces \emph{occurrence typing}, allowing custom type tests that can refine variable types based on the context in which they appear. 
Similar to our approach, the type of a variable can change depending on the type tests it encounters. 
In \ety{}, the set of type tests is fixed following the Erlang language. 
Further, both Erlang and Racket support encoding precise numeric types within their type language~\citep{DBLP:conf/padl/St-AmourTFF12}, though Erlang offers only a subset of the numeric types available in Racket.

\section{Conclusion}
\label{sec:conclusion}

This article presented a formalization of sequential Erlang based on set-theoretic types and semantic subtyping. 
With this work, we make significant progress in advancing static type safety of Erlang code in a feasible and practical manner.
On the theory, we have successfully proven type soundness and decidability of type checking. 
Our system generates informative and user-friendly error messages, aiding developers in diagnosing and resolving type-related issues more effectively.
While we have extended the type checker to support maps and records, addressing one of the limitations mentioned by \cite{conf/ifl/SchimpfWB22}, we need to leave out a detailed presentation due to space limitations. 
The performance of the type checker has been substantially improved, surpassing the state-of-the-art performance of CDuce in certain cases, particularly for the internal constraint solver. 
We have also implemented type inference for local functions and type checking for programs consisting of multiple modules, significantly enhancing the system's usability and expressiveness. 

However, several challenges and opportunities for future work remain. 
Some language constructs, such as try-catch and opaque types, are not yet supported. 
Additionally, as Erlang continues to evolve, new constructs like maybe expressions and map comprehensions are being introduced, which could benefit greatly from precise typing using set-theoretic types. 
The refinement of types for certain language constructs and comparison operators can also be improved to enhance accuracy and expressiveness. 
While type inference for local functions is now supported, inference for non-local functions remains slow, requiring further research to optimize performance. 
Lastly, although the tally algorithm has become more efficient, solving large constraint sets remains a challenge, and additional research is needed to improve its scalability and efficiency. 

The migration of larger programs to static type checking typically occurs incrementally, meaning that some parts of the codebase are statically checked while others remain dynamically typed. 
Given the progress we have made, gradual typing as outlined by \cite{castagna2024guardanalysissafeerasure} would be a valuable extension to our system.
Gradual typing allows statically and dynamically typed code to coexist seamlessly, which is particularly useful as it is common that modules consist of both exported and typed functions, and private untyped ones.
It could also provide a practical solution for parts of the program that are inherently difficult to type statically, such as dynamically loaded code or constructs that are not yet fully supported.


\section*{Conflict of interest}
None.

\bibliographystyle{ACM-Reference-Format}
\bibliography{references}

\newpage
\renewcommand\thesection{\Alph{section}}
\appendix

\section{Type Safety: Declarative System}
\label{sec:appendix-type-safety-declarative-system}

This appendix proves that the declarative type system defined in \Cref{sec:formal-type-system} is sound.

\begin{definition}[Free type variables] 
  \label{def:free-tyvars}
  The type variables free in type $t$, written $\TyVars{t}$, are coinductively defined as
  $\TyVarsFunc{0}{\Monoty}{\emptyset}$:
  \begin{align*}
    \TyVarsFunc{0}{\Monoty}{\mathcal{T}} &= \begin{cases}
      \emptyset & \text{if } \Monoty \in \mathcal{T} \\
      \TyVarsFunc{1}{\Monoty}{\mathcal{T} \cup \{ \Monoty \}} & \text{otherwise}
    \end{cases}\\
    \TyVarsFunc{1}{\Tyvar}{\mathcal{T}} &= \{ \Tyvar \}\\
    \TyVarsFunc{1}{\Basety}{\mathcal{T}} &= \emptyset\\
    \TyVarsFunc{1}{\Pairty{\Monoty_1}{\Monoty_2}}{\mathcal{T}} &= \TyVarsFunc{0}{\Monoty_1}{\mathcal{T}} \cup \TyVarsFunc{0}{\Monoty_2}{\mathcal{T}}\\
    \TyVarsFunc{1}{\Monoty_1 \to \Monoty_2}{\mathcal{T}} &= \TyVarsFunc{0}{\Monoty_1}{\mathcal{T}} \cup \TyVarsFunc{0}{\Monoty_2}{\mathcal{T}}\\
    \TyVarsFunc{1}{\Monoty_1 \Union \Monoty_2}{\mathcal{T}} &= \TyVarsFunc{1}{\Monoty_1}{\mathcal{T}} \cup \TyVarsFunc{1}{\Monoty_2}{\mathcal{T}}\\
    \TyVarsFunc{1}{\neg \Monoty_1}{\mathcal{T}} &= \TyVarsFunc{1}{\Monoty_1}{\mathcal{T}}
  \end{align*}

  \begin{itemize}
  \item The type variables free in some type scheme are defined as
  $\TyVars{\TyScm{\TyvarSet} \Monoty} = \TyVars{\Monoty} \setminus \TyvarSet$.

  \item The type variables free in some type substitution are
  $\FreeTyVars{[t_i/\alpha_i \mid i \in I]} = \medcup\nolimits_{i \in I}\FreeTyVars{t_i}$.

  \item We extend the definition of $\FreeTyVarsName$ to other constructs in the obvious way.

  \item We say a type $\Monoty$ is closed if $\TyVars{\Monoty}$ is empty.
  \end{itemize}
\end{definition}

\begin{property}[Projections of product types]
  \label{prop:projection}
  For some type $\Monoty \IsSubty \Pairty{\top}{\top}$, there exist projections $\ProjL{\Monoty}$ and $\ProjR{\Monoty}$
  such that
  \begin{EnumThm}
  \item $\Monoty \IsSubty \Pairty{\ProjL{\Monoty}}{\ProjR{\Monoty}}$;
  \item if $\Monoty \IsSubty \Pairty{\Monoty_1}{\Monoty_2}$ then $\ProjI{\Monoty} \IsSubty \Monoty_i$;
  \item if $\Monoty \IsSubty \Monoty' \IsSubty \Pairty{\top}{\top}$ then $\ProjI{\Monoty} \IsSubty \ProjI{\Monoty'}$;
  \item for all type substitutions $\Tysubst$, $\ProjI{\Monoty\Tysubst} \IsSubty \ProjI{\Monoty}{\Tysubst}$.
  \end{EnumThm}

  See \citet[Property A.31, page 26]{journals/corr/CastagnaP016}.
\end{property}

\begin{property}[Properties of subtyping]
  \label{prop:subty}
  For types $\Monoty_1, \Monoty_2$:
  \begin{EnumThm}
  \item $\ErlBot \IsSubty t \IsSubty \ErlTop$ for any $t$.
  \item If $t_1 \to t_2 \IsSubty t_1' \to t_2'$ then
    $t_1' \IsSubty t_1$ and $t_2 \IsSubty t_2'$.
  \item If $\Pairty{t_1}{t_2} \IsSubty \Pairty{t_1'}{t_2'}$ then
    $t_1 \IsSubty t_1'$ and $t_2 \IsSubty t_2'$.
  \item If $t_1 \IsSubty t_2$ then $t_1\Tysubst \IsSubty t_2\Tysubst$ for any type substitution $\Tysubst$.
  \end{EnumThm}
\end{property}
\begin{proof}
  Claims (i), (ii), and (iii) follow from the interpretation of types as sets as values,
  see \cite{journals/jacm/FrischCB08}.
  For (iv) see \cite[Lemma~3.15, page~10]{conf/icfp/CastagnaX11}.
\end{proof}

\begin{definition}[Bound pattern variables]
  The expression variables bound by a pattern, written $\PatBound{p}$, are defined as
  \begin{mathpar}
    \PatBound{p} =
    \begin{cases}
      \{x\}  & \textrm{if}~ p = x \\
      \PatBound{p_1} \cup \PatBound{p_2} & \textrm{if}~ p = (p_1, p_2)\\
      \emptyset & \textrm{otherwise}
    \end{cases}
  \end{mathpar}

  The expression variables bound by a guarded pattern are $\PatBound{\WithGuard p g} = \PatBound{p}$.
\end{definition}

\begin{definition}[Free expression variables]
  \label{def:free-exp}
  The expression variables free in some expression $e$, written $\FreeE{e}$, are defined as
    \begin{align*}
      \FreeE{x} & = \{ x \} \\
      \FreeE{\Const} & = \emptyset \\
      \FreeE{\Abs{x} e} &= \FreeE{e} \setminus \{x\}\\
      \FreeE{\App{e_1}{e_2}} &= \FreeE{e_1} \cup \FreeE{e_2}\\
      \FreeE{\Pair{e_1}{e_2}} &= \FreeE{e_1} \cup \FreeE{e_2}\\
      \FreeE{\Case{e}{\Multi{(\PatCls{\Pg_i}{e_i})}}}
                    &= \FreeE{e} \cup \medcup_{i \in I}(\FreePg{\Pg_i} \cup (\FreeE{e_i} \setminus \PatBound{\Pg_i}))
    \end{align*}
  The expression variables free in some pattern $p$, written $\FreeP{p}$, are defined as
    \begin{align*}
      \FreeP{v} & = \FreeE{v} \\
      \FreeP{\Wildcard} & = \emptyset \\
      \FreeP{x} & = \emptyset \\
      \FreeP{\Capt{x}} & = \{x\} \\
      \FreeP{(p_1, p_2)} & = \FreeP{p_1} \cup \FreeP{p_2}
    \end{align*}
  The expression variables free in some guard $g$, written $\FreeG{g}$, are defined as
    \begin{align*}
      \FreeG{\GuardTrue} & = \emptyset \\
      \FreeG{\GuardOracle} & = \emptyset \\
      \FreeG{\Is{\Gt}{v}} & = \FreeE{v} \\
      \FreeG{\Is{\Gt}{x}} & = \{x\} \\
      \FreeG{\GuardAnd{g_1}{g_2}} & = \FreeG{g_1} \cup \FreeG{g_2}
    \end{align*}
  The expression variables free in some guarded pattern $\Pg$, written $\FreePg{\Pg}$, are defined as
  $
    \FreePg{\WithGuard p g} = \FreeP p \cup (\FreeG g \setminus \PatBound{p})
  $.\\
  Most of the time, we omit the subscripts in
  $\FreeSym{e}$, $\FreeSym{g}$, $\FreeSym{p}$, $\FreeSym{pg}$ and simply write $\FreeSym{}$.

\end{definition}

\begin{lemma}[Elimination of subsumption]
  \label{lem:sub-elim}
  If $\ExpTy{\Venv}{e}{\Monoty}$, then there exist a derivation for
  $\ExpTy{\Venv}{e}{\Monoty'}$ not ending with rule \Rule{d-sub} such that $\Monoty' \IsSubty \Monoty$.
\end{lemma}
\begin{proof}
  By straightforward induction on the typing derivation.
\end{proof}

\begin{definition}[Subtyping for environments]
  We write $\Venv \IsSubty \Venv'$ to denote that
  $\Dom{\Venv} = \Dom{\Venv'}$ and it holds that
  $\Venv(x) \IsSubty \Venv'(x)$ for all $x \in \Dom{\Venv}$.
\end{definition}

\begin{lemma}[Weakening]
  \label{lem:weakening}
  If $\ExpTy{\Venv}{e}{\Monoty}$  and $\Venv' \IsSubty \Venv$, then
  $\ExpTy{\Venv'}{e}{\Monoty}$.
\end{lemma}
\begin{proof}
  By straightforward induction on the derivation of $\ExpTy{\Venv}{e}{\Monoty}$.
\end{proof}

\begin{lemma}[Extra variable]
  \label{lem:extra-var}
  If $\ExpTy{\Venv}{e}{\Monoty}$ and $x \notin \Dom{\Senv} \cup \Dom{\Venv}$, then
  $\ExpTy{\Venv, x: t'}{e}{\Monoty}$ for any $t'$.
\end{lemma}

\begin{proof}
  By straightforward induction on the derivation of $\ExpTy{\Venv}{e}{\Monoty}$.
\end{proof}

\begin{lemma}[Value shapes]
  \label{lem:value-type-determines-shape}
  Let $v$ be a value.
  \begin{EnumThm}
  \item If $\ExpTy{\Venv}{v}{\TyOfConst{\Const}}$ for some $\Const$, then $v = \Const'$ for some $\Const'$ with
    $\TyOfConst{\Const} = \TyOfConst{\Const'}$ and $\Const = \Const'$ if $\Const$ is an int.

  \item If $\ExpTy{\Venv}{v}{\Monoty_1 \to \Monoty_2}$, then $v$ is of the form
    $\Abs{x} e$ and $\ExpTy{\Venv,x:\Monoty_1}{e}{\Monoty_2}$.

  \item If $\ExpTy{\Venv}{v}{\Pairty{\Monoty_1}{\Monoty_2}}$, then $v$ is of the form
    $\Pair{v_1}{v_2}$, $\ExpTy{\Venv}{v_1}{\Monoty_1}$, and $\ExpTy{\Venv}{v_2}{\Monoty_2}$.
  \end{EnumThm}
\end{lemma}

\begin{proof}
  \begin{EnumThm}
  \item If $\ExpTy{\Venv}{v}{\TyOfConst{\Const}}$, then by \Cref{lem:sub-elim}
    $\ExpTy{\Venv}{v}{\Monoty}$ ending with rule \Rule{d-const} and $\Monoty \IsSubty \TyOfConst{\Const}$.
    \begin{itemize}
    \item If $\Const$ is an int, then $\Monoty = \TyOfConst{\Const} = \Const$, and so $v = \Const$.
    \item If $\Const$ is a float, then $t = \Float$, so $v$ must also be a float constant.
    \end{itemize}
  \item If $\ExpTy{\Venv}{v}{\Monoty_1 \to \Monoty_2}$, then
    $\ExpTy{\Venv}{v}{\Monoty_1' \to \Monoty_2'}$ with the last applied rule being \Rule{d-abs} and
    $\Monoty_1' \to \Monoty_2'  \IsSubty \Monoty_1 \to \Monoty_2$ by \Cref{lem:sub-elim}.
    Then $v = \Abs{x}{e}$ and from the premise of the rule
    $\ExpTy{\Venv, x:\Monoty_1'}{e}{\Monoty_2'}$.
    We get $\ExpTy{\Venv, x:\Monoty_1}{e}{\Monoty_2}$ with \Cref{prop:subty}(ii), \Cref{lem:weakening}, and
    rule \Rule{d-sub}.
  \item If $\ExpTy{\Venv}{v}{\Pairty{\Monoty_1}{\Monoty_2}}$, then the claim follows similar as in the
    preceding case with \Cref{prop:subty}(iii).
  \end{EnumThm}
\end{proof}

\begin{lemma}
  \label{lem:value-u-or-notu}
  Assume $\ExpTy{\Venv}{v}{\Monoty}$.
  Then for any type $u$, either $\ExpTy{\Venv}{v}{u}$ or $\ExpTy{\Venv}{v}{\Neg u}$.
\end{lemma}
\begin{proof}
  Follows from the interpretation of types as sets of values and the fact
  that, for any type u, the union of the interpretation of $u$ and $\Neg u$ is the full set of values.
  See \citet[Lemma~6.22, page~30]{journals/jacm/FrischCB08}.
\end{proof}

\begin{lemma}
  \label{lem:value-not-bot}
  For any $\Senv$ and $\Venv$,
  there exists no value $v$  with $\ExpTy{\Venv}{v}{\ErlBot}$.
\end{lemma}
\begin{proof}
  See \citet[Lemma~6.19]{journals/jacm/FrischCB08}.
  We prove the implication \enquote{\emph{if $\ExpTy{\Venv}{v}{t}$, then $t \neq \ErlBot$}} by induction on the derivation
  of $\ExpTy{\Venv}{v}{t}$. If the derivation ends with one of the rules \Rule{d-const}, \Rule{d-abs}, or \Rule{d-pair},
  the claim is obvious. For rule \Rule{d-sub}, the claim follows from the \IH{} The derivation cannot end
  with any of the four other rules because $v$ is a value.
\end{proof}

\begin{lemma}
  \label{lem:val-match-implies-ty}
  If $\ExpTy{\Venv}{v}{t}$ for some $t$ and $v \ValMatches \GuardTy$,
  then $\ExpTy{\Venv}{v}{\TyOf{\GuardTy}}$.
\end{lemma}
\begin{proof}
  \begin{CaseDistinction}{on the form of $\Gt$}
    \CdCase{$\Gt = \Int$} Then $v$ must be an int constant, so the claim follows with \Rule{c-const}.
    \CdCase{$\Gt = \Float$} Then $v$ must be a flaot constant, so the claim follows with \Rule{c-const}.
    \CdCase{$\Gt = \AnyPair$} Then $v = (v_1, v_2)$. From assumption $\ExpTy{\Venv}{v}{t}$ with
    \Cref{lem:sub-elim}, then $\ExpTy{\Venv}{v}{\Pairty{t_1}{t_2}}$ for some $t_1,t_2$.
    Now, with rule \Rule{d-sub}
    \begin{igather}
      \ExpTy{\Venv}{v}{\BraceBelow{\TyOf{\AnyPair}}{{} = \Pairty{\ErlTop}{\ErlTop}}}
    \end{igather}
    \CdCase{$\Gt = \AnyFun$} Analogously to the preceding case.
  \end{CaseDistinction}
\end{proof}

\begin{lemma}[Totality of dynamic semantics for pattern matching]
  \label{lem:dyn-match-total}
  ~
  \begin{EnumThm}
  \item If $\Free{p} \subseteq \Dom{\FunEnv}$, then $\PatSubst{v}{p}$ is
    defined; that is, either $\PatSubstFail$ or some $\Valsubst$.
  \item If $\Free{g} \subseteq \Dom{\FunEnv}$, then $\EvalGuard{g}$
    is defined; that is, either $\True$ or $\False$.
  \item If $\Free{\Pg} \subseteq \Dom{\FunEnv}$, then $\PatSubst{v}{\Pg}$
    is defined; that is, either $\PatSubstFail$ or some $\Valsubst$.
  \end{EnumThm}
\end{lemma}
\begin{proof}
  Claim (i) follows by induction on $p$, claim (ii) by induction on $g$.
  Claim (iii) follows from (i) and (ii).
\end{proof}

\begin{lemma}
  \label{lem:pat-vars}
  ~
  \begin{EnumThm}
  \item If $\PatSubst{v}{p} = \Valsubst$, then $\Dom{\Valsubst} = \PatBound{p}$.
  \item If $\PatSubst{v}{(\WithGuard{p}{g})} = \Valsubst$, then $\Dom{\Valsubst} = \PatBound{p}$.
  \item $\PatBound{p} = \Dom{\PatEnv{t}{p}}$ for any $t$.
  \item $\PatBound{\Pg} \subseteq \Dom{\PatEnv{t}{\Pg}}$ for any $t$.
  \end{EnumThm}
\end{lemma}
\begin{proof}
  Claim (i) follows by induction on $p$. Claim (ii) follows by (i).
  Claim (iii) follows by induction on $p$. Claim (iv) with (iii).
\end{proof}

\begin{lemma}
  \label{lem:pat-types-closes}
  ~
  \begin{EnumThm}
  \item $\Env{g}$ is closed.
  \item $\PatTy{p}{\Venv}$ is closed if $\Venv$ is closed.
  \item $\PgTy{\Pg}{\DirPatTy}$ is closed.
  \end{EnumThm}
\end{lemma}
\begin{proof}
  Claim (i) follows by induction on $g$, claim (ii) by induction on $p$. Claim (iii) follows from (i) and (ii).
\end{proof}

\begin{lemma}
  \label{lem:free-exp-vars}
  If $\ExpTy{\Venv}{e}{t}$ then $\Free{e} \subseteq \Dom{\Senv} \cup \Dom{\Venv}$
\end{lemma}
\begin{proof}
  Straightforward induction on $\ExpTy{\Venv}{e}{t}$.
\end{proof}

\begin{lemma}
  \label{lem:val-matches-gt}
  If $\ExpTy{\Venv}{v}{t}$ and $t \IsSubty \TyOf{\GuardTy}$ then $v \ValMatches \GuardTy$.
\end{lemma}
\begin{proof}
  Follows from inspection all possible combinations of guard types $\GuardTy$ ($\Int$, $\Float$,
  $\AnyPair$, $\AnyFun$) and shapes of values (int constant, float constant, pair, lambda).
\end{proof}

\begin{lemma}
  \label{lem:val-matches-subst}
  $v \ValMatches \GuardTy$ if, and only if, $v\Valsubst \ValMatches \GuardTy$ for all $\Valsubst$.
\end{lemma}
\begin{proof}
  Straightforward induction on $v$. Note that $v$ and $v\Valsubst$ have the same shape
  on the outermost level.
\end{proof}

\begin{lemma}
  \label{lem:env-g-subst}
  $\Env{g\Valsubst} = \{x : \Env{g}(x) \mid x \in \Dom{\Env{g}} \setminus \Dom{\Valsubst} \}$
\end{lemma}
\begin{proof}
  Straightforward induction on $g$ with $\Env{\Is{\GuardTy}{v}} = \EmptyEnv$ for any value $v$.
\end{proof}

\begin{lemma}
  \label{lem:pg-env-subst}
  Let $\Venv = \PatEnv{t}{\Pg}$ and $\Venv' = \PatEnv{t}{(\Pg\Valsubst)}$.
  Then $\Venv' = \{x : \Venv(x) \mid x \in \Dom{\Venv} \setminus \Dom{\Valsubst}\}$.
\end{lemma}
\begin{proof}
  Assume $\Pg = \WithGuard{p}{g}$. \Wlog{}, $\PatBound{p}$ fresh, so
  $\PatBound{p} \cap \Dom{\Valsubst} = \emptyset$.
  Then $\Pg \Valsubst = \WithGuard{p\Valsubst}{g\Valsubst}$. The claim then follows
  with \Cref{lem:env-g-subst}.
\end{proof}

\begin{lemma}
  \label{lem:value-has-both-types}
  If $\ExpTy{\Venv}{v}{\Monoty}$ and $\ExpTy{\Venv}{v}{\Monoty'}$, then $\ExpTy{\Venv}{v}{\Monoty \Inter \Monoty'}$.
\end{lemma}
\begin{proof}
  We generalize the claim and prove that
  $\ExpTy{\Venv}{e}{\Monoty}$ and $\ExpTy{\Venv}{e}{\Monoty'}$ imply $\ExpTy{\Venv}{e}{\Monoty \Inter \Monoty'}$.
  The proof is by induction on the typing derivations, see \citet[Lemma~6.15, page~27]{journals/jacm/FrischCB08}.
\end{proof}

\begin{lemma}
  \label{lem:safe-implies-eval-guard}
  Assume $\ExpTy{\Venv}{v}{\PatLowerTy{p}{\Env{g}}}$.
  If $\PatSubst{v}{p} = \Valsubst$ and $\SafeDown{g}{\PatBound{p}}$, then $\EvalGuard{g \Valsubst} = \True$.
\end{lemma}
\begin{proof}
  From $\SafeDown{g}{\PatBound{p}}$, we know that $g$ is a conjunction of type tests
  $(\Is{\Gt_i}{x_i})_{i \in I}$ and $(\Is{\Gt_j'}{v_j})_{j \in J}$ such that
  $x_i \in \PatBound{p}$ for all $i \in I$ and $v_j \ValMatches \Gt_j'$ for all $j \in J$.
  From \Cref{lem:pat-vars}, we have $\PatBound{p} = \Dom{\Valsubst}$.
  With $v_j \ValMatches \Gt_j'$, we also have $v_j\Valsubst \ValMatches \Gt_j'$ by \Cref{lem:val-matches-subst}.
  To establish $\EvalGuard{g\Valsubst} = \True$, we then only need to show
  \begin{igather}
    \Valsubst(x_i) \ValMatches \Gt_i ~\textrm{for all}~i \in I
  \end{igather}
  We proceed by induction on $p$.
  \begin{CaseDistinction}{on the form of $p$}
    \CdCase{$p = v'$ or $p = \Wildcard$ or $p = \Capt{y}$} Either $\PatSubst{v}{p} = \PatSubstFail$ (impossible case) or
    $\PatSubst{v}{p} = []$, so $I = \emptyset$.

    \CdCase{$p = x$} Then $\Valsubst = [v/x]$ and $x_i = x$ for all $i \in I$.
    Assume $I \neq \emptyset$, otherwise the claim holds trivially.

    We have $\ExpTy{\Venv}{v}{\Env{g}(x)}$ from the
    assumption $\ExpTy{\Venv}{v}{\PatLowerTy{p}{\Env{g}}}$. With \Cref{lem:value-not-bot} $\Env{g}(x) \neq \ErlBot$,
    so all $\Gt_i$ are identical and $\Env{g}(x) = \TyOf{\Gt_1}$. From \Cref{lem:val-matches-gt} then
    $v \ValMatches \Gt_1$. Thus, $\Valsubst(x_i) \ValMatches \Gt_i$ for all $i \in I$ as required.

    \CdCase{$p = (p_1, p_2)$} Then $v = (v_1, v_2)$ and
    \begin{igather}
      \PatSubst{v_i}{p_i} = \Valsubst_i \quad (i = 1,2) \QLabel{eq:vi-pi}\\
      \Valsubst_1 \ValEq \Valsubst_2\\
      \Valsubst = \Valsubst_1 \cup \Valsubst_2
    \end{igather}
    Also, we have $\PatLowerTy{p}{\Env{g}} = \Pairty{\PatLowerTy{p_1}{\Env{g}}}{\PatLowerTy{p_2}{\Env{g}}}$,
    so with the assumption $\ExpTy{\Venv}{v}{\PatLowerTy{p}{\Env{g}}}$ and \Cref{lem:value-type-determines-shape}
    \begin{igather}
      \ExpTy{\Venv}{v_i}{\PatLowerTy{p_i}{\Env{g}}} \quad (i = 1,2) \QLabel{eq:ty-vi}
    \end{igather}
    Now construct guard $g_1$ as the conjunction of type tests $\Is{\Gt_i}{x_i}$, but only
    taking those with $x_i \in \PatBound{p_1}$. Construct guard $g_2$ analogously. Then
    \begin{igather}
      \SafeDown{g_i}{\PatBound{p_i}} \quad (i = 1,2) \QLabel{eq:safe-pi}
    \end{igather}
    Now use \QRef{eq:ty-vi}, \QRef{eq:vi-pi}, and \QRef{eq:safe-pi} to apply the \IH{}
    This yields $\EvalGuard{g_i\Valsubst_i}$ for $i = 1,2$.
    With $\PatBound{p} = \PatBound{p_1} \cup \PatBound{p_2}$, we know that each
    type test $\Is{\Gt_i}{x_i}$ is contained in either $g_1$ or $g_2$. By construction
    of $\Valsubst$, we then have $\Valsubst(x_i) \ValMatches \Gt_i$ for all $i \in I$.
  \end{CaseDistinction}
\end{proof}

\begin{lemma}
  \label{lem:eval-guard-implies-safe}
  If $\EvalGuard{g\Valsubst} = \True$ then $\SafeUp{g} = \True$.
\end{lemma}
\begin{proof}
  Induction on $g$. The case $g = \GuardAnd{g_1}{g_2}$ directly follows from the \IH{}, the cases
  $g = \GuardTrue$ and $g = \GuardOracle$ and $g = \Is{\Gt}{x}$ hold trivially.
  Thus assume $g = \Is{\Gt}{v}$. For this case,
  $\SafeUp{g} = \False$ only if not $v \ValMatches \Gt$.
  From \Cref{lem:val-matches-subst} we then get that $v\Valsubst \ValMatches \Gt$ does not
  hold either. Thus, $\EvalGuard{g\Valsubst}$ must be $\False$, which is a contradiction.
  Thus, $\SafeUp{g} = \True$ as required.
\end{proof}

\begin{lemma}
  \label{lem:eval-guard-implies-not-bot}
  If $\EvalGuard{g\Valsubst} = \True$ then $\Env{g} \neq \ErlBot$.
\end{lemma}
\begin{proof}
  Assume $\Env{g} = \ErlBot$. Thus, there exists $x \in \Dom{\Env{g}}$ with $\Env{g}(x) = \ErlBot$.
  Hence, $g$ contains (at least) two type tests $\Is{\Gt_1}{x}$ and $\Is{\Gt_2}{x}$ with
  $\Gt_1 \neq \Gt_2$ because only types resulting from merges of environments may result in $\ErlBot$.
  \begin{itemize}
  \item $x \in \Dom{\Valsubst}$ with $\Valsubst(x) = v$. From $\EvalGuard{g\Valsubst} = \True$
    then $v \ValMatches \Gt_1$ and $v \ValMatches \Gt_2$. This is impossible since $\Gt_1 \neq \Gt_2$.
  \item $x \notin \Dom{\Valsubst}$. Then from $\EvalGuard{g\Valsubst} = \True$ we have $\FunEnv(x) = v$
    and $v \ValMatches \Gt_1$ and $v \ValMatches \Gt_2$. This is impossible since $\Gt_1 \neq \Gt_2$.
  \end{itemize}
\end{proof}

\begin{lemma}
  \label{lem:gp-typing-helper}
  Assume $\ExpTy{\Venv}{v}{t}$.
  If $\PatSubst{v}{p} = \Valsubst$ and $\EvalGuard{g\Valsubst} = \True$ and $\Env{g} \neq \ErlBot$,
  then $\ExpTy{\Venv}{v}{\PatUpperTy{p}{\Env{g}}}$.
\end{lemma}
\begin{proof}
  By induction on $p$. For the induction to go through, we need to generalize the claim as follows:
  \begin{quote}\em\raggedright
    Assume $\ExpTy{\Venv}{v}{t}$.
    If $\PatSubst{v}{p} = \Valsubst$ and $\EvalGuard{g\Valsubst'} = \True$ for some $\Valsubst'$
    with $\Valsubst \subseteq \Valsubst'$ and $\Env{g} \neq \ErlBot$,
    then $\ExpTy{\Venv}{v}{\PatUpperTy{p}{\Env{g}}}$.
  \end{quote}
  \begin{CaseDistinction}{on the form of $p$}
    \CdCase{$p = v'$} Then $v = v'$ because of $\PatSubst{v}{p} = \Valsubst$.
    If $v = \Const$, then $\PatUpperTy{p}{\Env{g}} = \TyOf{\Const}$, so the claim follows with rule
    \Rule{d-const}.
    Otherwise, $\PatUpperTy{p}{\Env{g}} = \ErlTop$, so the claim follows because $v$ is well-typed.

    \CdCase{$p = \Wildcard$}
    Then $\PatUpperTy{p}{\Env{g}} = \ErlTop$, so the claim is obvious.

    \CdCase{$p = x$} Then $\Valsubst = [v/x]$. If $x \notin \Dom{\Env{g}}$ then
    $\PatUpperTy{p}{\Env{g}} = \ErlTop$, so the claim is obvious.

    Now assume $x \in \Dom{\Env{g}}$, so $g$ contains at least one type test on $x$.
    With $\Env{g} \neq \ErlBot$, we know that
    all type tests in $g$ on $x$ are of the form $\Is{\GuardTy}{x}$.
    Hence,
    \begin{igather}
      \PatUpperTy{p}{\Env{g}} = \Env{g}(x) = \TyOf{\GuardTy}
    \end{igather}

    From $\EvalGuard{g\Valsubst'} = \True$ and $[v/x] = \Valsubst \subseteq \Valsubst'$,
    we have $\Valsubst(x) \ValMatches \GuardTy$, so $v \ValMatches \GuardTy$. With
    \Cref{lem:val-match-implies-ty} then $\ExpTy{\Venv}{v}{\TyOf{\GuardTy}}$ as required.

    \CdCase{$p = \Capt{y}$}
    Then $\PatUpperTy{p}{\Env{g}} = \ErlTop$, so the claim is obvious.

    \CdCase{$p = (p_1, p_2)$} From $\PatSubst{v}{p} = \Valsubst$ we get
    \begin{igather}
      v = (v_1, v_2) \\
      \PatSubst{v_i}{p_i} = \Valsubst_i \quad (i=1,2) \\
      \Valsubst_1 \ValEq \Valsubst_2 \\
      \Valsubst = \Valsubst_1 \cup \Valsubst_2
    \end{igather}
    With assumption $\Valsubst \subseteq \Valsubst'$ we have $\Valsubst_i \subseteq \Valsubst'$
    for $i=1,2$. Thus, applying the \IH{} yields for $i=1,2$:
    \begin{igather}
      \ExpTy{\Venv}{v_i}{\PatUpperTy{p_i}{\Env{g}}}
    \end{igather}
    Noting that $\PatUpperTy{p}{\Env{g}} = \Pairty{\PatUpperTy{p_1}{\Env{g}}}{\PatUpperTy{p_2}{\Env{g}}}$
    and applying rule \Rule{d-pair} yields
    $\ExpTy{\Venv}{v}{\PatUpperTy{p}{\Env{g}}}$ as required.
  \end{CaseDistinction}
\end{proof}

\begin{lemma}[Characteristics of potential and accepting types]
  \label{lem:gp-typing}
  Assume $\ExpTy{\Venv}{v}{t}$ for some $t$.
  \begin{EnumThm}
  \item If $\ExpTy{\Venv}{v}{\PgLowerTy{\Pg}}$ then $\PatSubst{v}{\Pg} = \Valsubst$ for some $\Valsubst$.
  \item If $\PatSubst{v}{\Pg} = \Valsubst$ for some $\Valsubst$ then $\ExpTy{\Venv}{v}{\PgUpperTy{\Pg}}$.
  \item If $\PatSubst{v}{\Pg} = \PatSubstFail$ then $\ExpTy{\Venv}{v}{\Neg\PgLowerTy{\Pg}}$.
  \end{EnumThm}
\end{lemma}

\begin{proof}
  \renewcommand\LabelQualifier{lem:gp-typing}
  Claim (iii) follows with (i) and \Cref{lem:value-u-or-notu}.

  \noindent \emph{Proof of claim (i).} Let $\Pg = \WithGuard{p}{g}$. Then with \Cref{lem:value-not-bot}
  $\PgLowerTy{\Pg} \neq \ErlBot$. By definition, we have
  \begin{igather}
    \SafeDown{g}{\PatBound{p}} = \True \\
    \Env{g} \neq \ErlBot \\
    p~\textrm{linear} \\
    \PgLowerTy{\Pg} = \PatLowerTy{p}{\Env{g}} \neq \ErlBot
  \end{igather}
  We proceed by induction on $p$
  \begin{CaseDistinction}{on the shape of $p$}
    \CdCase{$p = v'$}
    Then,
    \begin{igather}
      \PgLowerTy{\Pg} = \PatLowerTy{p}{\Env{g}} = \Const \\
      v' = \Const \\
      \TyOf{\Const} = \Const
    \end{igather}
    by definition of $\PgLowerTy{\Pg}$.
    Hence, $\ExpTy{\Venv}{v}{\Const}$ by assumption. Since only ints are singleton types,
    $\Const$ must be an int, so $v = \Const$ by \Cref{lem:value-type-determines-shape}(i).
    Thus, $\PatSubst{v}{p} = []$ by \Rule{match-val}.
    With \Cref{lem:safe-implies-eval-guard} $\EvalGuard{g} = \True$.
    Thus, $\PatSubst{v}{\Pg} = []$ by \Rule{match-true}.

    \CdCase{$p = \Wildcard$}
    Then $\PatSubst{v}{p} = []$ by \Rule{match-wild}.
    With \Cref{lem:safe-implies-eval-guard} again $\EvalGuard{g} = \True$,
    so the claim holds by \Rule{match-true}.

    \CdCase{$p = x$}
    Then $\PatSubst{v}{p} = [v/x]$ by \Rule{match-var}.
    With \Cref{lem:safe-implies-eval-guard} again $\EvalGuard{g} = \True$,
    so the claim holds by \Rule{match-true}.

    \CdCase{$p = \Capt{y}$} Impossible, since $\PatLowerTy{\Capt{y}}{\Env{g}} = \ErlBot$.

    \CdCase{$p = (p_1, p_2)$} Then,
    \begin{igather}
      \PatLowerTy{p}{\Env{g}} = \Pairty{\PatLowerTy{p_1}{\Env{g}}}{\PatLowerTy{p_2}{\Env{g}}}
    \end{igather}
    With \Cref{lem:value-type-determines-shape}(iii) we have
    \begin{igather}
      v = (v_1, v_2) \\
      \ExpTy{\Venv}{v_i}{\PatLowerTy{p_i}{\Env{g}}}\quad i=1,2 \QLabel{eq:ty-vi}
    \end{igather}
    Pattern $p$ is linear, so $\PatBound{p_1} \cap \PatBound{p_2} = \emptyset$ and
    $p_1,p_2$ are linear. Further, from $\SafeDown{g}{\PatBound{p}} = \True$, we know that $g$ only
    contains type tests for variables in $\PatBound{p}$.

    Construct $g_1$ by replacing type tests for variables in $\PatBound{p_2}$ with $\GuardTrue$.
    Construct $g_2$ analogously. Then for $i=1,2$:
    \begin{igather}
      \SafeDown{g_i}{\PatBound{p_i}} \\
      \Env{g_i} \neq \ErlBot
    \end{igather}
    Define $\Pg_i = \WithGuard{p_i}{g_i}$. Then also for $i=1,2$:
    \begin{igather}
      \PgLowerTy{\Pg_i} = \PatLowerTy{p_i}{\Env{g_i}} = \PatLowerTy{p_i}{\Env{g}}
    \end{igather}
    With \QRef{eq:ty-vi}, it follows for $i = 1,2$:
    \begin{igather}
      \ExpTy{\Venv}{v_i}{\PgLowerTy{\Pg_i}}
    \end{igather}
    Applying the \IH{} now gives us $\PatSubst{v_i}{\Pg_i} = \Valsubst_i$ for
    some $\Valsubst_i$ and $i=1,2$. Inverting rule \Rule{match-true} yields for $i=1,2$
    \begin{igather}
      \PatSubst{v_i}{p_i} = \Valsubst_i \\
      \EvalGuard{g_i\Valsubst_i} = \True
    \end{igather}

    With \Cref{lem:pat-vars}(ii) $\Dom{\Valsubst_i} = \PatBound{p_i}$.
    From $\PatBound{p_1} \cap \PatBound{p_2} = \emptyset$ we get $\Valsubst_1 \ValEq \Valsubst_2$.
    Hence, $\PatSubst{v}{p} = \Valsubst_1 \cup \Valsubst_2$ by rule \Rule{match-pair}. Further,
    by construction of $g_i$ we also have $\EvalGuard{g\Valsubst} = \True$ for
    $\Valsubst = \Valsubst_1 \cup \Valsubst_2$. Thus, $\PatSubst{v}{\Pg} = \Valsubst$ by
    rule \Rule{match-true}.
  \end{CaseDistinction}

  \noindent \emph{Proof of claim (ii).} Assume $\Pg = \WithGuard{p}{g}$. Inverting rule \Rule{match-true}
  yields $\PatSubst{v}{p} = \Valsubst$ and $\EvalGuard{g\Valsubst} = \True$.
  With \Cref{lem:eval-guard-implies-safe} and~\Cref{lem:eval-guard-implies-not-bot} then
  \begin{igather}
    \Env{g} \neq \ErlBot \\
    \SafeUp{g} = \True
  \end{igather}
  Then, $\PgUpperTy{\Pg} = \PatUpperTy{p}{\Env{g}}$. With \Cref{lem:gp-typing-helper} then
  $\ExpTy{\Venv}{v}{\PatUpperTy{p}{\Env{g}}}$ as required.
\end{proof}

\begin{lemma}[Subtyping of pattern environments]
  \label{lem:env-pattern-sub}
  Let $p$ be a pattern and $\Monoty, \Monoty'$ two types such that
  $\PatEnv{t}{p}$ and $\PatEnv{t'}{p}$ are defined.
  If $\Monoty \IsSubty \Monoty'$ then $(\PatEnv{\Monoty}{p}) \IsSubty (\PatEnv{\Monoty'}{p})$.
\end{lemma}
\begin{proof}
  By induction on $p$.
  Define $\Venv = \PatEnv{\Monoty}{p}$ and  $\Venv' = \PatEnv{\Monoty'}{p}$.

  \begin{CaseDistinction}{on the structure of $p$}
    \CdCase{$p = v$ or $p = \Wildcard$ or $p = \Capt{x}$}
    There is nothing to show since $\Dom{\Venv} = \EmptyEnv = \Dom{\Venv'}$.
    \CdCase{$p = x$}
    By definition, $\Dom{\Venv} = \{x\}= \Dom{\Venv'}$ and $\Venv(x) = \Monoty \IsSubty \Monoty' = \Venv'(x)$ by assumption.
    \CdCase{$p = (p_1, p_2)$}
    By assumption, $\Monoty \IsSubty \Monoty'$, hence by \Cref{prop:projection},
    $\ProjI{\Monoty} \IsSubty \ProjI{\Monoty'}$. Applying the \IH{} yields
    \begin{igather}
      (\PatEnv{\ProjI{t}}{p_i}) \IsSubty (\PatEnv{\ProjI{t'}}{p_i})
    \end{igather}
    Then
    \begin{igather}
      \PatEnv{\Monoty}{(p_1, p_2)} =
      (\PatEnv{\ProjL{\Monoty}}{p_1}) \InterEnv (\PatEnv{\ProjR{\Monoty}}{p_2}) \IsSubty
      (\PatEnv{\ProjL{t'}}{p_1}) \InterEnv (\PatEnv{\ProjR{t'}}{p_2}) =
      \PatEnv{t'}{(p_1, p_2)}
    \end{igather}
  \end{CaseDistinction}
\end{proof}

\begin{lemma}[Subtyping of guarded pattern environments]
  \label{lem:env-guarded-pattern-sub}
  Let $\Pg$ be a guarded pattern and $\Monoty, \Monoty'$ two types such that
  $\PatEnv{t}{p}$ and $\PatEnv{t'}{p}$ are defined.
  If $\Monoty \IsSubty \Monoty'$ then $(\PatEnv{\Monoty}{\Pg}) \IsSubty (\PatEnv{\Monoty'}{\Pg})$.
\end{lemma}

\begin{proof}
  Let $\Pg = \WithGuard p g$.
  Define $\Venv$ and $\Venv'$ as follows:
  \begin{igather}
    \Venv := \PatEnv{\Monoty}{\Pg} = \PatEnv{\Monoty}{p} \InterEnv \Env{g}\\
    \Venv' := \PatEnv{\Monoty'}{\Pg} = \PatEnv{\Monoty'}{p} \InterEnv \Env{g}
  \end{igather}
  By \Cref{lem:env-pattern-sub} $\PatEnv{t}{p} \IsSubty \PatEnv{t'}{p'}$. Thus
  $\Dom{\Venv} = \Dom{\Venv'}$. Assume $x \in \Dom{\Venv}$.
  If $x \in \Dom{\PatEnv{t}{p}}$ then
  $x \in \Dom{\PatEnv{t'}{p}}$ and
  $(\PatEnv{t}{p})(x) \IsSubty \PatEnv{t'}{p}(x)$.
  Hence $\Venv(x) \IsSubty \Venv'(x)$ by definition of intersection for environments.
\end{proof}

\begin{lemma}
  \label{lem:dyn-match-guard-ty}
  Assume $\ExpTy{\Venv}{v}{t}$ and $\PatSubst{v}{p} = \Valsubst$.
  If $x \in \Dom{\Valsubst}$ and $\Valsubst(x) \ValMatches \GuardTy$, then
  $\ExpTy{\Venv}{\Valsubst(x)}{\TyOf{\GuardTy}}$.
\end{lemma}
\begin{proof}
  By induction on $p$.
  \begin{CaseDistinction}{on the form of $p$}
    \CdCase{$p = v'$ or $p = \Wildcard$ or $p = \Capt{y}$} Impossible since $\PatSubst{v}{p} = \PatSubstFail$ or
    $\PatSubst v p = []$ in this case.

    \CdCase{$p = x$} Then $\Valsubst = [v/x]$, so $\Valsubst(x) = v$ and $v \ValMatches \Gt$
    by assumption. The claim now follows by \Cref{lem:val-match-implies-ty}.

    \CdCase{$p = (p_1, p_2)$} Then $v = (v_1, v_2)$ and
    \begin{igather}
      \PatSubst{v_i}{p_i} = \Valsubst_i \quad (i=1,2)\\
      \Valsubst_1 \ValEq \Valsubst_2 \\
      \Valsubst = \Valsubst_1 \cup \Valsubst_2
    \end{igather}
    From assumption $\ExpTy{\Venv}{v}{t}$, we get with \Cref{lem:sub-elim} that for some $t_1, t_2$
    \begin{igather}
      \ExpTy{\Venv}{v_i}{t_i} \quad (i=1,2)\\
      \Pairty{t_1}{t_2} \IsSubty t
    \end{igather}
    Assume $x \in \Dom{\Valsubst}$ with $\Valsubst(x) \ValMatches \Gt$. Then
    $x \in \Dom{\Valsubst_1}$ and/or $x \in \Dom{\Valsubst_2}$, and if $x$ is
    in both domains, then $\Valsubst_1(x) = \Valsubst_2(x)$

    Now assume $i \in \{1,2\}$ with $x \in \Dom{\Valsubst_i}$. Then $\Valsubst_i(x) = \Valsubst(x)$, so
    $\Valsubst_i(x) \ValMatches \Gt$. Applying the \IH{} then yields
    \begin{igather}
      \ExpTy{\Venv}{\Valsubst_i(x)}{\TyOf{\Gt}}
    \end{igather}
    The claim now follows with $\Valsubst_i(x) = \Valsubst(x)$.
  \end{CaseDistinction}
\end{proof}

\begin{definition}
  For some finite mapping $\mathcal F$ and some set $A$,
  we write $\RestrictDomain{\mathcal F}{A}$ to denote the finite mapping with
  restricted domain $A$; that is,
  $\RestrictDomain{\mathcal F}{A} = \{ x: \mathcal{F}(x) \mid x \in A \cap \Dom{\mathcal F} \}$.
\end{definition}

\begin{lemma}
  \label{lem:patty-subst}
  ~
  \begin{EnumThm}
  \item $\SafeDown{g}{A}$ implies $\SafeDown{g\Valsubst}{A}$ provided $\Dom{\Valsubst} \cap A = \emptyset$.
  \item $\PatLowerTy{p}{\Venv_1} \IsSubty \PatLowerTy{p\Valsubst}{\Venv_2}$ provided
    $\PatBound{p} \cap \Dom{\Valsubst} = \emptyset$ and
    $\RestrictDomain{\Venv_1}{\PatBound{p}} = \RestrictDomain{\Venv_2}{\PatBound{p}}$.
  \item $\PgLowerTy{\Pg} \IsSubty \PgLowerTy{\Pg\Valsubst}$ provided
    $\PatBound{\Pg} \cap \Dom{\Valsubst} = \emptyset$.
  \end{EnumThm}
\end{lemma}
\begin{proof}
  \emph{Proof of (i).} We proceed by induction on $g$. If $g = \GuardAnd{g_1}{g_2}$, the claim follows from the \IH{}
  If $g = \GuardTrue$ or $g = \GuardOracle$, the claim is obvious. Thus, assume $g$ is a type test.
  \begin{itemize}
  \item $g = \Is{\Gt}{v}$ and $v \ValMatches \Gt$.
    Then $v\Valsubst$ is a value and $v\Valsubst \ValMatches \Gt$ by \Cref{lem:val-matches-subst}.
    Thus, $\SafeDown{g\Valsubst}{A}$.
  \item $g = \Is{\Gt}{x}$ and $x \in A$. With assumption $\Dom{\Valsubst} \cap A = \emptyset$, we have $e\Valsubst = x$,
    so $\SafeDown{g\Valsubst}{A}$.
  \item Otherwise, $\SafeDown{g}{A} = \False$, so nothing is to prove.
  \end{itemize}

  \emph{Proof of (ii).} Induction on $p$.
  Assume $\PatLowerTy{p}{\Venv_1} \neq \ErlBot$, otherwise nothing is to prove
  because $\ErlBot \IsSubty t$ for any $t$ by \Cref{prop:subty}(i).
  \begin{itemize}
  \item $p = v$. From $\PatLowerTy{p}{\Venv_1} \neq \ErlBot$, we have $v = \Const$ and $\TyOf{\Const} = \Const$
    and $\PatLowerTy{p}{\Venv_1} = \Const$ for some $\Const$. Then $p\Valsubst = \Const$ and
    $\PatLowerTy{p\Valsubst}{\Venv_2} = \Const$.
  \item $p = \Wildcard$. Obvious because
    $p\Valsubst = \Wildcard = p$ and $\PatLowerTy{p\Valsubst}{\Venv} = \ErlTop$ for any $\Venv$.
  \item $p = x$. Then $x \in \PatBound{p}$, so $p\Valsubst = p = x$
    by assumption $\PatBound{p} \cap \Dom{\Valsubst} = \emptyset$.
    The claim follows by assumption
    $\RestrictDomain{\Venv_1}{\PatBound{p}} = \RestrictDomain{\Venv_2}{\PatBound{p}}$.
  \item $p = \Capt x$. Impossible since $\PatLowerTy{p}{\Venv_1} \neq \ErlBot$.
  \item $p = (p_1, p_2)$. Follows from the \IH{}
  \end{itemize}

  \emph{Proof of (iii).}
  Assume $\Pg = \WithGuard{p}{g}$.
  From assumption $\PatBound{\Pg} \cap \Dom{\Valsubst} = \emptyset$ we also have
  $\PatBound{p} \cap \Dom{\Valsubst} = \emptyset$.
  If $\PgLowerTy{\Pg} = \ErlBot$, the claim is obvious. Otherwise:
  \begin{igather}
    \PgLowerTy{\Pg} = \PatLowerTy{p}{\Env{g}} \\
    \SafeDown{g}{\PatBound p} \\
    \Env{g} \neq \ErlBot \\
    p ~\textrm{linear}
  \end{igather}
  Hence with (i)
  $\SafeDown{g\Valsubst}{\PatBound{p\Valsubst}}$ because $\PatBound{p} = \PatBound{p\Valsubst}$.
  From \Cref{lem:env-g-subst} $\Env{g\Valsubst} \subseteq \Env{g}$ and
  $\RestrictDomain{\Env{g\Valsubst}}{\PatBound{p}} = \RestrictDomain{\Env{g}}{\PatBound{p}}$.
  Thus,
  $\Env{g\Valsubst} \neq \ErlBot$. Further, $p\Valsubst$ is linear so
  \begin{igather}
    \PgLowerTy{\Pg} = \PatLowerTy{p}{\Env{g}} \ReasonAbove{\IsSubty}{\textrm{(ii)}}
    \PatLowerTy{p\Valsubst}{\Env{g\Valsubst}} = \PgLowerTy{\Pg\Valsubst}
  \end{igather}
\end{proof}

\begin{lemma}
  \label{lem:safe-implies-not-bot}
  If $\SafeUp{g\Valsubst}$ and $\Env{g\Valsubst} \neq \ErlBot$, then $\Env{g} \neq \ErlBot$.
\end{lemma}
\begin{proof}
  \renewcommand\LabelQualifier{lem:safe-implies-not-bot}
  Define $\Venv = \Env{g}$. Assume $\Venv = \ErlBot$. Then there exists $x \in \Dom{\Venv}$ with
  $\Venv(x) = \ErlBot$. Hence, $g$ contains (at least) two type tests $\Is{\Gt_1}{x}$ and $\Is{\Gt_2}{x}$
  such that
  \begin{igather}
    \TyOf{\Gt_1} \Inter \TyOf{\Gt_2} = \ErlBot \QLabel{eq:eq}
  \end{igather}
  From $\Env{g\Valsubst} \neq \ErlBot$, we know that $x \in \Dom{\Valsubst}$. From
  $\SafeUp{g\Valsubst}$ we get $\Valsubst(x) \ValMatches \Gt_1$ and
  $\Valsubst(x) \ValMatches \Gt_2$. But this contradicts \QRef{eq:eq}. Hence, $\Venv \neq \ErlBot$.
\end{proof}

\begin{lemma}
  \label{lem:patty-subst-up}
  ~
  \begin{EnumThm}
  \item $\SafeUp{g\Valsubst}$ implies $\SafeUp{g}$.
  \item $\PatUpperTy{p\Valsubst}{\Venv_1} \IsSubty \PatUpperTy{p}{\Venv_2}$ provided
    $\PatBound{p} \cap \Dom{\Valsubst} = \emptyset$ and
    $\RestrictDomain{\Venv_1}{\PatBound{p}} = \RestrictDomain{\Venv_2}{\PatBound{p}}$.
  \item $\PgUpperTy{\Pg\Valsubst} \IsSubty \PgUpperTy{\Pg}$ provided
    $\PatBound{\Pg} \cap \Dom{\Valsubst} = \emptyset$.
  \end{EnumThm}
\end{lemma}
\begin{proof}
  \emph{Proof of (i).} We proceed by induction on $g$. If $g = \GuardAnd{g_1}{g_2}$, the claim follows from the \IH{}
  If $g = \GuardTrue$ or $g = \GuardOracle$, the claim is obvious. Thus, assume $g$ is a type test. If $\SafeUp{g} = \True$, we are done.
  Assume $\SafeUp{g} = \False$. Then $g = \Is{\Gt}{v}$ for some value $v$ with not $v \ValMatches \Gt$. By \Cref{lem:val-matches-subst}
  not $v\Valsubst \ValMatches \Gt$. Thus, $\SafeUp{g\Valsubst} = \False$, which is a contradiction.

  \emph{Proof of (ii).} Induction on $p$. Assume $\PatUpperTy{p\Valsubst}{\Venv_1} \neq \ErlBot$,
  otherwise nothing is to prove.
  \begin{itemize}
  \item $p = v$. If $v$ is some $\Const$, then $v\Valsubst = \Const$ and
    $\PatUpperTy{v\Valsubst}{\Venv_1} = \TyOf{\Const} = \PatUpperTy{v}{\Venv_2}$.
  \item $p = \Wildcard$, obvious.
  \item $p = x$. Then $x \in \PatBound{p}$, so $p\Valsubst = p = x$
    by assumption $\PatBound{p} \cap \Dom{\Valsubst} = \emptyset$.
    The claim follows by assumption
    $\RestrictDomain{\Venv_1}{\PatBound{p}} = \RestrictDomain{\Venv_2}{\PatBound{p}}$.
  \item $p = \Capt x$. Then $\PatUpperTy{\Capt{x}}{\Venv_2} = \ErlTop$.
  \item $p = (p_1, p_2)$. Follows from the \IH{}
  \end{itemize}

  \emph{Proof of (iii).}
  Let $\Pg = \WithGuard p g$. From the assumption $\PatBound{\Pg} \cap \Dom{\Valsubst} = \emptyset$,
  we also have $\PatBound{p} \cap \Dom{\Valsubst} = \emptyset$.
  If $\PgUpperTy{\Pg\Valsubst} = \ErlBot$, the claim is obvious. Otherwise
  \begin{igather}
    \PgUpperTy{\Pg\Valsubst} = \PatUpperTy{p\Valsubst}{\Env{g\Valsubst}} \\
    \SafeUp{g\Valsubst} \\
    \Env{g\Valsubst} \neq \ErlBot
  \end{igather}
  With \Cref{lem:safe-implies-not-bot} then $\Env{g} \neq \ErlBot$. With (i) also $\SafeUp{g}$.
  From \Cref{lem:env-g-subst}
  $\RestrictDomain{\Env{g\Valsubst}}{\PatBound{p}} = \RestrictDomain{\Env{g}}{\PatBound{p}}$.
  Thus
  \begin{igather}
    \PgUpperTy{\Pg\Valsubst} = \PatUpperTy{p\Valsubst}{\Env{g\Valsubst}} \ReasonAbove{\IsSubty}{\textrm{(ii)}}
    \PatUpperTy{p}{\Env{g}} = \PgUpperTy{\Pg}
  \end{igather}
\end{proof}

\begin{lemma}[Correctness of environment generation for patterns]
  \label{lem:correct-env-gen-pat}
  Assume $\ExpTy{\Venv}{v}{t}$ and $\PatSubst{v}{p} = \Valsubst$.
  Then for all $x \in \Dom{\Valsubst}$: $\ExpTy{\Venv}{x\Valsubst}{(\PatEnv{t}{p})(x)}$.
\end{lemma}
\begin{proof}
  \renewcommand\LabelQualifier{lem:correct-env-gen-pat}
  Induction on $p$.
  \begin{CaseDistinction}{on the shape of $p$}
    \CdCase{$p = v$ or $p = \Wildcard$ or $p = \Capt{y}$} Impossible, since either
    $\Valsubst = []$ or $\PatSubst{v}{p} = \PatSubstFail$.

    \CdCase{$p = y$} Then $\Valsubst = [v/y]$ and $\PatEnv{t}{p} = \{y : t \}$.
    If $x \in \Dom{\Valsubst}$, then $x = y$ and $x\Valsubst = v$ and
    $(\PatEnv t p)(x) = t$, so the claim follows by the assumption $\ExpTy{\Venv}{v}{t}$.

    \CdCase{$p = (p_1, p_2)$} Then
    \begin{igather}
      \PatEnv t p = (\PatEnv{\ProjL{t}}{p_1}) \InterEnv (\PatEnv{\ProjR{t}}{p_2}) \QLabel{eq:def-patenv}\\
      v = (v_1, v_2) \\
      \PatEnv{v_i}{p_i} = \Valsubst_i \quad (i=1,2) \QLabel{eq:substi} \\
      \Valsubst_1 \ValEq \Valsubst_2 \\
      \Valsubst = \Valsubst_1 \cup \Valsubst_2
    \end{igather}
    From the assumption $\ExpTy{\Venv}{v}{t}$ and \Cref{lem:sub-elim} for some $t_1, t_2$
    \begin{igather}
      \ExpTy{\Venv}{v_i}{t_i} \QLabel{eq:ty-vi} \\
      \Pairty{t_1}{t_2} \IsSubty t
    \end{igather}
    Applying the \IH{} to \QRef{eq:substi} and \QRef{eq:ty-vi} for $i=1,2$ yields
    \begin{igather}
      \ExpTy{\Venv}{x\Valsubst_i}{(\PatEnv{t_i}{p_i})(x)} \quad (\forall x \in \Dom{\Valsubst_i}) \QLabel{eq:ty-xsubst}
    \end{igather}
    Now assume $x \in \Dom{\PatEnv t p}$. With \QRef{eq:def-patenv}, we have three possible cases.
    \begin{itemize}
    \item $x \in \Dom{\PatEnv{\ProjL{t}}{p_1}}$ and $x \notin \Dom{\PatEnv{\ProjR{t}}{p_2}}$
      With \QRef{eq:substi} and \Cref{lem:pat-vars} then
      $x \in \Dom{\Valsubst_1}$ and $x \notin \Dom{\Valsubst_2}$.

      Hence, $x\Valsubst = x\Valsubst_1$ and
      $(\PatEnv t p)(x) = (\PatEnv{\ProjL{t}}{p_1})(x)$. With \Cref{prop:projection}(iii) also
      $t_1 \IsSubty \ProjL{t}$. Thus, with \QRef{eq:ty-xsubst} and \Cref{lem:env-pattern-sub}
      \begin{igather}
        \ExpTy{\Venv}{x\Valsubst_1}{(\PatEnv{\ProjL{t}}{p_1})(x)}
      \end{igather}
      as required.
    \item $x \notin \Dom{\PatEnv{\ProjL{t}}{p_1}}$ and $x \in \Dom{\PatEnv{\ProjR{t}}{p_2}}$. Similarly as the preceding case.
    \item $x \in \Dom{\PatEnv{\ProjL{t}}{p_1}}$ and $x \in \Dom{\PatEnv{\ProjR{t}}{p_2}}$.
      Then $x \in \Dom{\Valsubst_1}$ and $x \in \Dom{\Valsubst_2}$ and $\Valsubst_1(x) = \Valsubst_2(x)$.
      Further,
      \begin{igather}
        (\PatEnv t p)(x) = (\PatEnv{\ProjL{t}}{p_1})(x) \Inter (\PatEnv{\ProjR{t}}{p_2})(x)
      \end{igather}
      With the same reasoning as in the first case, we can show that
      \begin{igather}
        \ExpTy{\Venv}{x\Valsubst_i}{(\PatEnv{\ProjI{t}}{p_i})(x)}
      \end{igather}
      for $i = 1,2$. Then with \Cref{lem:value-has-both-types}
      \begin{igather}
        \ExpTy{\Venv}{x\Valsubst}{(\PatEnv t p)(x)}
      \end{igather}
      as required.
    \end{itemize}
  \end{CaseDistinction}
\end{proof}

\begin{definition}
  A function environment $\FunEnv$ is in accordance with a type schema environment $\Senv$,
  written $\FunEnv \Accord \Senv$, if and only if the following holds:
  \begin{itemize}
  \item $\FunEnv = \DefSym_{i \in I}$
  \item $\Envs{\DefSym_i} = \MetaPair{\Senv_i}{\Senv_i'}$ for all $i \in I$
  \item $\DefOk{\Senv_{i \in I}}{\DefSym_i}$ for all $i \in I$
  \item $\Senv = \Senv'_{i \in I}$
  \end{itemize}
\end{definition}

\begin{lemma}[Correctness of environment generation for guarded patterns]
  \label{lem:correct-env-gen}
  Assume $\ExpTy{\Venv}{v}{t}$ and $\PatSubst{v}{\Pg} = \Valsubst$ and
  $\FunEnv \Accord \Senv$.
  Then for all $x \in \Dom{\Valsubst}$: $\ExpTy{\Venv}{x\Valsubst}{(\PatEnv{t}{\Pg})(x)}$.
\end{lemma}

\begin{proof}
  \renewcommand\LabelQualifier{lem:correct-env-gen}
  Assume $\Pg = \WithGuard p g$.
  From $\PatSubst{v}{\Pg} = \Valsubst$ we get
  \begin{igather}
    \PatSubst{v}{p} = \Valsubst \QLabel{eq:subst} \\
    \EvalGuard{g\Valsubst} = \True \QLabel{eq:guard-true}
  \end{igather}
  Assume $x \in \Dom{\Valsubst}$. With \Cref{lem:pat-vars} then $x \in \Dom{\PatEnv{t}{\Pg}}$.
  By definition, we have
  \begin{igather}
    \PatEnv{t}{\Pg} = \PatEnv{t}{p} \InterEnv \Env{g} \QLabel{eq:def-t-pg}
  \end{igather}
  \begin{itemize}
  \item $x \in \Dom{\PatEnv{t}{p}}$ and $x \notin \Dom{\Env{g}}$.
    The claim then follows with \Cref{lem:correct-env-gen-pat}.

  \item $x \notin \Dom{\PatEnv{t}{p}}$ and $x \in \Dom{\Env{g}}$.
    Impossible because $\Dom{\Valsubst} = \Dom{\PatEnv{t}{p}}$ by \Cref{lem:pat-vars}.

  \item $x \in \Dom{\PatEnv{t}{p}}$ and $x \in \Dom{\Env{g}}$.
    With \Cref{lem:correct-env-gen-pat} we have
    \begin{igather}
      \ExpTy{\Venv}{x\Valsubst}{(\PatEnv{t}{p})(x)} \QLabel{eq:ty-tp}
    \end{igather}
    From $x \in \Dom{\Env{g}}$, we know that $g$ contains type tests $\Is{\Gt_i}{x}$
    for $i \in I \neq \emptyset$. From \QRef{eq:guard-true} we know that
    \begin{igather}
      \Valsubst(x) \ValMatches \Gt_i \quad (\forall i \in I)
    \end{igather}
    Hence, all $\Gt_i$ are identical. Call this guard type $\Gt$.
    With assumption $\ExpTy{\Venv}{v}{t}$, \QRef{eq:subst}, $x \in \Dom{\Valsubst}$,
    and $\Valsubst(x) \ValMatches \Gt$, we may apply \Cref{lem:dyn-match-guard-ty} and get
    \begin{igather}
      \ExpTy{\Venv}{x\Valsubst}{\TyOf{\Gt}}
    \end{igather}
    We also have $\Env{g}(x) = \TyOf{\Gt}$, so with \QRef{eq:ty-tp} and \QRef{eq:def-t-pg} then
    $\ExpTy{\Venv}{x\Valsubst}{(\PatEnv{t}{\Pg})(x)}$ as required.
  \end{itemize}
\end{proof}

\begin{lemma}
  \label{lem:upper-ty-not-bot-value-has-ty}
  Assume $\ExpTy{\Venv}{v}{t}$ for some $t$.
  Let $\Venv' = \PatEnv{u}{\Pg}$ and $\Valsubst = [v/x]$.
  If $x \in \Dom{\Venv'}$ with $x \notin \PatBound{\Pg}$ and $\PgUpperTy{\Pg\Valsubst} \neq \ErlBot$,
  then $\ExpTy{\Venv}{v}{\Venv'(x)}$.
\end{lemma}
\begin{proof}
  Assume $\Pg = \WithGuard{p}{g}$. With $x \in \Dom{\Venv'}$, we know that $g$ contains type tests
  $(\Is{\Gt_i}{x})_{i \in I}$ for $I \neq \emptyset$, and
  \begin{igather}
    \Venv'(x) = \InterBig_{i \in I}{\TyOf{\Gt_i}}
  \end{igather}
  From $\PgUpperTy{\Pg\Valsubst} \neq \ErlBot$ and $x \notin \PatBound{\Pg}$, we know $\SafeUp{g\Valsubst} = \True$.
  Thus, $\SafeUp{\Is{\Gt_i}{v}}$ for all $i \in I$. Hence, $v \ValMatches \Gt_i$ for all $i \in I$.
  With \Cref{lem:val-match-implies-ty} then $\ExpTy{\Venv}{v}{\TyOf{\Gt_i}}$ for all $i \in I$.
  Then by \Cref{lem:value-has-both-types}, $\ExpTy{\Venv}{v}{\Venv'(x)}$ as required.
\end{proof}

\begin{lemma}[Value substitution lemma]
  \label{lem:expr-subst-single-mono}
  Assume $x \notin \Dom{\Venv}$.
  If $\ExpTy{\Venv, \{x : t'\}}{e}{t}$ and $\ExpTy{\Venv}{v}{t'}$, then
  $\ExpTy{\Venv}{e[v/x]}{\Monoty}$. 
\end{lemma}

\begin{proof}
  \renewcommand\LabelQualifier{lem:expr-subst-single-mono}
  Define $\Venv' = \{x:t'\}$ and $\Valsubst = [v/x]$. We proceed by induction on the typing derivation of
  $\ExpTy{\Venv, \Venv'}{e}{\Monoty}$.
  \begin{CaseDistinction}{on the last rule of the derivation}
    \CdCase{\Rule{d-var}} We distinguish two cases:
    \begin{itemize}
    \item $e = x$: Then $e\Valsubst = v$ and $t = t'$, so $\ExpTy{\Venv}{e\Valsubst}{t}$ by assumption.
    \item $e = y \neq x$: Then $e\Valsubst = y$ and $t = \Venv(y)$, so $\ExpTy{\Venv}{e\Valsubst}{t}$ by rule \Rule{d-var}.
    \end{itemize}

    \CdCase{\Rule{d-var-poly}} Then $e = y \neq x$ and $t \in \Inst{\Senv(y)}$, so $\ExpTy{\Venv}{e\Valsubst}{t}$ by rule \Rule{d-var-poly}.

    \CdCase{\Rule{d-const}} Trivial.

    \CdCase{\Rule{d-pair}} Straightforward from the \IH{}
    \CdCase{\Rule{d-app}} Straightforward from the \IH{}

    \CdCase{\Rule{d-abs}} Then $e = \lambda y . e'$. \Wlog{}, $y$ fresh. From the premise of the rule
    $\ExpTy{\Venv, \Venv', y:t_1}{e'}{t_2}$ for $t = t_1 \to t_2$. Then $\ExpTy{\Venv, y:t_1}{e'\Valsubst}{t_2}$
    from the \IH{} Hence $\ExpTy{\Venv}{e\Valsubst}{t}$ via rule \Rule{d-abs}.

    \CdCase{\Rule{d-case}}
    Then $e = \Case{e'}{(\Pg_i \to e_i)_{i \in I}}$. Inverting the rule yields:
    \begin{igather}
      \ExpTy{\Venv,\Venv'}{e'}{u} \QLabel{eq:ty-ep}\\
      u \IsSubty \UnionBig_{i \in I} \PgLowerTy{\Pg_i} \QLabel{eq:u-sub}\\
      t = \UnionBig_{i \in I}{t_i'} \QLabel{eq:def-t}
    \end{igather}
    And for all $i \in I$:
    \begin{igather}
      \Free{\Pg_i} \subseteq \Dom{\Senv} \cup \Dom{\Venv} \cup \Dom{\Venv'} \QLabel{eq:free1}\\
      t_i = (u \WithoutTy \UnionBig_{j < i} \PgLowerTy{\Pg_j}) \Inter \PgUpperTy{\Pg_i}\\
      \Venv_i = (\Venv, \Venv') \InterEnv (\PatEnv{t_i}{\Pg_i}) \QLabel{eq:def-venvi}
      \\
      {
        \begin{cases}
          t_i' = \ErlBot ~\textrm{and}~
          \Free{e_i} \subseteq \Dom{\Senv} \cup \Dom{\Venv_i}
          & \textrm{if}~ t_i \IsSubty \ErlBot
          \\
          \ExpTy{\Venv_i}{e_i}{t_i'}
          & \textrm{otherwise}
        \end{cases}
      } \QLabel{eq:ty-ei}
    \end{igather}
    \Wlog{}, $\PatBound{\Pg_{i \in I}}$ fresh, so
    \begin{igather}
      \PatBound{\Pg_i} \cap \Dom{\Valsubst} = \emptyset \quad (\forall i \in I) \QLabel{eq:fresh} \\
      e\Valsubst = \Case{e'\Valsubst}{(\Pg_i\Valsubst \to e_i\Valsubst)_{i \in I}}
    \end{igather}
    Applying the \IH{} on \QRef{eq:ty-ep} yields
    \begin{igather}
      \ExpTy{\Venv}{e'\Valsubst}{u} \QLabel{eq:ep-subst-ty}
    \end{igather}
    With \QRef{eq:fresh} and \Cref{lem:patty-subst} we have
    for any $i \in I$ that $\PgLowerTy{\Pg_i} \IsSubty \PgLowerTy{\Pg_i\Valsubst}$.
    With \QRef{eq:u-sub} then
    \begin{igather}
      u \IsSubty \UnionBig_{i \in I} \PgLowerTy{(\Pg_i\Valsubst)} \QLabel{eq:u-sub2}
    \end{igather}
    Define for all $i \in I$
    \newcommand\tis{t_i^\star}
    \newcommand\tisp{t_i^{\star\prime}}
    \newcommand\tix{t_{i,x}}
    \newcommand\Venvis{\Venv_i^\star}
    \newcommand\Venvit{\widetilde{\Venv}_i}
    \newcommand\Venvitt{\widetilde{\widetilde{{\Venv}}}_i}
    \begin{igather}
      \tis = (u \WithoutTy \UnionBig_{j < i} \PgLowerTy{\Pg_j\Valsubst}) \Inter \PgUpperTy{\Pg_i\Valsubst} \QLabel{eq:def-tis} \\
      \Venvis = \PatEnv{\tis}{\Pg_i\Valsubst} \QLabel{eq:def-venvis}
    \end{igather}
    With \QRef{eq:fresh} and \Cref{lem:patty-subst} and~\Cref{lem:patty-subst-up}, we have
    for any $i \in I$ that $\PgLowerTy{\Pg_i} \IsSubty \PgLowerTy{\Pg_i\Valsubst}$ and
    $\PgUpperTy{\Pg_i\Valsubst} \IsSubty \PgUpperTy{\Pg_i}$.
    For any $i \in I$, we then get by set-theoretic reasoning and the definitions of $\tis$ and $t_i$
    that
    \begin{igather}
      \tis \IsSubty t_i   \QLabel{eq:tis-subty}
    \end{igather}

    From \QRef{eq:free1} and the definition of $\Valsubst$ we have
    \begin{igather}
      \Free{\Pg_i \Valsubst} \subseteq \Dom{\Senv} \cup \Dom{\Venv} \QLabel{eq:free2}\\
    \end{igather}
    Now assume $i \in I$.
    \begin{itemize}
    \item Case $\tis \IsSubty \ErlBot$. From \QRef{eq:ty-ei} we have (using \Cref{lem:free-exp-vars} in the case where
      $t_i \not\IsSubty \ErlBot$) $\Free{e_i} \subseteq \Dom{\Senv} \cup \Dom{\Venv_i}$. Hence
      \begin{igather}
        \Free{e_i\Valsubst} \subseteq \Dom{\Senv} \cup \Dom{\Venvis} \QLabel{eq:free-ei-subst}
      \end{igather}
      Define $\tisp = \ErlBot$ in this case.
    \item Case $\tis \not\IsSubty \ErlBot$. Define
      \begin{igather}
        \Venvit = \PatEnv{t_i}{\Pg_i}\\
        \Venvitt = \PatEnv{t_i}{(\Pg_i\Valsubst)}\\
        \tix = {
          \begin{cases}
            \Venvit(x) & \textrm{if}~ x \in \Dom{\Venvit} \\
            \ErlTop    & \textrm{otherwise}
          \end{cases}
        }
      \end{igather}
      If $x \in \Dom{\Venvit}$, we have by \Cref{lem:pg-env-subst} $\Venvit = \Venvitt, x:\tix$.
      Otherwise, the same lemma gives us $\Venvit = \Venvitt$.

      Thus, with \QRef{eq:def-venvi} and the definition of $\tix$ we have
      \begin{igather}
        \Venv_i = (\Venv, \Venv') \InterEnv \Venvit = (\Venv, \{x: t'\}) \InterEnv \Venvit
        = (\Venv \InterEnv \Venvitt), x : (t' \Inter \tix)  \QLabel{eq:eq-venvi}
      \end{igather}
      From the assumption $\ExpTy{\Venv}{v}{t'}$, we get with
      \Cref{lem:weakening} and~\Cref{lem:extra-var}
      \begin{igather}
        \ExpTy{\Venv \InterEnv \Venvitt}{v}{t'} \QLabel{eq:ty-v}
      \end{igather}
      With $\tis \not\IsSubty \ErlBot$ and \QRef{eq:def-tis} we get $\PgUpperTy{\Pg_i\Valsubst} \neq \ErlBot$.
      Thus, if $x \in \Dom{\Venvit}$, we have by \Cref{lem:upper-ty-not-bot-value-has-ty} that
      $\ExpTy{\Venv}{v}{\tix}$ with $\tix = \Venvit(x)$. Otherwise $x \notin \Dom{\Venvit}$,
      so we have $\ExpTy{\Venv}{v}{\tix}$ because then $\tix = \ErlTop$.
      With \QRef{eq:ty-v} and \Cref{lem:weakening} and~\Cref{lem:extra-var} then
      \begin{igather}
        \ExpTy{\Venv \InterEnv \Venvitt}{v}{(t' \Inter \tix)} \QLabel{eq:ty-v2}
      \end{igather}
      With $\tis \not\IsSubty \ErlBot$ and \QRef{eq:tis-subty} we have $t_i \not\IsSubty \ErlBot$.
      Then with \QRef{eq:ty-ei} and \QRef{eq:eq-venvi}
      \begin{igather}
        \ExpTy{(\Venv \InterEnv \Venvitt), x: (t' \Inter \tix)}{e_i}{t_i'} \QLabel{eq:ty-ei2}
      \end{igather}
      Applying the \IH{} to \QRef{eq:ty-v2} and \QRef{eq:ty-ei2} yields
      \begin{igather}
        \ExpTy{\Venv \InterEnv \Venvitt}{e_i \Valsubst}{t_i'}
      \end{igather}
      With \QRef{eq:def-venvis}, the definition of $\Venvitt$, with \QRef{eq:tis-subty}
      and \Cref{lem:env-guarded-pattern-sub}, we have $\Venvis \IsSubty \Venvitt$, so with
      \Cref{lem:weakening}
      \begin{igather}
        \ExpTy{\Venv \InterEnv \Venvis}{e_i \Valsubst}{t_i'} \QLabel{eq:ty-ei3}
      \end{igather}
      Define $\tisp = t_i'$. This finishes the case where $t_i \not\IsSubty \ErlBot$.
    \end{itemize}
    Using \QRef{eq:ep-subst-ty}, \QRef{eq:u-sub2}, \QRef{eq:free2}, \QRef{eq:def-tis}, \QRef{eq:free-ei-subst},
    and \QRef{eq:ty-ei3} as the premises of rule \Rule{d-case} gives us
    \begin{igather}
      \ExpTy{\Venv}{e\Valsubst}{\UnionBig_{i \in I}\tisp}
    \end{igather}
    By construction of $\tisp$, we have $\tisp \IsSubty t_i'$ for all $i \in I$. Hence,
    $\UnionBig\nolimits_{i \in I} \tisp \IsSubty \UnionBig\nolimits_{i \in I} t_i'$.
    Using rule \Rule{d-sub} and \QRef{eq:def-t} then gives us
    $\ExpTy{\Venv}{e\Valsubst}{t}$ as required.

    \CdCase{\Rule{d-sub}} Straightforward from the \IH{}
  \end{CaseDistinction}
\end{proof}

\begin{lemma}[Substitution lemma for multiple values]
  \label{lem:expr-subst}
  Assume $\ExpTy{\Venv,\Venv'}{e}{t}$ with $\Dom{\Venv} \cap \Dom{\Venv'} = \emptyset$.
  Let $\Valsubst$ be a value substitution with $\Dom{\Valsubst} = \Dom{\Venv'}$ and
  $\ExpTy{\Venv}{\Valsubst(x)}{\Venv'(x)}$
  for all $x \in \Dom{\Valsubst}$.
  Then $\ExpTy{\Venv}{e\Valsubst}{t}$.
\end{lemma}
\begin{proof}
  By induction on the number of elements in $\Venv'$, using \Cref{lem:expr-subst-single-mono}.
\end{proof}

\begin{lemma}
  \label{lem:value-has-one-union-type}
  If $\ExpTy{\Venv}{v}{{\UnionBig_{i \in I}{\PgLowerTy{\Pg_i}}}}$ and $I \neq \emptyset$,
  then there exists an $j \in I$ such that $\ExpTy{\Venv}{v}{\PgLowerTy{\Pg_j}}$.
\end{lemma}
\begin{proof}
  Define the following grammar for simple types:
  \begin{igather}
    s ::= \ErlTop \mid \ErlBot \mid \Constty \mid \Int \mid \Float \mid \ErlBot \to \ErlTop \mid \Pairty s s
  \end{igather}
  It is straightforward to verify that for any guard pattern $\Pg$ there exists some simple type
  $s$ such that $\PgLowerTy{\Pg}$ is equivalent to $s$.
  (For any guard $g$, $\Env{g}$ maps variables to simple types. For any pattern $p$,
  $\PatLowerTy{p}{\Venv}$ is a simple type given that $\Venv$ contains only simple types.)

  The rest of the proof relies on the model of types,
  see \citet[Lemma~A.14, page ~27]{journals/corr/CastagnaP016}
  and \cite{Frisch2004,conf/popl/Castagna0XA15}.
\end{proof}

\begin{theorem}[Progress for expressions]
  \label{thm:progress}
  Assume $\ExpTy{\EmptyEnv}{e}{t}$ with $\FunEnv \Accord \Senv$. Then either $e$ is
  a value or there exists $\estar$ such that $\Reduce{\FunEnv}{e}{\estar}$.
\end{theorem}
\begin{proof}
  \renewcommand\LabelQualifier{thm:progress}
  By induction on the derivation of $\ExpTy{\EmptyEnv}{e}{t}$.
  \begin{CaseDistinction}{on the last rule of the derivation}
    \CdCase{\Rule{d-var}} Then $e = x$ with $x \in \Dom{\EmptyEnv}$,
    which is impossible.

    \CdCase{\Rule{d-var-poly}}
    Then $e = x$ for some $x \in \Dom{\Senv}$. From $\FunEnv \Accord \Senv$, we know that
    $(x = v) \in \FunEnv$ or $(x : \sigma) = v \in \FunEnv$. Thus, $x$ reduces to $v$ by
    \Rule{red-var} and \Rule{red-context}.

    \CdCase{\Rule{d-const}, \Rule{d-abs}} $e$ is a value.

    \CdCase{\Rule{d-app}}
    Then $e = \App{e_1}{e_2}$ and $\ExpTy{\EmptyEnv}{e_1}{t' \to t}$ and $\ExpTy{\EmptyEnv}{e_2}{t'}$.
    From the \IH{}, we know that one of the following holds:
    \begin{itemize}
    \item $e_1$ takes a step. Then $e$ takes a step.
    \item $e_1$ is a value and $e_2$ takes a step. Then $e$ takes a step.
    \item Both $e_1$ and $e_2$ are values. By \Cref{lem:value-type-determines-shape}, we know that $e_1$ is a lambda,
      so $e$ takes a step by \Rule{red-abs} and \Rule{red-context}.
    \end{itemize}

    \CdCase{\Rule{d-pair}} Straightforward from the \IH{}

    \CdCase{\Rule{d-case}}
    Then $e = \Case{e'}{(\Pg_i \to e_i)_{i \in I}}$. Inverting the rule yields:
    \begin{igather}
      \ExpTy{\EmptyEnv}{e'}{t'} \QLabel{eq:ty-ep}\\
      t' \IsSubty \UnionBig_{i \in I} \PgLowerTy{\Pg_i} \QLabel{eq:tp-sub}\\
      t = \UnionBig_{i \in I}{t_i'}
    \end{igather}
    And for all $i \in I$:
    \begin{igather}
      \Free{\Pg_i} \subseteq \Dom{\Senv} \QLabel{eq:free}\\
      t_i = (t' \WithoutTy \UnionBig_{j < i} \PgLowerTy{\Pg_j}) \Inter \PgUpperTy{\Pg_i}\\
      \Venv_i = \PatEnv{t_i}{\Pg_i}
      \\
      {
        \begin{cases}
          t_i' = \ErlBot ~\textrm{and}~
          \Free{e_i} \subseteq \Dom{\Senv} \cup \Dom{\Venv_i}
          & \textrm{if}~ t_i \IsSubty \ErlBot
          \\
          \ExpTy{\Venv_i}{e_i}{t_i'}
          & \textrm{otherwise}
        \end{cases}
      }
    \end{igather}
    Applying the \IH{} to \QRef{eq:ty-ep} gives us that either $e'$ reduces or is a value.
    If $e'$ reduces, then $e$ reduces using rule \Rule{red-context} and we are done.

    Thus, assume that $e'$ is a value $v$. With \QRef{eq:ty-ep}, \QRef{eq:tp-sub}, and
    \Cref{lem:value-has-one-union-type}, we know that there exists $j \in I$ such that
    \begin{igather}
      \ExpTy{\EmptyEnv}{v}{\PgLowerTy{\Pg_j}}
    \end{igather}
    By \Cref{lem:gp-typing}(i) we then have for some $\Valsubst'$ that
    $\PatSubst{v}{\Pg_j} = \Valsubst'$. Hence, there exists $k \leq j$ with
    \begin{igather}
      \PatSubst{v}{\Pg_k} = \Valsubst \quad \textrm{for some}~\Valsubst \\
      \PatSubst{v}{\Pg_{k'}} = \PatSubstFail \quad (\forall k' < k)
    \end{igather}
    With rules \Rule{red-case} and \Rule{red-context} we then have
    for $\estar = e_k \Valsubst$ that
    $\Reduce{\FunEnv}{e}{\estar}$.

    \CdCase{\Rule{d-sub}} Follows from the \IH{}
  \end{CaseDistinction}
\end{proof}

\begin{theorem}[Progress for programs]
  \label{thm:progress-prog}
  Assume $\ProgTy{\Letrec{\FunEnv}{e}}{\Monoty}$ for some $\Monoty$.
  Then, either $e$ is a value or there exists an expression $\estar$
  such that $\Reduce{}{\Letrec{\FunEnv}{e}}{\Letrec{\FunEnv}{\estar}}$.
\end{theorem}

\begin{proof}
  From $\ProgTy{\Letrec{\FunEnv}{e}}{\Monoty}$ we get by inverting rule
  \Rule{d-prog}
  \begin{igather}
    \FunEnv = \DefSym_{i \in I} \\
    \Envs{\DefSym_i} = \MetaPair{\Senv_i}{\Senv_i'} \quad \Forall{i \in I}\\
    \DefOk{\Senv_{i \in I}}{\DefSym_i} \quad \Forall{i \in I}\\
    \ExpSchemaTy{\Senv'_{i \in I}}{\EmptyEnv}{e}{\Monoty}
  \end{igather}
  Thus, $\FunEnv \Accord \Senv'_{i \in I}$. The claim now follows with
  \Cref{thm:progress} and reduction rule \Rule{red-prog}.
\end{proof}

\begin{lemma}[Type substitution lemma]
    \label{lem:type-subst}
    If $\ExpSchemaTy{\Senv}{\Venv}{e}{\Monoty}$, then $\ExpSchemaTy{\Senv\Tysubst}{\Venv\Tysubst}{e}{\Monoty\Tysubst}$.
\end{lemma}

\begin{proof}
    \begin{CaseDistinction}{on the derivation of $\ExpSchemaTy{\Senv}{\Venv}{e}{\Monoty}$}
    \CdCase{\Rule{d-var}}
    We have $\ExpSchemaTy{\Senv}{\Venv}{x}{\Monoty}$ and $\Venv(x) = \Monoty$.
    By definition, $\Venv\Tysubst(x) = \Monoty\Tysubst$ and the claim holds.

    \CdCase{\Rule{d-var-poly}}
    We have $\ExpSchemaTy{\Senv}{\Venv}{x}{\Monoty}$, $\Senv(x) = \TyScm{A} u$, and
    $\Monoty \in \Inst{\Senv(x)}$, so $t = u\Tysubst^\star$ for some $\Tysubst^\star$ with $\Dom{\Tysubst^\star} = A$.
    By $\alpha$-renaming, we assume $A \cap \Dom{\Tysubst} = \emptyset$ and
    $A \cap \FreeTyVars{\Tysubst} = \emptyset$.
    Then, $(\Senv\Tysubst)(x) = \TyScm{A}u\Tysubst$.
    For $\Tysubst' = [\alpha\Tysubst^\star\Tysubst/\alpha \mid \alpha \in A]$, we have
    $\Monoty\Tysubst = u\Tysubst\Tysubst'$ and hence, $\Monoty\Tysubst \in \Inst{\TyScm{A}u\Tysubst}$.
    The claim follows by application of \Rule{d-var-poly}.

    \CdCase{\Rule{d-const}} The case is straightforward.

    \CdCase{\Rule{d-abs}}
    We have $\ExpSchemaTy{\Senv}{\Venv}{\Abs{x}{e'}}{\Monoty_1 \to \Monoty_2}$ with
    $\ExpTy{\Venv, x: \Monoty_1}{e'}{\Monoty_2}$ and $t = t_1 \to t_2$.
    By \IH{}, we have $\ExpSchemaTy{\Senv\Tysubst}{\Venv\Tysubst, x: \Monoty_1\Tysubst}{e}{\Monoty_2\Tysubst}$.
    By \Rule{d-abs} then
    $\ExpSchemaTy{\Senv\Tysubst}{\Venv\Tysubst}{\Abs{x}{e}}{\Monoty_1\Tysubst \to \Monoty_2\Tysubst}$.

    \CdCase{\Rule{d-app} and \Rule{d-pair}} The claim follows directly from the \IH{}.

    \CdCase{\Rule{d-case}}
    $e = \Case{e'}\Multi{(\Pg_i \to e_i)}$. We get by inverting the rule
    \begin{igather}
      \ExpTy{\Venv}{e'}{u} \QLabel{eq:ty-u}\\
      u \IsSubty \UnionBig\nolimits_{i \in I} \PgLowerTy{\Pg_i} \QLabel{eq:u-sub}
    \end{igather}
    And for all $i \in I$:
    \begin{igather}
      \Free{\Pg_i} \subseteq \Dom{\Senv} \cup \Dom{\Venv} \QLabel{eq:free-pgi}\\
      \Monoty_i = (u \WithoutTy \UnionBig\nolimits_{j < i} \PgLowerTy{\Pg_j}) \Inter \PgUpperTy{\Pg_i} \\
      \Venv_i = \Venv \InterEnv (\PatEnv{\Monoty_i}{\Pg_i})
      \\
      {
        \begin{cases}
          \Monoty_i' = \ErlBot ~\textrm{and}~
          \Free{e_i} \subseteq \Dom{\Senv} \cup \Dom{\Venv_i}
          & \textrm{if}~ \Monoty_i \IsSubty \ErlBot
          \\
          \ExpTy{\Venv_i}{e_i}{\Monoty_i'}
          & \textrm{otherwise}
        \end{cases}
      } \QLabel{eq:ty-ei}\\
      t = \UnionBig\nolimits_{i \in I}{\Monoty_i'}
    \end{igather}
    From \QRef{eq:ty-u} and the \IH{}:
    \begin{igather}
      \ExpSchemaTy{\Senv\Tysubst}{\Venv\Tysubst}{e'}{u\Tysubst}
    \end{igather}
    From \QRef{eq:u-sub}, \Cref{prop:subty}, and \Cref{lem:pat-types-closes}:
    \begin{igather}
      u\Tysubst \IsSubty \UnionBig\nolimits_{i \in I} \PgLowerTy{\Pg_i}
    \end{igather}
    Now assume $i \in I$. From \QRef{eq:free-pgi} we get
    \begin{igather}
      \Free{\Pg_i} \subseteq \Dom{\Senv\Tysubst} \cup \Dom{\Venv\Tysubst}
    \end{igather}
    Define
    \begin{igather}
      t_i^\star = (u\Tysubst \WithoutTy \UnionBig\nolimits_{j < i} \PgLowerTy{\Pg_j}) \Inter \PgUpperTy{\Pg_i}
      \ReasonAbove{=}{\textrm{\Cref{lem:pat-types-closes}}} t_i\Tysubst \QLabel{eq:def-ti-star} \\
      \Venv_i^\star = \Venv\Tysubst \InterEnv (\PatEnv{t_i^\star}{\Pg_i}) \QLabel{eq:def-venvi-star}
    \end{igather}
    \begin{itemize}
    \item Case $t_i^\star \IsSubty \ErlBot$. Define $u_i = \ErlBot$. We also have
      $\Free{e_i} \subseteq \Dom{\Senv\Tysubst} \cup \Dom{\Venv_i^\star}$ because of
      \QRef{eq:ty-ei} and \Cref{lem:free-exp-vars} and $\Dom{\Venv_i} = \Dom{\Venv_i^\star}$.
    \item Otherwise, we must have that $t_i \neq \IsSubty \ErlBot$ because $t_i \IsSubty \ErlBot$ implies
      $t_i\Tysubst \IsSubty \ErlBot$ (\Cref{prop:subty}). Hence, applying the \IH{} to \QRef{eq:ty-ei} yields
      \begin{igather}
        \ExpSchemaTy{\Senv\Tysubst}{\Venv\Tysubst}{e_i}{t_i'\Tysubst}
      \end{igather}
      Define $u_i := t_i' \Tysubst$. We also have $\Venv_i^\star = \Venv_i\Tysubst$ from
      \QRef{eq:def-ti-star} and \QRef{eq:def-venvi-star}.
    \end{itemize}
    We just have established the premises of rule \Rule{d-case}. Applying the rule yields
    \begin{igather}
      \ExpSchemaTy{\Senv\Tysubst}{\Venv\Tysubst}{e}{\UnionBig_{i \in I} u_i}
    \end{igather}
    We have $u_i \IsSubty t_i' \Tysubst$ for all $i \in I$. (If $t_i^\star \IsSubty \ErlBot$ then
    $u_i = \ErlBot \IsSubty t_i' \Tysubst$. Otherwise $u_i = t_i'\Tysubst$.)
    Hence
    \begin{igather}
      \UnionBig_{i \in I} u_i \IsSubty \UnionBig_{i \in I} t_i'\Tysubst = t\Tysubst
    \end{igather}
    Hence with rule \Rule{d-sub} $\ExpSchemaTy{\Senv\Tysubst}{\Venv\Tysubst}{e}{t \Tysubst}$ as required.

    \CdCase{\Rule{d-sub}}
    The claim follows directly from the \IH{} and \Cref{prop:subty}(iv).
    \end{CaseDistinction}
\end{proof}

\begin{theorem}[Subject reduction for expressions]
  \label{thm:subj-reduction}
  Assume $\ExpTy{\Venv}{e}{t}$ and $\Dom{\Senv} \cap \Dom{\Venv} = \emptyset$ and $\FunEnv \Accord \Senv$.
  If $\Reduce{\FunEnv}{e}{\estar}$, then $\ExpTy{\Venv'}{\estar}{t}$ for
  some $\Venv'$ with $\Venv' \IsSubty \Venv$.
\end{theorem}

\begin{proof}
  \renewcommand\LabelQualifier{thm:subj-reduction}
  By induction on the derivation of $\ExpTy{\Venv}{e}{\Monoty}$.
  \begin{CaseDistinction}{on the last rule of the derivation}
    \CdCase{\Rule{d-var}} Then $x \notin \Dom{\Senv}$, so $x \notin \Dom{\FunEnv}$ because
    of $\FunEnv \Accord \Senv$. But then $e = x$ does not reduce.

    \CdCase{\Rule{d-var-poly}} Then $e = x \in \Dom{\Senv}$ and $t \in \Inst{\Senv(x)}$.
    With $\Reduce{\FunEnv}{x}{\estar}$, we get by inverting rule \Rule{red-var} that
    either $(x = v) \in \FunEnv$ or $(x : \sigma = v) \in \FunEnv$ and that $\estar = v$.
    \begin{itemize}
    \item $(x = v) \in \FunEnv$.
      Then by $\FunEnv \Accord \Senv$ there exists some monotype $u$ such that
      \begin{igather}
        \ExpTy{\EmptyEnv}{v}{u} \\
        \Senv(x) = \forall A . u \quad \textrm{with}~A = \FreeTyVars{u}
      \end{igather}
      Hence, $t = u\Tysubst$ for some $\Tysubst$ with $\Dom{\Tysubst} = A$. By \Cref{lem:type-subst}
      $\ExpTy{\EmptyEnv}{v}{t}$
      (note $\FreeTyVars{\Senv} = \emptyset$ because of $\FunEnv \Accord \Senv$).
      By \Cref{lem:extra-var} then $\ExpTy{\Venv}{v}{t}$.
    \item $(x : \sigma = v) \in \FunEnv$.
      Then by $\FunEnv \Accord \Senv$ there exist monotypes $t_{i \in I}$ such that
      \begin{igather}
        \Senv(x) = \forall A . \InterBig_{i \in I} t_i \quad
        \textrm{with}~A \subseteq \UnionBig_{i \in I} \FreeTyVars{t_i} \\
        \ExpTy{\EmptyEnv}{v}{t_i} \quad (\forall i \in I)
      \end{igather}
      Hence, $t = \InterBig_{i \in I} (t_i\Tysubst)$ for some $\Tysubst$ with $\Dom{\Tysubst} = A$.
      By \Cref{lem:type-subst} $\ExpTy{\EmptyEnv}{v}{t_i\Tysubst}$ for all $i \in I$.
      By \Cref{lem:extra-var} and \Cref{lem:value-has-both-types} then $\ExpTy{\Venv}{v}{t}$.
    \end{itemize}

    \CdCase{\Rule{d-const}, \Rule{d-abs}}
    These cases do not occur, since these expressions do not reduce.

    \CdCase{\Rule{d-app}} Hence $e = e_1\,e_2$ with $\ExpTy{\Venv}{e_1}{t' \to t}$ and $\ExpTy{\Venv}{e_2}{t'}$.
    There are three cases:
    \begin{enumerate}
    \item $\Reduce{\FunEnv}{\App {e_1}{e_2}}{\App {e_1'}{e_2}}$ with $\Reduce{\FunEnv}{e_1}{e_1'}$.
      By \IH{}, $\ExpTy{\Venv'}{e_1'}{t' \to t}$ for $\Venv' \IsSubty \Venv$ and we can apply again \Rule{d-app}
      with \Cref{lem:weakening}.
    \item $\Reduce{\FunEnv}{\App{e_1}{e_2}}{\App{e_1}{e_2'}}$ and $e_1$ is a value and
      and $\Reduce{\FunEnv}{e_2}{e_2'}$.
      By \IH{}, $\ExpTy{\Venv'}{e_2'}{t'}$ for $\Venv' \IsSubty \Venv$ and we can apply again \Rule{d-app}
      with \Cref{lem:weakening}.
    \item $e_1 = \lambda x . e'$ and $e_2 = v_2$ for some $x, e', v_2$.
      Further, $\Reduce{\FunEnv}{\App{(\lambda x . e')}{v_2}}{e'[v_2/x]}$ with $\estar = e'[v_2/x]$.

      With assumption $\ExpTy{\Venv}{\lambda x . e'}{t' \to t}$ and \Cref{lem:value-type-determines-shape} then
      $\ExpTy{\Venv,x:\Monoty'}{e'}{\Monoty}$.
      With assumption $\ExpTy{\Venv}{v_2}{t'}$
      and \Cref{lem:expr-subst-single-mono} then $\ExpTy{\Venv}{\estar}{t}$ as required.
    \end{enumerate}

    \CdCase{\Rule{d-pair}} $e = \Pair{e_1}{e_2}$ with $\ExpTy{\Venv}{e_i}{t_i}$ for $i = 1,2$ and
    $t = \Pairty{t_1}{t_2}$. There are three cases:
    \begin{enumerate}
    \item $\Reduce{\FunEnv}{\Pair{e_1}{e_2}}{\Pair{e_1'}{e_2}}$ with $\Reduce{\FunEnv}{e_1}{e_1'}$.
      By \IH{} $\ExpTy{\Venv'}{e_1'}{\Monoty_1}$ for $\Venv' \IsSubty \Venv$ and we can apply again \Rule{d-pair}
      with \Cref{lem:weakening}.
    \item $\Reduce{\FunEnv}{\Pair{v_1}{e_2}}{\Pair{v_1}{e_2'}}$ with $\Reduce{\FunEnv}{e_2}{e_2'}$.
      By \IH{} $\ExpTy{\Venv'}{e_2'}{\Monoty_2}$ for $\Venv' \IsSubty \Venv$ and we can apply again \Rule{d-pair}.
      with \Cref{lem:weakening}.
    \item $e = {\Pair{v_1}{v_2}}$. Impossible, since values do not reduce.
    \end{enumerate}

    \CdCase{\Rule{d-case}} $e = \Case{e'}{(\Pg_i \to e_i)_{i \in I}}$ and we get by inverting rule \Rule{d-case}
    \begin{igather}
      \ExpTy{\Venv}{e'}{t'} \QLabel{eq:ep-tp}\\
      t' \IsSubty \UnionBig_{i \in I} \PgLowerTy{\Pg_i} \QLabel{eq:tp-sub}
    \end{igather}
    and for all $i \in I$:
    \begin{igather}
      \Free{\Pg_i} \subseteq \Dom{\Senv} \cup \Dom{\Venv} \QLabel{eq:free-pgi} \\
      t_i = (t' \WithoutTy \UnionBig_{j < i} \PgLowerTy{\Pg_j}) \Inter \PgUpperTy{\Pg_i} \QLabel{eq:def-ti}\\
      \Venv_i = \Venv \InterEnv (\PatEnv{t_i}{\Pg_i}) \QLabel{eq:def-venvi}\\
      {
        \begin{cases}
          \Monoty_i' = \ErlBot ~\textrm{and}~
          \Free{e_i} \subseteq \Dom{\Senv} \cup \Dom{\Venv_i}
          & \textrm{if}~ \Monoty_i \IsSubty \ErlBot
          \\
          \ExpTy{\Venv_i}{e_i}{\Monoty_i'}
          & \textrm{otherwise}
        \end{cases}
      } \QLabel{eq:prop-ti}\\
      t = \UnionBig_{i \in I}{t_i'} \QLabel{eq:eq-t}
    \end{igather}
    With $\Reduce{\FunEnv}{e}{\estar}$, we must distinguish two cases.

    \begin{enumerate}
    \item $\Reduce{\FunEnv}{e'}{e''}$ and $\estar = \Case{e''}{(\Pg_i \to e_i)_{i \in I}}$
      With \QRef{eq:ep-tp} and the \IH{} then $\ExpTy{\Venv'}{e''}{t'}$ for $\Venv' \IsSubty \Venv$.
      Define $\Venv_i' = \Venv' \InterEnv (\PatEnv{t_i}{\Pg_i})$, so $\Venv_i' \IsSubty \Venv_i$.
      \begin{itemize}
      \item If $t_i \IsSubty \ErlBot$, define $t_i' = \ErlBot$ and have
        $\Free{e_i} \subseteq \Dom{\Senv} \cup \Dom{\Venv_i'}$ by \QRef{eq:prop-ti}.
      \item Otherwise, we have with \QRef{eq:prop-ti} and \Cref{lem:weakening}
        $\ExpTy{\Venv_i'}{e_i}{\Monoty_i'}$
      \end{itemize}
      Then with \Rule{d-case} $\ExpTy{\Venv'}{\estar}{t}$ as required.
    \item $e'$ is a value $v$. Then we get by inverting reduction rule \Rule{red-case} that
      there exists some $j \in I$ and some $\Valsubst$ such that
      \begin{igather}
        \PatSubst{v}{\Pg_j} = \Valsubst \QLabel{eq:v-pgj}\\
        \PatSubst{v}{\Pg_i} = \PatSubstFail \quad (\forall i < j) \QLabel{eq:v-fail}\\
        \estar = e_j \Valsubst
      \end{igather}
      From \QRef{eq:v-pgj} we get with \Cref{lem:gp-typing}(ii) and
      from \QRef{eq:v-fail} with \Cref{lem:gp-typing}(iii)
      \begin{igather}
        \ExpTy{\Venv}{v}{\PgUpperTy{\Pg_j}} \QLabel{eq:v-upty}\\
        \ExpTy{\Venv}{v}{\Neg\PgLowerTy{\Pg_i}} \quad (\forall i < j)
      \end{igather}
      With $e' = v$ and \QRef{eq:ep-tp}, \QRef{eq:tp-sub}, \QRef{eq:def-ti} then
      \begin{igather}
        \ExpTy{\Venv}{v}{t_j} \QLabel{eq:v-tj}
      \end{igather}
      With \Cref{lem:value-not-bot} we know $t_j \not\IsSubty \ErlBot$. Hence with \QRef{eq:prop-ti}
      \begin{igather}
        \ExpTy{\Venv_j}{e_j}{t_j'} \QLabel{eq:ejtj}
      \end{igather}
      With \QRef{eq:v-tj}, \QRef{eq:v-pgj}, the assumption $\FunEnv \Accord \Senv$, and
      \Cref{lem:correct-env-gen}, we get for all $x \in \Dom{\Valsubst}$:
      \begin{igather}
        \ExpTy{\Venv}{\Valsubst(x)}{(\PatEnv{t_j}{\Pg_j})(x)} \QLabel{eq:xsubst-ty}
      \end{igather}
      Our goal is now to give a type to $\estar = e_j \Valsubst$. For this, we define
      \newcommand\VenvT{\widetilde{\Venv}_j}
      \newcommand\VenvS{\Venv^\star_j}
      \newcommand\PgEnv{\PatEnv{t_j}{\Pg_j}}
      \begin{igather}
        \VenvT =
        \Venv \InterEnv \{x : (\PgEnv)(x) \mid x \in \Dom{\PgEnv}, x \notin \PatBound{p_j}, x \notin \Dom{\Senv} \}\\
        \VenvS = \{x : (\PgEnv)(x) \mid x \in \Dom{\PgEnv}, x \in \PatBound{p_j} \}\\
      \end{igather}
      \Wlog{},
      \begin{igather}
        \Dom{\VenvT} \cap \Dom{\VenvS} = \emptyset \QLabel{eq:dom-disj}
      \end{igather}
      We have by definition of $\Venv_j$ that $\Venv_j = \VenvT, \VenvS$ (note that
      $\Dom{\Venv} \cap \Dom{\Senv} = \emptyset$ by assumption). From \QRef{eq:ejtj} then
      \begin{igather}
        \ExpTy{\VenvT,\VenvS}{e_j}{t_j'} \QLabel{eq:ejtj2}
      \end{igather}
      With \Cref{lem:pat-vars}
      \begin{igather}
        \Dom{\Valsubst} = \PatBound{p_j} = \Dom{\VenvS} \QLabel{eq:doms}
      \end{igather}
      Hence, with \QRef{eq:xsubst-ty}
      $\ExpTy{\Venv}{\Valsubst(x)}{\VenvS(x)}$ for all $x \in \Dom{\Valsubst}$.
      We later show that $\VenvT \IsSubty \Venv$. Then by \Cref{lem:weakening}
      \begin{igather}
        \ExpTy{\VenvT}{\Valsubst(x)}{\VenvS(x)} \quad (\forall x \in \Dom{\Valsubst}) \QLabel{eq:etax-ty}
      \end{igather}
      Applying \Cref{lem:expr-subst} to \QRef{eq:dom-disj}, \QRef{eq:ejtj2}, \QRef{eq:doms}, and \QRef{eq:etax-ty} yields
      \begin{igather}
        \ExpTy{\VenvT}{e_j\Valsubst}{t_j'}
      \end{igather}
      With \QRef{eq:eq-t} and $\estar = e_j\Valsubst$ and rule \Rule{d-sub} then
      \begin{igather}
        \ExpTy{\VenvT}{\estar}{t}
      \end{igather}
      Using $\VenvT$ as the $\Venv'$ mentioned in the theorem finishes this case.

      We still need to show $\VenvT \IsSubty \Venv$.
      We have $\Dom{\VenvT} = \Dom{\Venv}$ by the following reasoning:
      \begin{itemize}
      \item $\Dom{\VenvT} \subseteq \Dom{\Venv}$. Assume $x \in \Dom{\VenvT}$. The interesting case is where
        $x \in \Dom{\PgEnv}$ and $x \notin \PatBound{p_j}$. Then $x \in \Free{\Pg_j}$, so $x \in \Dom{\Venv}$
        by \QRef{eq:free-pgi} and definition of $\VenvT$.
      \item $\Dom{\Venv} \subseteq \Dom{\VenvT}$. This holds by construction of $\VenvT$.
      \end{itemize}
      Now assume $x \in \Dom{\VenvT}$. If $x \notin \Dom{\PgEnv}$ then obviously $\VenvT(x) = \Venv(x)$.
      Otherwise $\VenvT(x) = \Venv(x) \Inter (\PgEnv)(x) \IsSubty \Venv(x)$.
    \end{enumerate}

    \CdCase{\Rule{d-sub}} Follows from the \IH{} and rule \Rule{d-sub}.
  \end{CaseDistinction}
\end{proof}

\begin{theorem}[Subject reduction for programs]
  \label{thm:subj-reduction-prog}
  If $\ProgTy{\Prog}{\Monoty}$ and $\Reduce{}{\Prog}{\Prog'}$, then $\ProgTy{\Prog'}{\Monoty}$.
\end{theorem}

\begin{proof}
  Assume $\Prog = \Letrec{\FunEnv}{e}$. Then by inverting
  rule $\Rule{red-prog}$ we have $\Prog' = \Letrec{\FunEnv}{e'}$ and
  $\Reduce{\FunEnv}{e}{e'}$.
  From $\ProgTy{\Prog}{\Monoty}$ we get by inverting rule \Rule{d-prog}:
  \begin{igather}
    \FunEnv = \DefSym_{i \in I}\\
    \Envs{\DefSym_i} = \MetaPair{\Senv_i}{\Senv_i'} \quad (\forall i \in I)\\
    \DefOk{\Senv_{i \in I}}{\DefSym_i} \quad (\forall i \in I)\\
    \ExpSchemaTy{\Senv'_{i \in I}}{\EmptyEnv}{e}{\Monoty}
  \end{igather}

  Therefore, $\FunEnv \Accord \Senv'_{i \in I}$.
  With \Cref{thm:subj-reduction} then $\ExpSchemaTy{\Senv'_{i \in I}}{\EmptyEnv}{e'}{\Monoty}$.
  By \Rule{d-prog}, the claim holds.
\end{proof}

\begin{corollary}[Type soundness]
  \label{thm:soundness}
  Let $\Prog = \Letrec{\FunEnv}{e}$ such that $\ProgTy{\Prog}{t}$ for some $t$.
  Then, either $\Prog$ diverges or reduces to $\Prog' = \Letrec{\FunEnv}{v}$ for some value $v$
  with $\ProgTy{\Prog'}{\Monoty}$.
\end{corollary}

\begin{proof}
  The claim follows directly from \Cref{thm:progress-prog} and \Cref{thm:subj-reduction-prog}.
\end{proof}


\FloatBarrier
\clearpage
\section{Soundness and Completeness of Algorithmic System}
\label{sec:appendix-algorithmic}

This appendix contains the proofs for showing that the algorithmic
system for type reconstruction
via constraint generation (\Cref{f:constr-gen}),
constraint rewriting (\Cref{f:constr-rew}), and constraint solving (tally)
is sound and complete with respect to the declarative type system (\Cref{f:decl-typing}).

\renewcommand\showfreshtyvars{-fresh}

\Cref{f:constr-gen-fresh} and~\Cref{f:constr-rew-fresh}
contain alternative definitions of the constraint generation and rewriting judgments
from \Cref{f:constr-gen} and~\Cref{f:constr-rew} in the main text.
These alternative definitions explicitly track fresh type variables in type variable sets
$A$. Otherwise, the alternative definitions are identical
to the definitions given in the main text. We employ the convention that
a missing type variable set indicates that there exist such a set but we do
not name it explicitly.
For example, the notation $\ConstrGenV{e}{\Monoty}{\Constr}{}$
implies that there exist some type variable set $\TyvarSet$ such that
$\ConstrGenV{e}{\Monoty}{\Constr}{\TyvarSet}$.

\subsection{Soundness}

We first show that if a program is typeable via the algorithmic system, then it is also
typeable via the declarative system.

\begin{lemma}[Soundness of environment generation for patterns]\label{lem:p-env}
  Given some type $\Monoty$ and pattern $p$.
  Assume $\PatTyEnvConstrV{\Monoty}{p}{\Constr}{\Venv'}{}$ and
  $\ConstrRewV{\Venv}{\Constr}{\SiConstrAux}{}$ and
  $\SubstSolves{\Tysubst}{\SiConstrAux}$ for some
  $\Constr, \Senv, \Venv, \Venv', \SiConstrAux$, and $\Tysubst$.
  Define $\Venv^\star = \PatEnv{\Monoty\Tysubst}{p}$.
  Then $\Venv\Tysubst \InterEnv \Venv^\star \IsSubty (\Venv \InterEnv \Venv')\Tysubst$.
\end{lemma}

\begin{proof}
  \renewcommand\LabelQualifier{lem:p-env}
  By induction on the structure of $p$.
  \begin{CaseDistinction}{on the structure of $p$}
    \CdCase{$p = \Wildcard$ or $p = v$ or $p = \Capt{x}$}
    Then $\Constr = \ConstrTrue, \Venv' = \emptyset, \PatEnv{\Monoty\Tysubst}{p} = \emptyset = \Venv^\star$.
    The claim follows immediately 
    \CdCase{$p = y$ for some variable $y$}
    Then $\Constr = \ConstrTrue$ and $\Venv' = \{y:\Monoty\}$.
    Further $\PatEnv{\Monoty\Tysubst}{p} = \{y:\Monoty\Tysubst\} = \Venv^\star$.
    Obviously,
    \begin{igather}
      \Dom{\Venv \InterEnv \Venv'} = \Dom{\Venv\Tysubst \InterEnv \Venv^\star} 
    \end{igather}
    Assume $x \in \Dom{\Venv \InterEnv \Venv'}$. We have to show
    $(\Venv\Tysubst \InterEnv \Venv^\star)(x) \IsSubty (\Venv \InterEnv \Venv')\Tysubst(x)$
    \begin{itemize}
    \item If $x \neq y$, we have
      $(\Venv\Tysubst \InterEnv \Venv^\star)(x) = \Venv(x)\Tysubst = (\Venv \InterEnv \Venv')\Tysubst(x)$.
    \item If $x = y$ define
      \begin{igather}
        u = \begin{cases}
          \Venv(y)  & \textrm{if}~y \in \Dom{\Venv}\\
          \top      & \textrm{otherwise}
        \end{cases}
      \end{igather}
      Then
      $(\Venv\Tysubst \InterEnv \Venv^\star)(x) = u\Tysubst \Inter t\Tysubst = (u \Inter t)\Tysubst = (\Venv \InterEnv \Venv')\Tysubst(x)$.
    \end{itemize}
    \CdCase{$p = (p_1, p_2)$}
    We have
    \begin{igather}
      \PatTyEnvConstr{\Monoty}{p}
      {\Constr_1 \ConstrAnd \Constr_2 \ConstrAnd \Monoty \IsSubtyConstr \Pairty{\Tyvar_1}{\Tyvar_2}}
      {\Venv'}
      \\
      \Venv' = {\Venv_1 \InterEnv \Venv_2}\\
      \Tyvar_i~\textrm{fresh}\quad
      \PatTyEnvConstr{\Tyvar_i}{p_i}{\Constr_i}{\Venv_i} \quad \ConstrRew{\Venv}{\Constr_i}{\SiConstrAux_i}
      \quad \SubstSolves{\Tysubst}{\SiConstrAux_i} \quad (i=1,2)\\
      C = C_1 \ConstrAnd C_2 \quad d = d_1 \ConstrAnd d_2 \\
      \SubstSolves{\Tysubst}{\Monoty \IsSubtyConstr \Pairty{\Tyvar_1}{\Tyvar_2}} \QLabel{eq:t-subty}
    \end{igather}
    Applying the \IH{} yields for $i=1,2$
    \begin{igather}
      (\Venv\Tysubst \InterEnv (\PatEnv{\Tyvar_i\Tysubst}{p_i})) \IsSubty
      (\Venv \InterEnv \Venv_i)\Tysubst \QLabel{eq:ih}
    \end{igather}
    We have by definition
    $\Venv^\star = \PatEnv{\Monoty\Tysubst}{p} = (\PatEnv{\ProjL{\Monoty\Tysubst}}{p_1}) \InterEnv (\PatEnv{\ProjR{\Monoty\Tysubst}}{p_2})$.
    Then
    \begin{igather}
      \Dom{\Venv\Tysubst \InterEnv \Venv^\star} =
      \Dom{\Venv\Tysubst \InterEnv \Venv_1 \InterEnv \Venv_2} =
      \Dom{(\Venv \InterEnv \Venv')\Tysubst}
    \end{igather}
    Assume $x \in \Dom{(\Venv \InterEnv \Venv')\Tysubst}$. Then we have
    \begin{igather}
      (\Venv\Tysubst \InterEnv \Venv^\star)(x) =
      (\Venv\Tysubst \InterEnv (\PatEnv{\ProjL{\Monoty\Tysubst}}{p_1}) \InterEnv (\PatEnv{\ProjR{\Monoty\Tysubst}}{p_2}))(x)
      \ReasonAbove{\IsSubty}{\textrm{see below}}\\
      (\Venv\Tysubst \InterEnv (\PatEnv{\Tyvar_1\Tysubst}{p_1}) \InterEnv (\PatEnv{\Tyvar_1\Tysubst}{p_2}))(x)
      \ReasonAbove{\IsSubty}{\textrm{by}~\QRef{eq:ih}}
      (\Venv\InterEnv\Venv_1\InterEnv\Venv_2)\Tysubst(x) = (\Venv \InterEnv \Venv')\Tysubst(x)
    \end{igather}
    as required. The step marked with \enquote{see below} follows with \Cref{lem:env-pattern-sub} and
    the fact that $\ProjI{\Monoty\Tysubst} \IsSubty \Tyvar_i\Tysubst$. This fact holds because we
    have from \QRef{eq:t-subty} that $t\Tysubst \IsSubty \Pairty{\Tyvar_1\Tysubst}{\Tyvar_2\Tysubst}$,
    so by \Cref{prop:projection}(ii) then
    $\ProjI{\Monoty\Tysubst} \IsSubty \Tyvar_i\Tysubst$.
  \end{CaseDistinction} \qed
\end{proof}

\begin{lemma}[Properties of environments for guard patterns]\label{lem:gp-props}~
  \begin{EnumThm}
  \item $\Dom{\PatEnv{t}{p}} = \PatBound{p}$
  \item $\Dom{\PatEnv{t}{\WithGuard{p}{g}}} = \PatBound{p} \cup \Free{g}$
  \item If $\PatTyEnvConstrV{\Monoty}{p}{\Constr}{\Venv}{}$, then $\Dom{\Venv} = \PatBound{p}$.
  \item If $\PatTyEnvConstrV{\Monoty}{\WithGuard{p}{g}}{\Constr}{\Venv}{}$, then
    $\Dom{\Venv} = \PatBound{p} \cup \Free{g}$.
  \item
    $\FreeTyVars{\PatTy{p}{\Venv}} \subseteq \FreeTyVars{\Venv}$ and
    $\FreeTyVars{\Env{g}} = \FreeTyVars{\PgUpperTy{\Pg}} = \FreeTyVars{\PgLowerTy{\Pg}} = \emptyset$
  \end{EnumThm}
\end{lemma}

\begin{proof}
  \begin{EnumThm}
  \item Follows by \Cref{lem:pat-vars}(iii).
  \item Follows with (i) and the fact that $\Dom{\Env{g}} = \Free{g}$.
  \item Follows by induction on the structure of $p$.
  \item Follows with (iii) and the fact that $\Dom{\Env{g}} = \Free{g}$.
  \item
    $\FreeTyVars{\PatTy{p}{\Venv}} \subseteq \FreeTyVars{\Venv}$ follows by induction on the
    structure of $p$, using the fact that $\TyOfConst{\Const}$ is closed for all $\Const.$
    $\FreeTyVars{\Env{g}} = \emptyset$ follows by induction on the structure of $g$.
    $\FreeTyVars{\PgUpperTy{\Pg}} = \FreeTyVars{\PgLowerTy{\Pg}} = \emptyset$ follows
    with the two preceding equations. \qed
  \end{EnumThm}
\end{proof}

\begin{lemma}\label{lem:interenv-assoc-comm}
  The $\InterEnv$ operator is associative and commutative.
\end{lemma}

\begin{proof}
  Follows by associativity and commutativity of the intersection operator on types. \qed
\end{proof}

\begin{lemma}[Soundness of environment generation for one guarded pattern]\label{lem:pg-env}
  Given some type $\Monoty$ and guarded pattern $\Pg$.
  Assume $\PatTyEnvConstrV{\Monoty}{\Pg}{\Constr}{\Venv'}{}$ and
  $\ConstrRewV{\Venv}{\Constr}{\SiConstrAux}{}$ and
  $\SubstSolves{\Tysubst}{\SiConstrAux}$ for some
  $\Constr,\Senv,\Venv,\Venv',\SiConstrAux$, and $\Tysubst$.
  Define $\Venv^\star = \PatEnv{\Monoty\Tysubst}{\Pg}$.
  Then $\Venv\Tysubst \InterEnv \Venv^\star \IsSubty (\Venv \InterEnv \Venv')\Tysubst$.
\end{lemma}

\begin{proof}
  \renewcommand\LabelQualifier{lem:pg-env}
  Assume $\Pg = \WithGuard{p}{g}$.
  We have
  \begin{igather}
    \Venv^\star = \PatEnv{\Monoty\Tysubst}{\Pg} = (\PatEnv{t\Tysubst}{p}) \InterEnv \Env{g}
  \end{igather}
  Further, by inverting rule \Rule{c-pg-env} used to derive $\PatTyEnvConstrV{\Monoty}{\Pg}{\Constr}{\Venv'}{}$:
  \begin{igather}
    \PatTyEnvConstrV{\Monoty}{p}{\Constr}{\Venv''}{} \\ 
    \Venv' = \Venv'' \InterEnv \Env{g} \QLabel{eq:2}
  \end{igather}
  Hence with \Cref{lem:p-env}
  \begin{igather}
    (\Venv\Tysubst \InterEnv (\PatEnv{t\Tysubst}{p})) \IsSubty (\Venv \InterEnv \Venv'')\Tysubst \QLabel{eq:1}
  \end{igather}
  Thus we get the desired result:
  \begin{igather}
    \Venv\Tysubst \InterEnv \Venv^\star =
    \Venv\Tysubst \InterEnv (\PatEnv{\Monoty\Tysubst}{p}) \InterEnv \Env{g}
    \ReasonAbove{\IsSubty}{\QRef{eq:1}~\textrm{and}~\Cref{lem:interenv-assoc-comm}}
    \Venv\Tysubst \InterEnv \Venv''\Tysubst \InterEnv \Env{g} \\
    \ReasonAbove{=}{\QRef{eq:2}~\textrm{and}~\Cref{lem:gp-props}(v)}
    (\Venv \InterEnv \Venv')\Tysubst
  \end{igather}\qed
\end{proof}

\begin{lemma}[Soundness of environment generation for multiple guarded patterns]\label{lem:pg-env-multi}
  Given some type $\Monoty$ and a sequence of guarded patterns $\Pg_{i \in I}$.
  Assume $\PatTyEnvConstrV{\Monoty}{\Pg_i}{\Constr_i}{\Venv_i}{}$ and
  $\ConstrRewV{\Venv}{\Constr_i}{\SiConstrAux_i}{}$ and
  $\SubstSolves{\Tysubst}{\SiConstrAux_i}$ for all $i \in I$ and some
  $\Constr_{i \in I}$, $\Venv_{i \in I}$, $\Senv$, $\Venv$, $\SiConstrAux_{i \in I}$, and $\Tysubst$.

  Define $\Venv' = \Venv \InterEnv \BigInterEnv_{i \in I} \Venv_i$ and
  $\Venv^\star = \Venv\Tysubst \InterEnv \BigInterEnv_{i \in I} (\PatEnv{\Monoty\Tysubst}{\Pg_i})$.
  Then $\Venv^\star \IsSubty \Venv'\Tysubst$.
\end{lemma}

\begin{proof}
  By induction on $|I|$ using \Cref{lem:pg-env}. \qed
\end{proof}

\begin{definition}
  The free expression variables of a constraint $C$, written $\Free{C}$, are defined as follows:
  \[
  \begin{array}{r@{~}l}
    \Free{\SubtyConstr{\Monoty}{\Monoty}} & = \emptyset\\
    \Free{\SubtyConstr{x}{\Monoty}} & = \{x\}\\
    \Free{\DefConstr{\Venv}{\Constr}} & = \Free{C} \setminus \Dom{\Venv}\\
    \Free{\CaseConstr{C}{
      \Multi{(\InConstr{\Venv_i}{C_i}{\hat{C}_i})}
    }}
    & = \Free{C} \cup (\medcup_{i \in I}\Free{\hat{C}_i}) \cup (\medcup_{i \in I}(\Free{C_i} \setminus \Dom{\Venv_i}))
    \\
    \Free{C_1 \ConstrAnd C_2} &= \Free{C_1} \cup \Free{C_2}
  \end{array}
  \]
\end{definition}

\begin{lemma}[Free variables in constraint generation and rewriting]\label{lem:free-constr}~
  \begin{EnumThm}
  \item
    If $\ConstrRew{\Venv}{C}{d}$ then $\Free{C} \subseteq \Dom{\Senv} \cup \Dom{\Venv}$.
  \item
    If $\ConstrGen{e}{t}{C}$ and $\ConstrRew{\Venv}{C}{d}$ then
    $\Free{e} \subseteq \Dom{\Senv} \cup \Dom{\Venv}$.
  \end{EnumThm}
\end{lemma}

\begin{proof}
  The proof of (i) is by induction on the structure of $C$.
  The proof of (ii) is by induction on the structure of $e$. Both proofs are straightforward.
\end{proof}

\begin{lemma}[Soundness of algorithmic typing for expressions]\label{lem:soundness-algo-exp}
  If $\ConstrGen{e}{\Monoty}{\Constr}$ and
  $\ConstrRew{\Venv}{\Constr}{d}$ and
  $\SubstSolves{\Tysubst}{d}$, then
  $\ExpSchemaTy{\Senv\Tysubst}{\Venv\Tysubst}{e}{t\Tysubst}$.
\end{lemma}

\begin{proof}
  \renewcommand\LabelQualifier{lem:soundness-algo-exp}
  By structural induction on $e$.
  \begin{CaseDistinction}{on the form of $e$}
    \CdCase{$e = x$}
    Then $C = \SubtyConstr{x}{t}$.
    We distinguish three cases:
    \begin{itemize}
    \item $\Venv(x) = t'$ and $x \notin \Dom{\Senv}$.
      In this case, $\ConstrRew{\Venv}{\Constr}{d}$
      has been derived by rule \Rule{rc-var}, so $d = \SubtyConstr{t'}{t}$. Hence, $t'\Tysubst \IsSubty t\Tysubst$ 
      and so $\ExpSchemaTy{\Senv\Tysubst}{\Venv\Tysubst}{e}{t\Tysubst}$ with rules \Rule{d-var} and \Rule{d-sub}.
    \item $\Senv(x) = \forall A . t'$ and $x \notin \Dom{\Venv}$.
      In this case, $\ConstrRew{\Venv}{\Constr}{d}$
      has been derived by rule \Rule{rc-var-poly}, so $d = \SubtyConstr{t'}{t}$. Hence, $t'\Tysubst \IsSubty t\Tysubst$.
      By renaming variables in $A$, we may assume that $(\Senv\Tysubst)(x) = \forall A . (t'\Tysubst)$.
      With rule \Rule{d-var-poly} 
      $\ExpSchemaTy{\Senv\Tysubst}{\Venv\Tysubst}{x}{t'\Tysubst}$. The claim now follows with
      rule \Rule{d-sub}.
    \item $x \in \Dom{\Senv}$ and $x \in \Dom{\Venv}$. Impossible, since there is
        no rule to derive $\ConstrRew{\Venv}{\Constr}{d}$.
    \end{itemize}
    \CdCase{$e = \Const$} Straightforward by rule \Rule{d-const} since $\TyOfConst{\Const}$ is closed.

    \CdCase{$e = \lambda x . e_1$} We get by inverting rule \Rule{c-abs} for fresh $\alpha,\beta$:
    \begin{igather}
      \ConstrGen{\lambda x . e_1}{t}{
        \BraceBelow{\DefConstr{\{x : \alpha\}}{C_1} \ConstrAnd \SubtyConstr{(\alpha \to \beta)}{t}}{
          {} = C
        }
      } \\
      \ConstrGen{e_1}{\beta}{C_1} \QLabel{deriv-C1}
    \end{igather}
    Inverting rules \Rule{rc-and} and \Rule{rc-def} yields
    \begin{igather}
      \ConstrRew{\Venv}{\DefConstr{\{x : \alpha\}}{C_1}}{d_1}\\
      d = d_1 \ConstrAnd \SubtyConstr{(\alpha \to \beta)}{t} \QLabel{def-d}\\
      \ConstrRew{\Venv,x:\alpha}{C_1}{d_1} \QLabel{deriv-d1}
    \end{igather}
    Applying the \IH{} to \QRef{deriv-C1} and \QRef{deriv-d1} yields
    $\ExpSchemaTy{\Senv\Tysubst}{\Venv\Tysubst, x:\alpha\Tysubst}{e_1}{\beta\Tysubst}$,
    so with rule \Rule{d-abs}
    \begin{igather}
      \ExpSchemaTy{\Senv\Tysubst}{\Venv\Tysubst}{e}{\alpha\Tysubst \to \beta\Tysubst}
    \end{igather}
    With \QRef{def-d} and the assumption $\SubstSolves{\Tysubst}{d}$ then
    $\alpha\Tysubst \to \beta\Tysubst \IsSubty t\Tysubst$, so by rule \Rule{d-sub}
    $\ExpSchemaTy{\Senv\Tysubst}{\Venv\Tysubst}{e}{t\Tysubst}$ as required.

    \CdCase{$e = e_1\,e_2$}. Follows with the \IH{} 
    \CdCase{$e = (e_1, e_2)$}. Follows with the \IH{} 

    \CdCase{$e = \Case{e_0}{\Multi{(\PatCls{\Pg_i}{e_i})}}$}
    We have $\ConstrGen{\Case{e_0}{\Multi{(\Pg_i \to e_i)}}}{\Monoty}{C}$,
    so by inverting rule \Rule{c-case}
    \begin{igather}
      \ConstrGen{e_0}{\Tyvar}{\Constr_0} \QLabel{eq:e0C0}\\
      C_0' = C_0 \ConstrAnd \MedConstrAnd\nolimits_{i \in I}C_i'' \ConstrAnd
      (\SubtyConstr{\Tyvar}{\UnionBig_{i \in I}{\PgLowerTy{\Pg_i}}}) \\
      \Tyvar,\TyvarAux ~\textrm{fresh}\\
      \Forall{i \in I}\quad
      \Monoty_i = (\Tyvar \WithoutTy \UnionBig_{j < i}{\PgLowerTy{\Pg_j}}) \Inter \PgUpperTy{\Pg_i}\\
      \PatTyEnvConstr{\Monoty_i}{\Pg_i}{C_i''}{\Venv_i} \QLabel{eq:ti1}\\
      \ConstrGen{e_i}{\TyvarAux}{C_i'} \QLabel{eq:eiCip}\\
      u_i = \UnionBig_{j < i}{\PgLowerTy{\Pg_j}} \Union \Neg \PgUpperTy{\Pg_i}\\
      C = (\CaseConstr{C_0'}{\Multi{(\InConstr{\Venv_i}{C_i'}{\SubtyConstr{\Tyvar}{\MonotyAlt_i}})}}) \ConstrAnd
      \SubtyConstr{\TyvarAux}{\Monoty}
  \end{igather}
  Further $\ConstrRew{\Venv}{C}{d}$ so by inverting rules \Rule{rc-case} and \Rule{rc-and}
  \begin{igather}
      \ConstrRew{\Venv}{C_0}{d_0} \QLabel{eq:C0d0}\\
      \Forall{i \in I}~ \ConstrRew{\Venv}{C_i''}{d_i''} \QLabel{eq:dipp}\\
      \ConstrRew{\Venv \InterEnv \Venv_i}{C_i'}{d_i'} \QLabel{eq:Cipdip}\\
      d_0' = d_0 \ConstrAnd \MedConstrAnd\nolimits_{i \in I} d_i'' \ConstrAnd
      (\SubtyConstr{\Tyvar}{\UnionBig_{i \in I}{\PgLowerTy{\Pg_i}}}) \\
      d_i = d_i' \ConstrOr \SubtyConstr{\alpha}{u_i}\\
      d = d_0' \ConstrAnd \MedConstrAnd\nolimits_{i \in I}{d_i} \ConstrAnd \SubtyConstr{\beta}{t}
  \end{igather}
  Our goals to show are
  \begin{igather}
    \ExpSchemaTy{\Senv\Tysubst}{\Venv\Tysubst}{e_0}{\hat{t}_0} ~\textrm{for some}~\hat{t}_0\tag{G1}\\
    \hat{t}_0 \IsSubty \UnionBig_{i \in I}{\PgLowerTy{\Pg_i}} \tag{G2}\\
    (\forall i \in I) \quad
    \hat{t}_i = (\hat{t}_0 \WithoutTy \UnionBig_{j < i}{\PgLowerTy{\Pg_j}}) \Inter \PgUpperTy{\Pg_i}\\
    \Free{e_i} \subseteq \Dom{\Senv} \cup \Dom{\Venv \InterEnv (\PatEnv{\hat{t}_i}{\Pg_i})}
    \tag{G3}\\
    \Free{\Pg_i} \subseteq \Dom{\Senv} \cup \Dom{\Venv} \tag{G4}\\
    {\begin{cases}
        \hat{t_i}' = \bot & \textrm{if}~\hat{t}_i \IsSubty \bot \\
        \ExpSchemaTy{\Senv\Tysubst}{
          \Venv\Tysubst \InterEnv (\PatEnv{\hat{t}_i}{\Pg_i})
        }{e_i}{\hat{t}_i'} & \textrm{otherwise}
      \end{cases}} \tag{G5} \\
    \UnionBig_{i \in I}{\hat{t}_i} \IsSubty t \Tysubst \tag{G6}
  \end{igather}
  The claim then follows with rules \Rule{d-case} and \Rule{d-sub}.

  \begin{description}
  \item[(G1)]
    From $\SubstSolves{\Tysubst}{d}$ and the definition of $d$, we have $\SubstSolves{\Tysubst}{d_0}$.
    Hence, with \QRef{eq:e0C0}, \QRef{eq:C0d0} and the \IH{}, we get
    $\ExpSchemaTy{\Senv\Tysubst}{\Venv\Tysubst}{e_0}{\alpha\Tysubst}$. Define
    $\hat{t}_0 = \alpha\Tysubst$. This proves (G1).
  \item[(G2)]
    From $\SubstSolves{\Tysubst}{d}$ and the definition of $d, d_0'$, we have
    $\SubstSolves{\Tysubst}{\SubtyConstr{\Tyvar}{\UnionBig_{i \in I}{\PgLowerTy{\Pg_i}}}}$.
    With \Cref{lem:gp-props}, we know that $\UnionBig_{i \in I}{\PgLowerTy{\Pg_i}}$ is closed.
    Hence, $\hat{t}_0 = \alpha\Tysubst \IsSubty \UnionBig_{i \in I}{\PgLowerTy{\Pg_i}}$, which proves (G2).
  \item[(G3)]
    By \QRef{eq:eiCip} and \QRef{eq:Cipdip}, we have with \Cref{lem:free-constr} that
    \begin{igather}
      \begin{array}{c}
        \Free{e_i} \subseteq \Dom{\Venv} \cup \Dom{\Venv \InterEnv \Venv_i} =
        \Dom{\Venv} \cup \Dom{\Venv} \cup \Dom{\Venv_i}
      \end{array}
    \end{igather}
    Assuming $\Pg_i = \WithGuard{p_i}{g_i}$, we have with \QRef{eq:ti1}, and
    \Cref{lem:gp-props} that
    \begin{igather}
      \begin{array}{c}
        \Dom{\Venv_i} = \PatBound{p_i}  \cup \Free{g_i} =
        \Dom{\PatEnv{\hat{t}_i}{\Pg_i}}
      \end{array}
    \end{igather}
    This proves (G3).
  \item[(G4)]
    From \QRef{eq:ti1} by inverting rule \Rule{c-pg-env}
    \begin{igather}
      \Free{\Pg_i} \subseteq \Free{C_i''} \ReasonAbove{\subseteq}{\textrm{definition of}~C_0', C} \Free{C}
    \end{igather}
    From the assumption $\ConstrRew{\Venv}{\Constr}{d}$ and \Cref{lem:free-constr}(i), we have
    $\Free{C} \subseteq \Dom{\Senv} \cup \Dom{\Venv}$. This prove (G4).
  \item[(G5)]
    Now suppose $i \in I$.
    From $\SubstSolves{\Tysubst}{d}$ and the definition of $d$, we have
    $\SubstSolves{\Tysubst}{d_i}$ for $d_i = d_i' \ConstrOr \SubtyConstr{\alpha}{u_i}$.
    Hence, either $\SubstSolves{\Tysubst}{d_i'}$ or
    $\SubstSolves{\Tysubst}{\SubtyConstr{\alpha}{u_i}}$. We now show that in both cases (G5) holds.
    \begin{CaseDistinction}{on whether or not $\SubstSolves{\Tysubst}{\SubtyConstr{\alpha}{u_i}}$}
      \CdCase{$\SubstSolves{\Tysubst}{\SubtyConstr{\alpha}{u_i}}$}
      Hence
      \begin{igather}
        \begin{array}{c}
        \hat{t}_0 = \alpha\Tysubst \IsSubty
        u_i \Tysubst =
        (\UnionBig_{j < i}{\PgLowerTy{\Pg_j}} \Union \Neg \PgUpperTy{\Pg_i}) \Tysubst\\
        \ReasonAbove{=}{~\Cref{lem:gp-props}}
          (\UnionBig_{j < i}{\PgLowerTy{\Pg_j}} \Union \Neg \PgUpperTy{\Pg_i}) = u_i \QLabel{eq:ui}
        \end{array}
      \end{igather}
      We now show $\hat{t}_i \IsSubty \bot$, so that we can define $\hat{t}_i' = \bot$.
      We have
      \begin{igather}
        \hat{t}_i = (\hat{t}_0 \WithoutTy \UnionBig_{j < i}{\PgLowerTy{\Pg_j}}) \Inter \PgUpperTy{\Pg_i}
      \end{igather}
      We argue set-theoretically. Assume $\xi \in \hat{t}_i$. Then
      $\xi \in \hat{t}_0$ and
      $\xi \notin \UnionBig_{j < i}{\PgLowerTy{\Pg_j}}$  and
      $\xi \in \PgUpperTy{\Pg_i}$.
      With \QRef{eq:ui} then $\xi \in u_i$, so by definition of $u_i$ then
      $\xi \in \UnionBig_{j < i}{\PgLowerTy{\Pg_j}}$ or
      $\xi \notin \PgUpperTy{\Pg_i}$.
      This is a contradiction, so $\xi \notin \hat{t}_i$, so $\hat{t}_i \IsSubty \bot$.
      \CdCase{not $\SubstSolves{\Tysubst}{\SubtyConstr{\alpha}{u_i}}$}
      Hence $\SubstSolves{\Tysubst}{d_i'}$.
      If $\hat{t}_i \IsSubty \bot$, define $\hat{t}_i' = \bot$. Otherwise, continue as follows.
      With \QRef{eq:eiCip} and \QRef{eq:Cipdip} and the \IH{}, we get
      \begin{igather}
        \ExpSchemaTy{\Senv\Tysubst}{(\Venv \InterEnv\Venv_i)\Tysubst}{e_i}{\beta\Tysubst} \QLabel{eq:beta}
      \end{igather}
      Define
      \begin{igather}
        \Venv_i^\star = \PatEnv{\hat{t}_i}{\Pg_i}
      \end{igather}
      By definition of $t_i, \hat{t}_i, \hat{t}_0 = \alpha\Tysubst$, and with \Cref{lem:gp-props} we have
      $t_i\Tysubst = \hat{t}_i$. From $\SubstSolves{\Tysubst}{d}$ and the definition of $d, d_0'$, we have
      $\SubstSolves{\Tysubst}{d_i''}$. We then get with
      \QRef{eq:ti1}, \QRef{eq:dipp}, and \Cref{lem:pg-env-multi} that
      \begin{igather}
        (\Venv\Tysubst \InterEnv \Venv_i^\star) \IsSubty (\Venv \InterEnv \Venv_i)\Tysubst
      \end{igather}
      With \QRef{eq:beta} and \Cref{lem:weakening} then
      \begin{igather}
        \ExpSchemaTy{\Senv\Tysubst}{\Venv\Tysubst \InterEnv\Venv_i^\star}{e_i}{\beta\Tysubst}
      \end{igather}
      Define $\hat{t}_i' = \beta\Tysubst$. We have with $\SubstSolves{\Tysubst}{d}$ and the definition of $d$
      that $\SubstSolves{\Tysubst}{\SubtyConstr{\beta}{t}}$. Hence $\hat{t}_i' \IsSubty t\Tysubst$.
    \end{CaseDistinction}

    We have in both cases defined some $\hat{t}_i'$ that fulfills the condition of (G5).
  \item[(G6)]
    We have $\hat{t}_i' \IsSubty t\Tysubst$ for all $i \in I$. This proves (G6).
  \end{description}
\end{CaseDistinction} \qed
\end{proof}

\begin{lemma}[Soundness of environment generation for definitions]
  \label{lem:soundness-algo-def}
  Assume $\DefConstrGen{\DefSym}{\Constr}{\Senv}$ and
  $\ConstrRewFull{\Senv'}{\EmptyEnv}{C}{d}{}$ for some $\Senv'$ with $\Senv \subseteq \Senv'$.
  If $\SubstSolves{\Tysubst}{d}$, then $\DefOk{\Senv'\Tysubst}{\DefSym}$.
\end{lemma}

\begin{proof}
  \begin{CaseDistinction}{on the form of $\DefSym$}
    \CdCase{$\DefSym = (x = \lambda y.e)$}
    Then for $\alpha$ fresh by inverting rule \Rule{c-def-no-annot}
    \begin{igather}
      \ConstrGen{\lambda y . e}{\alpha}{C} \\
      \Senv = \{ x:\alpha \}
    \end{igather}
    With \Cref{lem:soundness-algo-exp} then
    $\ExpSchemaTy{\Senv'\Tysubst}{\EmptyEnv}{\lambda y . e}{\alpha\Tysubst}$.
    From $\Senv \subseteq \Senv'$ we have $\Senv'\Tysubst(x) = \alpha\Tysubst$.
    Hence, $\DefOk{\Senv'\Tysubst}{\DefSym}$ by rule \Rule{d-def-no-annot}.

    \CdCase{$\DefSym = (x : \sigma : \lambda y.e)$}
    Then by inverting rule \Rule{c-def-annot}
    \begin{igather}
      \FreeTyVars{\sigma} = \emptyset\\
      \sigma = \TyScm{A}{\InterBig_{i \in I}{(t_i' \to t_i)}}\\
      \Forall{i \in I}\quad \ConstrGen{e}{t_i}{C_i} \QLabel{gen-Ci}\\
      C = \MedConstrAnd_{i \in I} (\DefConstr{ \{y : t_i'\} }{C_i}) \\
      \Senv = \{ x:\sigma \}
    \end{igather}
    Further, by inverting rules \Rule{rc-and} and \Rule{rc-def}
    \begin{igather}
      \Forall{i \in I}\quad \ConstrRewFull{\Senv'}{\{y : t_i'\}}{C_i}{d_i}{} \QLabel{rew-Ci}\\
      d = \MedConstrAnd_{i \in I} d_i
    \end{igather}
    Define $\Senv_i = \Senv', y:t_i'$ for all $i \in I$. Then with \QRef{gen-Ci}, \QRef{rew-Ci},
    and \Cref{lem:soundness-algo-exp}
    \begin{igather}
      \ExpSchemaTy{\Senv_i\Tysubst}{\EmptyEnv}{e}{t_i \Tysubst}
    \end{igather}
    With $\FreeTyVars{\sigma} = \emptyset$ and by assuming $A$ fresh, we get
    $t_i\Tysubst = t_i$ and $t_i'\Tysubst = t_i'$ for all $i \in I$.
    Hence, $\ExpSchemaTy{\Senv'\Tysubst, y:t_i'}{\EmptyEnv}{e}{t_i}$,
    so $\DefOk{\Senv'\Tysubst}{\DefSym}$ by rule \Rule{d-def-annot}.
  \end{CaseDistinction} \qed
\end{proof}

We now restate and prove the soundness theorem for algorithmic typing from
\Cref{sec:algor-typing-rules}.

\paragraph{Theorem~\ref*{lem:soundness-algo-prog}
  \textnormal{(Soundness of algorithmic typing for programs)}}
\emph{
    Given program $\Prog$, some type $t$, program constraint $P$,
    simple constraint $d$, and some type substitution $\Tysubst$.
    Assume $\ConstrGen{\Prog}{\Monoty}{\ProgConstr}$ and
    $\ConstrRewProg{\ProgConstr}{d}$ and
    $\SubstSolves{\Tysubst}{d}$.
    Then $\ProgTy{\Prog}{t\Tysubst}$.
}

\begin{proof}
  \renewcommand\LabelQualifier{lem:soundness-algo-prog}
  Assume $\Prog = \Letrec{\DefSym_{i \in I}}{e}$. Then we get by inverting rule \Rule{c-prog}:
  \begin{igather}
    \DefConstrGen{\DefSym_i}{C_i}{\Senv_i} \quad(\forall i \in I) \QLabel{gen-Ci}\\
    \Senv = \Senv_{i \in I}\\
    C = \MedConstrAnd\nolimits_{i \in I}C_i \QLabel{def-C}\\
    \ConstrGen{e}{t}{C'} \QLabel{gen-Cp}\\
    P = \LetConstr{C}{\Senv}{C'}
  \end{igather}
  By inverting rule \Rule{rc-prog}:
  \begin{igather}
    \ConstrRew{\EmptyEnv}{C}{c} \QLabel{rew-C}\\
    \textrm{there exists}~\TysubstStar \in \Tally{c} \QLabel{tally}\\
    \ConstrRewFull{\Gen{\Senv\TysubstStar}}{\EmptyEnv}{C'}{c'}{} \QLabel{rew-Cp}\\
    d = \Equiv{\TysubstStar} \ConstrAnd c' \QLabel{equiv}
  \end{igather}
  Now examine definitions $\DefSym_{i \in I}$ and the environment $\Senv$ by inverting
  the rules \Rule{c-def-annot} and \Rule{c-def-no-annot}, assuming $|I| = n$ and $\alpha_i$ fresh.
  Note that all $\sigma_i$ are closed.
  Define environments $\Senv'$ and $\Senv''$, where $t_i = \alpha_i\TysubstStar$ for $i = 1,\ldots k$
  and $A_i = \FreeTyVars{t_i}$.
  \[
    \begin{array}{r@{{}={}}l@{~\ldots,~}l@{~}c@{~\ldots,~}l}
      \DefSym_{i \in I} &    (x_1 = e_1),           & (x_k = e_k),            & (x_{k+1} : \sigma_{k+1} = e_{k+1}),   & (x_n: \sigma_n = e_n) \\
      \Senv             & \{ x_1 = \alpha_1,        & x_k = \alpha_k,         & x_{k+1} = \alpha_{k+1},               & x_n = \alpha_n \} \\
      \Senv'            & \{ x_1 = t_1,             & x_k = t_k,              & x_{k+1} = \sigma_{k+1},               & x_n = \sigma_n \} \\
      \Senv''           & \{ x_1 = \forall A_1.t_1, & x_k = \forall A_k.t_k,  & x_{k+1} = \sigma_{k+1},               & x_n = \sigma_n \}
    \end{array}
  \]
  Then $\Senv\TysubstStar = \Senv'$ and $\Gen{\Senv'} = \Senv''$.

  Further, $\Envs{\DefSym_i} = \MetaPair{\Senv_i'}{\Senv_i''}$ for
  all $i \in I$ with $\Senv' = \Senv'_{i \in I}$ and $\Senv'' = \Senv''_{i \in I}$.
  From \QRef{tally} we get $\SubstSolves{\TysubstStar}{c}$. With \QRef{gen-Ci}, the definition of $C$
  in \QRef{def-C}, with \QRef{rew-C} and \Cref{lem:soundness-algo-def} then
  \begin{igather}
    \DefOk{\BraceBelow{\Senv\TysubstStar}{{} = \Senv'}}{\DefSym_i} \quad (\forall i \in I) \QLabel{defi-ok}
  \end{igather}
  From \QRef{gen-Cp}, \QRef{rew-Cp}, \QRef{equiv}, the assumption $\SubstSolves{\Tysubst}{d}$, and \Cref{lem:soundness-algo-exp}, we get
  \begin{igather}
    \ExpSchemaTy{\Gen{\Senv\TysubstStar}\Tysubst}{\EmptyEnv}{e}{t\Tysubst}
  \end{igather}
  Because $\Gen{\Senv\TysubstStar}$ is closed and $\Senv\TysubstStar = \Senv'$ and $\Gen{\Senv'} = \Senv''$ then
  \begin{igather}
    \ExpSchemaTy{\Senv''}{\EmptyEnv}{e}{t\Tysubst}
  \end{igather}
  With \QRef{defi-ok} and rule \Rule{d-prog} then $\ProgTy{\Prog}{t\Tysubst}$ as required. \qed
\end{proof}

\subsection{Completeness}

In this section, we show that a program typeable via the declarative system is also typeable via the
algorithmic system.

\begin{lemma}[Simple constraints from guarded patterns]
  \label{lem:tp-simp-constr}
  If $\PatTyEnvConstr{\Monoty}{\Pg}{\Constr}{\Venv}$ or
  $\PatTyEnvConstr{\Monoty}{p}{\Constr}{\Venv}$ then $C$ is a simple constraint $c$.
\end{lemma}

\begin{proof}
  Straightforward by definition. \qed
\end{proof}

\begin{lemma}[Completeness of environment generation for patterns]
  \label{lem:pat-ty-algo}
  Assume $t\Tysubst \IsSubty \PatTy{p}{\VenvStar}$ for $\VenvStar$ closed. Further
  $\PatTyEnvConstrV{\Monoty}{p}{\Constr}{\Venv}{\TyvarSet}$ for $A$ fresh.
  Then C is a simple constraint and there exists $\Tysubst'$ such that
  $\Dom{\Tysubst'} = A$ and
  $\SubstSolves{\Tysubst \cupdisjoint \Tysubst'}{C}$ and
  $\Venv(x)(\Tysubst \cupdisjoint \Tysubst') \IsSubty (\PatEnv{t\Tysubst}{p})(x)$
  for all $x \in \PatBound{p}$.
\end{lemma}

\begin{proof}
  \renewcommand\LabelQualifier{lem:pat-ty-algo}
  By induction on $p$.
  \begin{CaseDistinction}{on the shape of $p$}
    \CdCase{$p = v$ or $p = \Wildcard$ or $p = \Capt{x}$}
    Then $C = \ConstrTrue$ and $\Venv = \EmptyEnv$ and $A = \emptyset$.
    Choose $\Tysubst'$ as the identity substitution and the claim follows immediately.
    \CdCase{$p = x$}
    Then $C = \ConstrTrue$ and $\Venv = \{x : t\}$ and $A = \emptyset$.
    Choose $\Tysubst'$ as the identity substitution and the claim follows immediately.
    \CdCase{$p = (p_1, p_2)$}
    Inverting rule \Rule{c-pair-env} yields
    \begin{igather}
      \PatTyEnvConstrV{\alpha_i}{p_i}{C_i}{\Venv_i}{A_i} \quad i=1,2 \QLabel{eq1} \\
      \alpha_1, \alpha_2 \textrm{~fresh}\\
      A = A_1 \cupdisjoint A_2 \cupdisjoint \{\alpha_1, \alpha_2\} \\
      C = C_1 \ConstrAnd C_2 \ConstrAnd \SubtyConstr{t}{\Pairty{\alpha_1}{\alpha_2}}\\
      \Venv = \Venv_1 \InterEnv \Venv_2
    \end{igather}
    Define $\TysubstStar := \Tysubst \cupdisjoint [\ProjI{t\Tysubst} / \alpha_i  \mid i=1,2]$.
    We have the assumption
    $t\Tysubst \IsSubty \PatTy{p}{\VenvStar} = \Pairty{\PatTy{p_1}{\VenvStar}}{\PatTy{p_2}{\VenvStar}}$.
    From \Cref{prop:projection}(ii) then
    $\ProjI{t\Tysubst} \IsSubty \PatTy{p_i}{\VenvStar}$ for $i = 1,2$. Hence
    $\alpha_i \TysubstStar \IsSubty \PatTy{p_i}{\VenvStar}$. Thus we can apply the \IH{} to
    \QRef{eq1} and get for $i = 1,2$ that $C_i$ is a simple constraint and there exist $\Tysubst_i$ with
    \begin{igather}
      \Dom{\Tysubst_i} = A_i\\
      \SubstSolves{\TysubstStar \cupdisjoint \Tysubst_i}{C_i}\\
      \Venv_i(x)(\TysubstStar \cupdisjoint \Tysubst_i) \IsSubty
      (\PatEnv{\alpha_i\TysubstStar}{p_i})(x) \quad\textrm{for all}~x \in \PatBound{p_i} \QLabel{eq2}
    \end{igather}
    Now define
    \begin{igather}
      \Tysubst' := \Tysubst_1 \cupdisjoint \Tysubst_2 \cupdisjoint [\ProjI{t\Tysubst} / \alpha_i \mid i=1,2]
    \end{igather}
    Then we have
    \begin{igather}
      \Dom{\Tysubst'} = A \\
      \SubstSolves{\Tysubst \cupdisjoint \Tysubst'}{C_1 \ConstrAnd C_2}\\
      \SubstSolves{\Tysubst \cupdisjoint \Tysubst'}{\SubtyConstr{t}{\Pairty{\alpha_1}{\alpha_2}}}
    \end{igather}
    The last claim follows because
    \begin{igather}
      t(\Tysubst \cupdisjoint \Tysubst') = t\Tysubst
      \ReasonAbove{\IsSubty}{~\textit{\Cref{prop:projection}(i)}}
      \Pairty{\ProjL{t\Tysubst}}{\ProjR{t\Tysubst}} =
      (\Pairty{\alpha_1}{\alpha_2})(\Tysubst \cupdisjoint \Tysubst')
    \end{igather}
    Pick $x \in \PatBound{p}$. We still need to show that
    $(\Venv_1 \InterEnv \Venv_2)(x)(\Tysubst \cupdisjoint \Tysubst') \IsSubty (\PatEnv{t\Tysubst}{p})(x)$.
    \begin{itemize}
    \item Assume $x \in \PatBound{p_1}$ and $x \notin \PatBound{p_2}$. Then
      \begin{igather}
        \begin{array}{c}
          (\Venv_1 \InterEnv \Venv_2)(x)(\Tysubst \cupdisjoint \Tysubst')
          \ReasonAbove{=}{\textit{\Cref{lem:gp-props}}} \Venv_1(x)(\Tysubst \cupdisjoint \Tysubst')
          = \Venv_1(x)(\TysubstStar \cupdisjoint \Tysubst_1)
          \\
          \ReasonAbove{\IsSubty}{\QRef{eq2}} (\PatEnv{\alpha_1\TysubstStar}{p_1})(x)
          = (\PatEnv{\ProjL{t\Tysubst}}{p_1})(x)
          \\
          \ReasonAbove{=}{\textit{\Cref{lem:gp-props}}}
          (\PatEnv{(\ProjL{t\Tysubst}}{p_1}) \InterEnv (\PatEnv{\ProjR{t\Tysubst}}{p_2}))(x)
          = (\PatEnv{t\Tysubst}{p})(x)
        \end{array}
      \end{igather}
    \item Analogously if $x \notin \PatBound{p_1}$ and $x \in \PatBound{p_2}$.
    \item Assume $x \in \PatBound{p_1}$ and $x \in \PatBound{p_2}$. Then
      \begin{igather}
        \begin{array}{c}
          (\Venv_1 \InterEnv \Venv_2)(x)(\Tysubst \cupdisjoint \Tysubst') =
          \Venv_1(x)(\TysubstStar \cupdisjoint \Tysubst_1) \Inter
          \Venv_2(x)(\TysubstStar \cupdisjoint \Tysubst_2)
          \\
          \ReasonAbove{\IsSubty}{\QRef{eq2}}
          (\PatEnv{\alpha_1\TysubstStar}{p_1})(x) \Inter
          (\PatEnv{\alpha_2\TysubstStar}{p_2})(x)
          \\
          = (\PatEnv{\ProjL{t\Tysubst}}{p_1})(x) \Inter
          (\PatEnv{\ProjR{t\Tysubst}}{p_2})(x)
          \\
          = ((\PatEnv{\ProjL{t\Tysubst}}{p_1}) \InterEnv (\PatEnv{\ProjR{t\Tysubst}}{p_2}))(x)
          = (\PatEnv{t\Tysubst}{p})(x)
        \end{array}
      \end{igather}
    \end{itemize}
  \end{CaseDistinction} \qed
\end{proof}

\begin{lemma} [Completeness of environment generation for guarded patterns]
  \label{lem:pg-constr-satisfy}
  Assume $t\Tysubst \leq \PgUpperTy{\Pg}$ with $\Pg = \WithGuard{p}{g}$. Further, assume
  $\PatTyEnvConstrV{\Monoty}{\Pg}{\Constr}{\Venv}{\TyvarSet}$ for $A$ fresh.
  Then $C$ is a simple constraint and there exists $\Tysubst'$  with
  $\Dom{\Tysubst'} = A$ and
  $\SubstSolves{\Tysubst \cupdisjoint \Tysubst'}{C}$ and
  $\Venv(x)(\Tysubst \cupdisjoint \Tysubst') \IsSubty
  (\PatEnv{t\Tysubst}{\Pg})(x)$
  for all $x \in \PatBound{p} \cup \Free{g}$.
\end{lemma}

\begin{proof}
  \renewcommand\LabelQualifier{lem:pg-constr-satisfy}
  From $\PatTyEnvConstrV{\Monoty}{\Pg}{\Constr}{\Venv}{\TyvarSet}$ we have for some $\Venv^\star$
  \begin{igather}
    \PatTyEnvConstrV{\Monoty}{p}{\Constr}{\Venv^\star}{\TyvarSet}\\
    \Venv = \Venv^\star \InterEnv \Env{g}
  \end{igather}
  If $\PgUpperTy{\Pg} = \ErlBot$ then $t\Tysubst = \ErlBot$, else
  $\PgUpperTy{\Pg} = \PatUpperTy{p}{\Env{g}}$. In both cases
  $t\Tysubst \IsSubty \PatUpperTy{p}{\Env{g}}$. By \Cref{lem:gp-props} we have that $\Env{g}$ is closed.
  Then with \Cref{lem:pat-ty-algo}, there exists $\Tysubst'$ with
  \begin{igather}
    \Dom{\Tysubst'} = A\\
    \SubstSolves{\Tysubst \cupdisjoint \Tysubst'}{C} \\
    \VenvStar(x)(\Tysubst \cupdisjoint \Tysubst') \IsSubty
    (\PatEnv{t\Tysubst}{p})(x)
    \quad (\forall x \in \PatBound{p}) \QLabel{eq1}
  \end{igather}
  Pick $x \in \PatBound{p} \cup \Free{g}$. We need to show
  $\Venv(x)(\Tysubst \cupdisjoint \Tysubst') \IsSubty (\PatEnv{t\Tysubst}{\Pg})(x)$.
  \begin{itemize}
  \item Assume $x \notin \Dom{\Env{g}} = \Free{g}$. Then
    $\Venv(x) = \VenvStar(x)$ and
    $(\PatEnv{t\Tysubst}{\Pg})(x) = (\PatEnv{t\Tysubst}{p})(x)$. Hence, the claim follows from \QRef{eq1}.
  \item Assume $x \in \Dom{\Env{g}} = \Free{g}$. Define $\Tysubst'' := \Tysubst \cupdisjoint \Tysubst'$.
    \begin{itemize}
    \item If $x \notin \PatBound{p}$ then
      \begin{igather}
        \begin{array}{c}
          (\Venv \InterEnv \Venv')(x)\Tysubst'' = \Env{g}(x)\Tysubst''
          \ReasonAbove{=}{~\textit{\Cref{lem:gp-props}(v)}} \Env{g}(x) =
          (\PatEnv{t\Tysubst}{\Pg})(x)
        \end{array}
      \end{igather}
    \item If $x \in \PatBound{p}$ then
      \begin{igather}
        \begin{array}{c}
          (\Venv \InterEnv \Venv')(x)\Tysubst'' =
          (\VenvStar \InterEnv \Venv')(x)\Tysubst'' \Inter \Env{g}(x)\Tysubst''
          \\
          \ReasonAbove{=}{~\textit{\Cref{lem:gp-props}(v)}} (\VenvStar \InterEnv \Venv')(x)\Tysubst'' \Inter \Env{g}(x)
        \end{array}
      \end{igather}
      and
      \begin{igather}
        (\PatEnv{t\Tysubst}{\Pg})(x) =
        (\PatEnv{t\Tysubst}{p})(x) \Inter \Env{g}(x)
      \end{igather}
      The claim then follows with \QRef{eq1}.  
      \qed
    \end{itemize}
  \end{itemize}
\end{proof}

\begin{lemma}[Possibility of constraint rewriting]\label{lem:always-constr-rew}
  For all $\Senv, \Venv$, and $C$ with $\Free{C} \subseteq \Dom{\Senv} \cup \Dom{\Venv}$ exists
  $c$ and fresh $A$ with
  $\ConstrRewV{\Venv}{\Constr}{\SiConstr}{A}$.
\end{lemma}

\begin{proof}
  Straightforward induction on the structure of $C$. \qed
\end{proof}

\begin{lemma}[Free variables of constraint generation]\label{lem:free-C-e}
  If $\ConstrGen{e}{t}{C}$ then $\Free{C} \subseteq \Free{e}$.
\end{lemma}

\begin{proof}
  Straightforward induction on the structure of $C$. \qed
\end{proof}

\begin{lemma}\label{lem:bot-implies-sub}
  Let $I$ be some index set with $i \in I$, let
  $u_{k \in I}$ and $t$ be types.
  Assume $t \leq \UnionBig_{k \in I} u_k$ and define
  $t' = (\UnionBig_{j < i}u_j) \Union \Neg u_i$
  and $t'' = (t \WithoutTy  \UnionBig_{j < i}u_j) \Inter u_i$.
  If $t'' \leq \bot$ then $t \leq t'$.
\end{lemma}

\begin{proof}
  We argue set-theoretically. Assume $\xi \in t$ but $\xi \notin t'$. Then
  $\xi \notin \UnionBig_{j < i}u_j$ and $\xi \notin \Neg u_i$. Hence, $\xi \in u_i$.
  But then $\xi \in t''$. This contradicts $t'' \IsSubty \bot$. \qed
\end{proof}

\begin{lemma}[Completeness of algorithmic typing for expressions]\label{lem:completeness-algo-exp}
  Given $\Senv,\Venv,e,t$, and $\Tysubst$. If
  $\ExpSchemaTy{\Senv\Tysubst}{\Venv\Tysubst}{e}{t\Tysubst}$ then there exist $C,d$,
  fresh type variable sets $A_1, A_2$ and
  $\Tysubst'$ such that
  $\ConstrGenV{e}{t}{C}{A_1}$ and
  $\ConstrRewV{\Venv}{C}{d}{A_2}$ and
  $\SubstSolves{\Tysubst \cupdisjoint \Tysubst'}{d}$ and $\Dom{\Tysubst'} = A_1 \cupdisjoint A_2$.
\end{lemma}

\begin{proof}
  By structural induction on $e$
  \begin{CaseDistinction}{on the form of $e$}
    \CdCase{$e = x$} We have with rule \Rule{c-var}
    \begin{igather}
      \ConstrGenV{x}{t}{\BraceBelow{\SubtyConstr{x}{t}}{{} = C}}{\emptyset}
    \end{igather}
    With \Cref{lem:sub-elim} we have
    $\ExpSchemaTy{\Senv\Tysubst}{\Venv\Tysubst}{e}{t'}$ such that the derivation ends with one of the rules
    \Rule{d-var} or \Rule{d-var-poly}. Further, $t' \IsSubty t\Tysubst$.
    \begin{itemize}
    \item If the last rule is \Rule{d-var},
      then $\Venv(x) = t''$ with $t''\Tysubst = t'$ and $x \notin \Dom{\Senv}$. Then with rule \Rule{rc-var}
        \begin{igather}
          \ConstrRewV{\Venv}{\SubtyConstr{x}{t}}{\BraceBelow{\SubtyConstr{t''}{t}}{{}=: d}}{\emptyset}
        \end{igather}
        Let $\Tysubst'$ be the identity substitution. Then $\Dom{\Tysubst'} = \emptyset$ and
        it is straightforward to verify that $\SubstSolves{\Tysubst \cupdisjoint \Tysubst'}{d}$.
    \item If the last rule is \Rule{d-var-poly}, then $\Senv(x) = \forall A_2 . u$ for some $u$ and $A_2$ fresh
      with $t' = u(\Tysubst \cupdisjoint \Tysubst')$ for some $\Tysubst'$ with $\Dom{\Tysubst'} = A_2$.
      Further, $x \notin \Dom{\Venv}$. Hence with rule \Rule{rc-var-poly}
        \begin{igather}
          \ConstrRewV{\Venv}{\SubtyConstr{x}{t}}{\BraceBelow{\SubtyConstr{u}{t}}{{}=: d}}{A_2}
        \end{igather}
        And we have
        $u(\Tysubst \cupdisjoint \Tysubst') = t' \IsSubty t\Tysubst = t(\Tysubst \cupdisjoint \Tysubst')$,
        so $\SubstSolves{\Tysubst \cupdisjoint \Tysubst'}{d}$.
    \end{itemize}

    \CdCase{$e = \Const$} With \Cref{lem:sub-elim}
    $\ExpSchemaTy{\Senv\Tysubst}{\Venv\Tysubst}{\Const}{\TyOfConst{\Const}}$ and
    $\TyOfConst{\Const} \IsSubty t\Tysubst$.
    With rules \Rule{c-const} and \Rule{rc-subty}
    \begin{igather}
      \ConstrGenV{\Const}{t}{\SubtyConstr{\TyOfConst{\Const}}{t}}{\emptyset}\\
      \ConstrRewV{\Venv}{\SubtyConstr{\TyOfConst{\Const}}{t}}{
        \BraceBelow{\SubtyConstr{\TyOfConst{\Const}}{t}}{=: d}
        }{\emptyset}
    \end{igather}
    Choose $\Tysubst'$ as the identity substitution, and $A_1 = A_2 = \emptyset$.
    Noting that $\TyOfConst{\Const}$ is closed, we have
    $\SubstSolves{\Tysubst \cupdisjoint \Tysubst'}{d}$.

    \CdCase{$e = \lambda x . e'$}
    \renewcommand\LabelQualifier{lem:completeness-algo-exp-lambda}
    With \Cref{lem:sub-elim} we get by inverting rule \Rule{d-abs}:
    \begin{igather}
      \ExpSchemaTy{\Senv\Tysubst}{\Venv\Tysubst, x:t_1}{e'}{t_2} \QLabel{type-body} \\
      t_1 \to t_2 \IsSubty t\Tysubst \QLabel{subtyp-res}
    \end{igather}
    Choose fresh $\alpha,\beta$ and define
    \begin{igather}
      \Tysubst_1 := \Tysubst \cupdisjoint [ t_1/\alpha, t_2/\beta ] \QLabel{def-tysubst-1}
    \end{igather}
    Then with \QRef{type-body}:
    $\ExpSchemaTy{\Senv\Tysubst_1}{(\Venv, x:\alpha)\Tysubst_1}{e'}{\beta\Tysubst_1}$.
    Applying the \IH{} now yields the existence of $\Tysubst_2$ with
    \begin{igather}
      \ConstrGenV{e'}{\beta}{C'}{A_1'} \quad (A_1'~\textrm{fresh}) \QLabel{gen-Cp}\\
      \ConstrRewV{\Venv, x:\alpha}{C'}{d'}{A_2} \quad (A_2~\textrm{fresh}) \QLabel{rew-Cp}\\
      \SubstSolves{\Tysubst_1 \cupdisjoint \Tysubst_2}{d'} \QLabel{solve-dp}\\
      \Dom{\Tysubst_2} = A_1' \cupdisjoint A_2 \QLabel{dom-theta2}
    \end{igather}
    With \QRef{gen-Cp} and rule \Rule{c-abs} then
    \begin{igather}
      \ConstrGenV{\lambda x . e'}{t}{
        \BraceBelow{\DefConstr{ \{ x: \alpha \}}{C'} \ConstrAnd \SubtyConstr{(\alpha\to\beta)}{t}}{=: C}
      }{A_1} \QLabel{gen-C}\\
      A_1 = A_1' \cupdisjoint \{\alpha,\beta\} \QLabel{def-A1}
    \end{igather}
    With \QRef{rew-Cp} and rule \Rule{rc-def} now
    $\ConstrRewV{\Venv}{\DefConstr{\{x:\alpha\}}{C'}}{d'}{A_2}$.
    Hence, with \QRef{gen-C} then
    \begin{igather}
      \ConstrRewV{\Venv}{C}{
        \BraceBelow{d' \ConstrAnd \SubtyConstr{(\alpha\to\beta)}{t}}{{}=: d}
      }{A_2} \QLabel{rew-C}
    \end{igather}
    With \QRef{gen-C} and \QRef{rew-C}, we only need to prove that there exists $\Tysubst'$ with
    \begin{igather}
      \SubstSolves{\Tysubst \cupdisjoint \Tysubst'}{d} \tag{G1}\\
      \Dom{\Tysubst'} = A_1 \cupdisjoint A_2 \tag{G2}
    \end{igather}
    Define
    \begin{igather}
      \Tysubst' = [ t_1/\alpha, t_2/\beta ] \cupdisjoint \Tysubst_2 \QLabel{def-tysubst-p}
    \end{igather}
    Then we get goal (G2)
    \begin{igather}
      \Dom{\Tysubst'} = \{\alpha,\beta\} \cupdisjoint \Dom{\Tysubst_2}
      \ReasonAbove{=}{\QRef{dom-theta2}} A_1' \cupdisjoint \{\alpha,\beta\}\cupdisjoint A_2
      \ReasonAbove{=}{\QRef{def-A1}} A_1 \cupdisjoint A_2 \QLabel{dom-theta-p}
    \end{igather}
    Further, we have
    \begin{igather}
      \Tysubst \cupdisjoint \Tysubst' =
      \Tysubst \cupdisjoint [t_1/\alpha, t_2/\beta] \cupdisjoint \Tysubst_2
      \ReasonAbove{=}{\QRef{def-tysubst-1}} \Tysubst_1 \cupdisjoint \Tysubst_2
    \end{igather}
    We get with \QRef{solve-dp} that
    $\SubstSolves{\Tysubst_1 \cupdisjoint \Tysubst_2}{d'}$, so by definition of
    $d$ in \QRef{rew-C}, we only need to verify
    $(\alpha\to\beta)(\Tysubst \cupdisjoint \Tysubst') \IsSubty t (\Tysubst \cupdisjoint \Tysubst')$
    in order to prove (G1). By definition of $\Tysubst'$ in \QRef{def-tysubst-p} we know that
    $(\alpha\to\beta)(\Tysubst \cupdisjoint \Tysubst') = t_1 \to t_2$.
    From \QRef{dom-theta2}, \QRef{dom-theta-p} and the freshness of $\alpha,\beta,A_1',A_2$ we know
    $\FreeTyVars{t} \cap \Dom{\Tysubst'} = \emptyset$. Hence,
    $t(\Tysubst \cupdisjoint \Tysubst') = t\Tysubst$ and
    $(t_1 \to t_2) \IsSubty t\Tysubst$ follows with \QRef{subtyp-res}.

    \CdCase{$e = e_1\,e_2$} Follows with the \IH{} and rules \Rule{c-app} and \Rule{rc-and}. 
    \CdCase{$e = (e_1, e_2)$} Follows with the \IH{} and rules \Rule{c-pair} and \Rule{rc-and}. 

    \CdCase{$e = \Case{e_0}{\Multi{(\PatCls{\Pg_i}{e_i})}}$}
    \renewcommand\LabelQualifier{lem:completeness-algo-exp-case}
    With \Cref{lem:sub-elim}, we get by inverting rule \Rule{d-case}:
    \begin{igather}
      \ExpSchemaTy{\Senv\Tysubst}{\Venv\Tysubst}{e_0}{t_0\Tysubst} \QLabel{e0t0}\\
      t_0 \IsSubty \UnionBig_{i \in I} \PgLowerTy{\Pg_i} \QLabel{t0sub}\\
      \Forall{i \in I}\quad
      \Monoty_i = (t_0 \WithoutTy \UnionBig_{j < i} \PgLowerTy{\Pg_j}) \Inter \PgUpperTy{\Pg_i}
      \QLabel{def-ti}\\
      {
        \begin{cases}
          \Monoty_i' = \ErlBot ~\textrm{and}~
          \Free{e_i} \subseteq \Dom{\Senv\Tysubst} \cup
          \Dom{\Venv\Tysubst
            \InterEnv (\PatEnv{\Monoty_i}{\Pg_i})}
          & \textrm{if}~ \Monoty_i \IsSubty \ErlBot
          \\
          \ExpSchemaTy{\Senv\Tysubst}{\Venv\Tysubst
            \InterEnv (\PatEnv{\Monoty_i}{\Pg_i})
          }{e_i}{t_i'}
          & \textrm{otherwise}
        \end{cases}
      } \QLabel{tip} \\
      \Free{\Pg_i} \subseteq \Dom{\Senv} \cup \Dom{\Venv}
      \\
      \ExpSchemaTy{\Senv\Tysubst}{\Venv\Tysubst}{e}{\UnionBig_{i \in I}{\Monoty_i'}} \\
      t' := \UnionBig_{i \in I}{t_i'} \IsSubty t\Tysubst \QLabel{def-tp}
    \end{igather}
    We derive with rule \Rule{c-case}:
    \begin{igather}
      \ConstrGenV{e}{t}{
        \BraceBelow{
          (\CaseConstr{C_0'}{\Multi{(\InConstr{\Venv''_i}{D_i}{\SubtyConstr{\Tyvar}{u_i'}})}}) \ConstrAnd
          \SubtyConstr{\TyvarAux}{t}
        }{{} = C}
      }{A_1} \QLabel{deriv-C} \\
      A_1 = \{\alpha, \beta\} \cupdisjoint B_0 \cupdisjoint
      \medcupdisjoint_{i \in I}(B_i \cupdisjoint B_i') \QLabel{def-A1}\\
      \ConstrGenV{e_0}{\alpha}{C_0}{B_0} \qquad (\alpha~\textrm{fresh}) \QLabel{deriv-C0}
    \end{igather}
    and for all $i \in I$
    \begin{igather}
      u_i = (\alpha \WithoutTy \UnionBig_{j < i} \PgLowerTy{\Pg_j}) \Inter \PgUpperTy{\Pg_i}\\
      \PatTyEnvConstrV{u_i}{\Pg_i}{C_i}{\Venv_i}{B_i} \QLabel{deriv-Cip}\\
      \ConstrGenV{e_i}{\beta}{D_i}{B_i'} \QLabel{gen-beta}\\
      u_i' = (\UnionBig_{j < i} \PgLowerTy{\Pg_j}) \Union \Neg \PgUpperTy{\Pg_i} \QLabel{def-uip}
    \end{igather}
    All type variable sets $B_0, B_i, B_i'$ are assumed to be fresh. The $C_0'$ in \QRef{deriv-C}
    is defined as follows:
    \begin{igather}
      C_0' =
      C_0 \ConstrAnd \MedConstrAnd\nolimits_{i \in I}C_i \ConstrAnd
      (\SubtyConstr{\Tyvar}{\UnionBig_{i \in I}{\PgLowerTy{\Pg_i}}}) \QLabel{def-C0p}
    \end{igather}
    Define
    \begin{igather}
      \Tysubst^{\star} := \Tysubst \cupdisjoint [ t_0/\alpha ] \QLabel{def-theta-star}
    \end{igather}
    Then we can apply the \IH{} to \QRef{e0t0} and with \QRef{deriv-C0} we get
    \begin{igather}
      \ConstrRewV{\Venv}{C_0}{d_0}{B_0'}\\
      \SubstSolves{\Tysubst^\star \cupdisjoint \Tysubst_0'}{d_0} \QLabel{deriv-d0}\\
      \textrm{for some}~\Tysubst_0'~\textrm{with}~\Dom{\Tysubst_0'} = B_0 \cupdisjoint B_0' \QLabel{dom-theta0p}
    \end{igather}
    With \Cref{lem:tp-simp-constr} and \QRef{deriv-Cip}
    we know that $C_i$ are simple constraints for all $i \in I$.
    Hence with \QRef{def-C0p} and because simple constraints rewrite
    to simple constraints, we get
    \begin{igather}
      \ConstrRewV{\Venv}{C_0'}{
        \BraceBelow{
          d_0 \ConstrAnd \MedConstrAnd\nolimits_{i \in I}C_i \ConstrAnd
          (\SubtyConstr{\Tyvar}{\UnionBig_{i \in I}{\PgLowerTy{\Pg_i}}})
        }{{} = d_0'}
      }{B_0'} \QLabel{deriv-d0p}
    \end{igather}
    For each branch $i \in I$ (note that
    $\Free{\UnionBig_{j < i} \PgLowerTy{\Pg_j}} = \emptyset$
    and $\Free{\PgUpperTy{\Pg_i}} = \emptyset$ by \Cref{lem:gp-props}(v)), we have
    by definition of $t_i$ in \QRef{def-ti}:
    \begin{igather}
      u_i \Tysubst^\star = t_i \IsSubty \PgUpperTy{\Pg_i}
    \end{igather}
    Assume $\Pg_i = \WithGuard{p_i}{g_i}$.
    With \Cref{lem:pg-constr-satisfy} and \QRef{deriv-Cip} there exists
    $\Tysubst^\star_i$ such that
    \begin{igather}
      \Dom{\Tysubst^\star_i} = B_i \QLabel{dom-theta-i-star}\\
      \SubstSolves{\Tysubst^\star \cupdisjoint \Tysubst^\star_i} C_i \QLabel{deriv-Ci-Cip}\\
      \forall x \in \PatBound{p_i} \cup \Free{g_i}:\\
      (\Venv_i \InterEnv \Venv_i' 
      )(x)(\Tysubst^\star \cupdisjoint \Tysubst^\star_i) \IsSubty
      \BraceBelow{(\PatEnv{t_i}{\Pg_i})}{{} =: \Venv^\star_i}(x)
      \QLabel{env-subty}
    \end{igather}
    From the definition of $\Tysubst^\star$ in \QRef{def-theta-star} and from \QRef{t0sub}, we have
    \begin{igather}
      \SubstSolves{\Tysubst^\star}{\SubtyConstr{\alpha}{\UnionBig_{i \in I}{\PgLowerTy{\Pg_i}}}}
    \end{igather}
    Hence also with the definition of $d_0'$ in \QRef{deriv-d0p}, and with \QRef{deriv-d0}, \QRef{deriv-Ci-Cip}
    \begin{igather}
      \SubstSolves{
        \BraceBelow{
          \Tysubst^\star \cupdisjoint \Tysubst_0' \cupdisjoint \medcupdisjoint_{i \in I} \Tysubst^\star_i
        }{{} =: \Tysubst^\dstar} \QLabel{entails-d0p}
      }{d_0'}
    \end{igather}
    Define
    \begin{igather}
      \Tysubst^\top := \Tysubst^\dstar \cupdisjoint [ t'/\beta ] \QLabel{def-subst-T}
    \end{igather}
    Assume $i \in I$. We now prove three goals (G1), (G2), (G3):
    \begin{igather}
      \ConstrRewV{\Venv \InterEnv \Venv_i}{D_i}{d_i}{B_i''} \quad
      \textrm{for some}~d_i~\textrm{ and fresh}~B_i''\tag{G1}
      \\
      \textrm{there exists}~\Tysubst_i'~\textrm{with}~\Dom{\Tysubst_i'} = B_i' \cupdisjoint B_i'' \tag{G2}
      \\
      \SubstSolves{\Tysubst' \cupdisjoint \Tysubst_i'}{d_i} ~\textrm{or}~ t_i \IsSubty \ErlBot  \tag{G3}
    \end{igather}
    \begin{CaseDistinction}{on whether or not $t_i \IsSubty \bot$}
      \CdCase{not $t_i \IsSubty \bot$} We have from \QRef{tip} and the definition of $\Venv_i^\star$ in \QRef{env-subty}:
      \begin{igather}
        \ExpSchemaTy{\Senv\Tysubst}{(\Venv\Tysubst \InterEnv \Venv_i^\star)}{e_i}{t_i'} \QLabel{type-ei1}
      \end{igather}
      Assume the following goal (G4), which we will prove shortly:
      \begin{igather}
        (\Venv \InterEnv \Venv_i) \Tysubst^\top \IsSubty \Venv\Tysubst \InterEnv \Venv_i^\star \tag{G4}
      \end{igather}
      Then we get from \QRef{type-ei1} and \QRef{def-tp} by \Cref{lem:weakening} and rule \Rule{d-sub}
      \begin{igather}
        \ExpSchemaTy{\Senv\Tysubst}{(\Venv \InterEnv \Venv_i)\Tysubst^\top}{e_i}{\beta\Tysubst^\top} \QLabel{type-ei2}
      \end{igather}
      With \QRef{gen-beta}, \QRef{type-ei2}, and the \IH{} then for some $\Tysubst_i'$
      \begin{igather}
        \ConstrRewV{\Venv \InterEnv \Venv_i}{D_i}{d_i}{B_i''}      \\
        \SubstSolves{\Tysubst^\top \cupdisjoint \Tysubst_i'}{d_i}     \\
        \Dom{\Tysubst_i'} = B_i' \cupdisjoint B_i''                 \\
      \end{igather}
      To finish the case $t_i \not\leq \bot$, we still need to prove (G4). That is, we need to show
      $(\Venv \InterEnv \Venv_i) \Tysubst^\top \IsSubty \Venv\Tysubst \InterEnv \Venv_i^\star$.
      We have by \Cref{lem:gp-props} and the definitions of $\Tysubst_i^\star$ and $\Tysubst_i''$ that
      $\Dom{\Venv_i^\star} = \PatBound{p_i} \cup \Free{g_i} = \Dom{\Venv_i}$.
      Now assume $x \in \Dom{\Venv} \cup \Dom{\Venv_i^\star}$.
      We need to show
      $(\Venv \InterEnv \Venv_i) \Tysubst^\top(x) \IsSubty (\Venv\Tysubst \InterEnv \Venv_i^\star)(x)$.
      \begin{itemize}
      \item If $x \in \Dom{\Venv}$ and $x \notin \Dom{\Venv_i^\star} = \Dom{\Venv_i}$ then
        \begin{igather}
          (\Venv \InterEnv \Venv_i) \Tysubst^\top(x) =
          \Venv(x)\Tysubst^\top = \Venv(x)\Tysubst =
          (\Venv\Tysubst \InterEnv \Venv_i^\star)(x)
        \end{igather}
        The second equality holds because
        \begin{igather}
          \begin{array}{c}
            \Tysubst^\top =
            \Tysubst^\dstar \cupdisjoint [ t'/\beta ] =
            \Tysubst^\star \cupdisjoint \Tysubst_0' \cup (\medcupdisjoint_{j \in I}{\Tysubst_j^\star})
            \cupdisjoint [ t'/\beta ] =
            \\
            \Tysubst \cupdisjoint \Tysubst_0' \cup (\medcupdisjoint_{j \in I}{\Tysubst_j^\star})
            \cupdisjoint [t_0/\alpha, t'/\beta] \QLabel{eq-theta-top}
          \end{array}
        \end{igather}
        and all types variables in $\Dom{\Tysubst_0'}, \Dom{\medcupdisjoint_{j \in I}\Tysubst_j^\star}$, as well as
        $\alpha,\beta$ are fresh.
      \item If $x \in \Dom{\Venv}$ and $x \in \Dom{\Venv_i^\star} = \Dom{\Venv_i}$ then
        \[
          \begin{array}{r@{~}c@{~}l}
          (\Venv \InterEnv \Venv_i) \Tysubst^\top(x)
          & = & \textrm{by}~ \QRef{eq-theta-top},
            ~\textrm{freshness of type}
          \\ && \textrm{variables, definition of} \InterEnv{}{}
          \\
          \Venv(x)\Tysubst \Inter \Venv_i(x)\Tysubst^\top
          & \IsSubty & \textrm{by}~\QRef{env-subty}
          \\
          \Venv(x)\Tysubst \Inter \Venv_i^\star(x) &=
          \\
          (\Venv\Tysubst \InterEnv \Venv_i^\star)(x)
          \end{array}
        \]
      \item If $x \notin \Dom{\Venv}$ and $x \in \Dom{\Venv_i^\star} = \Dom{\Venv_i}$ then
        \begin{igather}
          \begin{array}{c}
          (\Venv \InterEnv \Venv_i) \Tysubst^\top(x) = \Venv_i(x)\Tysubst^\top
          \ReasonAbove{\IsSubty}{\QRef{env-subty}} \Venv_i^\star(x) =
            (\Venv\Tysubst \InterEnv \Venv_i^\star)(x)
          \end{array}
        \end{igather}
      \end{itemize}
      This finishes the proof of (G4) and thus the case $t_i \not\leq \bot$.
      \CdCase{$t_i \leq \bot$}
      We have
      \begin{igather}
        \Free{D_i} \ReasonAbove{\subseteq}{\QRef{gen-beta}, \Cref{lem:free-C-e}}
        \Free{e_i} \ReasonAbove{\subseteq}{\QRef{tip}}
        \Dom{\Senv\Tysubst} \cup
        \Dom{\Venv\Tysubst \InterEnv (\PatEnv{\Monoty_i}{\Pg_i})} \\
        \ReasonAbove{=}{\QRef{deriv-Cip}, \Cref{lem:gp-props}}
        \Dom{\Senv\Tysubst} \cup \Dom{\Venv \InterEnv \Venv_i}
      \end{igather}
      With \Cref{lem:always-constr-rew} then
      \begin{igather}
        \ConstrRewV{\Venv \InterEnv \Venv_i}{D_i}{d_i}{B_i''}
      \end{igather}
      for some $d_i$ and fresh $B_i''$.
      In this case, define
      \begin{igather}
        \Tysubst_i' := [ \bot/\alpha \mid \forall \alpha \in B_i' \cupdisjoint B_i'' ]
      \end{igather}
    \end{CaseDistinction}
    This finishes the proof of goals (G1), (G2), and (G3).

    With \QRef{deriv-C}, we have $\ConstrGenV{e}{t}{C}{A_1}$.
    With the definition of $C$ in \QRef{deriv-C}, with \QRef{deriv-d0p}, (G1), and with rule \Rule{rc-case}, we get
    \begin{igather}
      \ConstrRewV{\Venv}{C}{
        \BraceBelow{
          d_0' \ConstrAnd \MedConstrAnd_{i \in I}(d_i \ConstrOr \SubtyConstr{\alpha}{u_i'})
        }{{}=: d}
      }{A_2} \QLabel{deriv-d}\\
      A_2 = B_0' \cupdisjoint \medcupdisjoint_{i \in I}{B_i''} \QLabel{def-A2}
    \end{igather}
    We define $\Tysubst'$ as follows:
    \begin{igather}
      \Tysubst' =: [ t_0/\alpha, t'/\beta ] \cupdisjoint
      \Tysubst_0' \cupdisjoint \medcupdisjoint_{i \in I}(\Tysubst_i' \cupdisjoint \Tysubst_i^\star) \QLabel{def-theta-p}
    \end{igather}
    and verify that
    $\SubstSolves{\Tysubst \cupdisjoint \Tysubst'}{d}$ and
    $\Dom{\Tysubst'} = A_1 \cupdisjoint A_2$.
    We have
    \begin{igather}
      \Tysubst \cupdisjoint \Tysubst'
      \ReasonAbove{=}{\QRef{def-theta-star}}
      \Tysubst^\star \cupdisjoint [ t'/\beta ] \cupdisjoint
      \Tysubst_0' \cupdisjoint \medcupdisjoint_{i \in I}(\Tysubst_i' \cupdisjoint \Tysubst_i^\star)
      \ReasonAbove{=}{\QRef{entails-d0p}}
      \Tysubst^\dstar \cupdisjoint
      [ t'/\beta ]
      \cupdisjoint \medcupdisjoint_{i \in I}\Tysubst_i' \\
      \ReasonAbove{=}{\QRef{def-subst-T}}
      \Tysubst^\top \cupdisjoint \medcupdisjoint_{i \in I}\Tysubst_i' \QLabel{eq-subst}
    \end{igather}
    We then have for all $i \in I$
    \begin{igather}
      \SubstSolves{\Tysubst \cupdisjoint \Tysubst'}{d_i \ConstrOr (\SubtyConstr{\alpha}{u_i'})} \QLabel{solve-di-or}
    \end{igather}
    because \ldots
    \begin{itemize}
    \item \ldots if $t_i \not\leq \bot$ then $\SubstSolves{\Tysubst \cupdisjoint \Tysubst'}{d_i}$ by
      (G3) and \QRef{eq-subst}.
    \item \ldots if $t_i \leq \bot$ then $\SubstSolves{\Tysubst \cupdisjoint \Tysubst'}{\SubtyConstr{\alpha}{u_i'}}$
      by the following reasoning:
      \begin{igather}
        \alpha(\Tysubst \cupdisjoint \Tysubst') = t_0 \ReasonAbove{\IsSubty}{\QRef{t0sub}}
        \UnionBig_{i \in I} \PgLowerTy{\Pg_i} \\
        u_i'(\Tysubst\cupdisjoint\Tysubst') \ReasonAbove{=}{\QRef{def-uip}~\textrm{and}~\Cref{lem:gp-props}}
        (\UnionBig_{j < i} \PgLowerTy{\Pg_j}) \Union \Neg \PgUpperTy{\Pg_i}\\
        t_i \ReasonAbove{=}{\QRef{def-ti}}
        (t_0 \WithoutTy \UnionBig_{j < i} \PgLowerTy{\Pg_j}) \Inter \PgUpperTy{\Pg_i}\\
      \end{igather}
      Then with $t_i \leq \bot$ by \Cref{lem:bot-implies-sub}
      \begin{igather}
        \alpha(\Tysubst \cupdisjoint \Tysubst') \IsSubty
        u_i' = u_i' (\Tysubst \cupdisjoint \Tysubst')
      \end{igather}
    \end{itemize}
    Further, we have with \QRef{entails-d0p} and \QRef{eq-subst} that
    $\SubstSolves{\Tysubst \cupdisjoint \Tysubst'}{d_0'}$. Hence with the definition of
    $d$ in \QRef{deriv-d} and with \QRef{solve-di-or} then
    \begin{igather}
     \SubstSolves{\Tysubst \cupdisjoint \Tysubst'}{d}
    \end{igather}
    as required. It remains to show that $\Dom{\Tysubst'} = A_1 \cupdisjoint A_2$.
    With (G2) we have
    $\Dom{\Tysubst_i'} = B_i'' \cupdisjoint B_i'''$
    and with \QRef{dom-theta-i-star} $\Dom{\Tysubst_i^\star} = B_i \cupdisjoint B_i'$
    for all $i \in I$.
    With \QRef{dom-theta0p} we have $\Dom{\Tysubst_0'} = B_0 \cupdisjoint B_0'$.
    With the definition of $\Tysubst'$ in \QRef{def-theta-p} then
    \begin{igather}
      \Dom{\Tysubst'} = \{\alpha,\beta\} \cupdisjoint B_0 \cupdisjoint B_0' \cupdisjoint
      \medcupdisjoint_{i \in I}(B_i \cupdisjoint B_i' \cupdisjoint B_i'')
    \end{igather}
    Finally, with \QRef{def-A1} and \QRef{def-A2} then $A_1 \cupdisjoint A_2 = \Dom{\Tysubst'}$.
  \end{CaseDistinction} \qed
\end{proof}

\begin{definition}
  The relations $\DefEnvGen{\DefSym}{\Senv}$
  and $\DefEnvGen{\DefSym_{i \in I}}{\Senv}$ are defined by the following inference rules:
  \begin{mathpar}
    \inferrule{
      \alpha~\textrm{fresh}
    }{
      \DefEnvGen{(x = e)}{ \{x : \alpha\} }
    }

    \inferrule{
      \FreeTyVars{\sigma} = \emptyset
    }{
      \DefEnvGen{(x : \sigma = e)}{ \{x : \sigma\} }
    }

    \inferrule{
      (\forall i \in I)~~\DefEnvGen{\DefSym_i}{\Senv_i}
    }{
      \DefEnvGen{\DefSym_{i \in I}}{ \Senv_{i \in I} }
    }
  \end{mathpar}
\end{definition}

\begin{lemma}[Completeness of environment generation for single definitions]
  \label{lem:completeness-algo-def-single}
  Assume $\DefEnvGen{\DefSym}{\Senv'}$ and $\Senv' \subseteq \Senv$ for some $\Senv$.
  If $\DefOk{\Senv\Tysubst}{\DefSym}$, then
  $\DefConstrGenV{\DefSym}{C}{\Senv'}{A_1}$ and
  $\ConstrRewFull{\Senv}{\EmptyEnv}{C}{d}{A_2}$
  and there exists $\Tysubst'$ with
  $\SubstSolves{\Tysubst \cupdisjoint \Tysubst'}{d}$
  and $\Dom{\Tysubst'} = (A_1 \cupdisjoint A_2) \setminus \FreeTyVars{\Senv'}$.
\end{lemma}

\begin{proof}
  \begin{CaseDistinction}{on the form of $\DefSym$}

    \CdCase{$\DefSym = (x = \lambda y.e)$}
    With $\DefEnvGen{\DefSym}{\Senv'}$ we have $\Senv' = \{x : \alpha\}$ for some
    fresh $\alpha$. From $\Senv' \subseteq \Senv$ then $\Senv(x) = \alpha$.
    From assumption $\DefOk{\Senv\Tysubst}{\DefSym}$, we get by inverting rule
    \Rule{d-def-no-annot}
    \begin{igather}
      \ExpSchemaTy{\Senv\Tysubst}{\EmptyEnv}{\lambda y.e}{\alpha\Tysubst}
    \end{igather}
    With \Cref{lem:completeness-algo-exp} there exists $\Tysubst'$ with
    \begin{igather}
      \ConstrGenV{\lambda y.e}{\alpha}{C}{A_1'}\\
      \ConstrRewV{\EmptyEnv}{C}{d}{A_2}\\
      \SubstSolves{\Tysubst \cupdisjoint \Tysubst'}{d}\\
      \Dom{\Tysubst'} = A_1' \cupdisjoint A_2
    \end{igather}
    With rule \Rule{c-def-no-annot} also
    $\DefConstrGenV{\DefSym}{C}{\Senv'}{A_1}$ where
    $A_1 = \{\alpha\} \cupdisjoint A_1'$.

    \CdCase{$\DefSym = (x : \sigma = \lambda y.e)$}
    From $\DefOk{\Senv\Tysubst}{\DefSym}$ by inverting rule \Rule{d-def-annot}
    \begin{igather}
      \Senv\Tysubst(x) = \sigma = \Senv(x)\\
      \FreeTyVars{\sigma} = \emptyset \QLabel{sigma-closed}\\
      \sigma = \TyScm{A}{\InterBig_{i \in I}{(t_i' \to t_i)}} \\
      \Forall{i \in I}\quad
      \ExpSchemaTy{\Senv\Tysubst}{\{y:t_i'\}}{e}{t_i}
    \end{igather}
    Hence, with $A$ fresh
    $\ExpSchemaTy{\Senv\Tysubst}{\{y:t_i'\}\Tysubst}{e}{t_i\Tysubst}~\Forall{i \in I}$.
    Thus with \Cref{lem:completeness-algo-exp} for all $i \in I$:
    \begin{igather}
      \ConstrGenV{e}{t_i}{C_i}{B_i} \\
      \ConstrRewV{y:t_i'}{C_i}{d_i}{B_i'} \QLabel{rew-di}\\
      \SubstSolves{\Tysubst \cupdisjoint \Tysubst_i}{d_i} \QLabel{solve-di}\\
      \Dom{\Tysubst_i} = B_i \cupdisjoint B_i' \QLabel{dom-theta-i}
    \end{igather}
    Further, $\Senv' = \{x : \sigma\}$ from the assumption
    $\DefEnvGen{\DefSym}{\Senv'}$. So with rule \Rule{c-def-annot}
    \begin{igather}
      \DefConstrGenV{\DefSym}{C}{\Senv'}{A_1} \\
      C = \MedConstrAnd\nolimits_{i \in I} (\DefConstr{ \{y : t_i'\} }{C_i})\\
      A_1 = \medcupdisjoint_{i \in I}B_i
    \end{igather}
    And with \QRef{rew-di} and rules \Rule{rc-def} and \Rule{rc-and}
    \begin{igather}
      \ConstrRewV{\EmptyEnv}{C}{d}{A_2}\\
      d = \MedConstrAnd_{i \in I}d_i\\
      A_2 = \medcupdisjoint_{i \in I}B_i'
    \end{igather}
    Define $\Tysubst' = \medcupdisjoint\nolimits_{i \in I} \Tysubst_i$.
    Then $\SubstSolves{\Tysubst \cupdisjoint \Tysubst'}{d}$ by \QRef{solve-di}.
    And with \QRef{dom-theta-i}
    $\Dom{\Tysubst'} = A_1 \cupdisjoint A_2$. Note that $\FreeTyVars{\Senv'} = \emptyset$
    from \QRef{sigma-closed}.
  \end{CaseDistinction} \qed
\end{proof}

\begin{lemma}[Completeness of environment generation for multiple definitions]
  \label{lem:completeness-algo-def-multi}
  Assume $\Senv = \Senv_{i \in I}$ and $\DefEnvGen{\DefSym_{i \in I}}{\Senv}$ and
  $\DefOk{\Senv\Tysubst}{\DefSym_i}$ for all $i \in I$.
  Then $\DefConstrGenV{\DefSym_i}{C_i}{\Senv_i}{A_i}$ for all $i \in I$ and some $C_i$. Further,
  $\ConstrRewFull{\Senv}{\EmptyEnv}{\MedConstrAnd\nolimits_{i \in I} C_i}{d}{A}$
  and there exists $\Tysubst'$ with
  $\SubstSolves{\Tysubst \cupdisjoint \Tysubst'}{d}$
  and
  $\Dom{\Tysubst'} = (A \cupdisjoint \medcupdisjoint\nolimits_{i \in I}A_i) \setminus \FreeTyVars{\Senv}$.
\end{lemma}

\begin{proof}
  The proof is by induction on $|I|$.
  For the induction to go through, we have to generalize the claim:

  \begin{quote}\em
    Assume $\Senv = \Senv_{i \in I}$ and $\DefEnvGen{\DefSym_{i \in I}}{\Senv}$ and
    $\DefOk{\Senv'\Tysubst}{\DefSym_i}$ for all $i \in I$ and some $\Senv'$ with $\Senv \subseteq \Senv'$.
    Then $\DefConstrGenV{\DefSym_i}{C_i}{\Senv_i}{A_i}$ for all $i \in I$ and some $C_i$. Further,
    $\ConstrRewFull{\Senv'}{\EmptyEnv}{\MedConstrAnd\nolimits_{i \in I} C_i}{d}{A}$
    and there exists $\Tysubst'$ with
    $\SubstSolves{\Tysubst \cupdisjoint \Tysubst'}{d}$
    and
    $\Dom{\Tysubst'} = (A \cupdisjoint \medcupdisjoint\nolimits_{i \in I}A_i) \setminus \FreeTyVars{\Senv}$.
  \end{quote}
  The base case with $|I| = 0$ holds trivially. For the induction step, assume $|I| = n$ and
  that the generalized claim already holds for $n - 1$. We then prove the claim for $n$
  using \Cref{lem:completeness-algo-def-single}. We omit the tedious but simple reasoning. \qed
\end{proof}

\begin{definition}[Relations between type schemes and scheme environments]~
  \begin{itemize}
  \item We say that type scheme $\sigma_1$ is more general than type scheme $\sigma_2$, written
    $\sigma_1 \MoreGeneral \sigma_2$, iff for all $t_2 \in \Inst{\sigma_2}$ there exists
    $t_1 \in \Inst{\sigma_1}$ with $t_1 \IsSubty t_2$.

  \item We say that a scheme environment $\Senv_1$ is more general than a scheme environment
    $\Senv_2$, written $\Senv_1 \MoreGeneral \Senv_2$, iff $\Dom{\Senv_1} = \Dom{\Senv_2}$
    and $\Senv_1(x) \MoreGeneral \Senv_2(x)$ for all $x \in \Dom{\Senv_1}$.
  \end{itemize}
\end{definition}

\begin{property}[Properties of $\TallySym$]\label{lem:props-tally}~
  \begin{EnumThm}
  \item If $\Tysubst \in \Tally{c}$ then $\SubstSolves{\Tysubst}{c}$.
  \item If $\SubstSolves{\Tysubst'}{c}$ then there exists $\Tysubst \in \Tally{c}$ and $\Tysubst''$
    with $\alpha \Tysubst' \TyEquiv \alpha\Tysubst\Tysubst''$ for all $\alpha \notin \FreeTyVars{\Tysubst}$.
  \item If $\Tysubst \in \Tally{c}$ then $\Dom{\Tysubst} = \FreeTyVars{c}$ and
    $\FreeTyVars{\Tysubst}$ are fresh.
  \end{EnumThm}
\end{property}

\begin{proof}
  See \citet[Theorem~A.35 on page~39]{journals/corr/CastagnaP016} and
  \citet{conf/popl/Castagna0XA15}. The generalization to constraints with disjunctions is straightforward.
\end{proof}
\begin{lemma}
  \label{lem:tally-better}
  If $\SubstSolves{\Tysubst}{c}$ and $\alpha \in \Dom{\Tysubst}$, then exists
  $\Tysubst' \in \Tally{c}$ such that $\Gen{\alpha\Tysubst'} \MoreGeneral \Gen{\alpha\Tysubst}$.
\end{lemma}

\begin{proof}
  \renewcommand\LabelQualifier{lem:tally-better}
  By \Cref{lem:props-tally}, we get $\Tysubst' \in \Tally{c}$ with $\FreeTyVars{\Tysubst'}$ fresh
  (so $\alpha \notin \FreeTyVars{\Tysubst'}$)
  and some $\Tysubst''$ such that
  \begin{igather}
    \alpha\Tysubst \TyEquiv \alpha\Tysubst'\Tysubst'' \QLabel{eq1}
  \end{igather}
  Suppose $A = \FreeTyVars{\alpha\Tysubst}$. Then $\Gen{\alpha\Tysubst} = \forall A . \alpha\Tysubst$.
  \Wlog{}, $A$ fresh and $A \cap \Dom{\Tysubst''} = \emptyset$.
  Suppose $u \in \Inst{\Gen{\alpha\Tysubst}}$. Hence, there exists $\TysubstStar$
  with $\Dom{\TysubstStar} = A$ and $u = \alpha\Tysubst\TysubstStar$. Then
  \begin{igather}
    u = \alpha\Tysubst\TysubstStar \ReasonAbove{\TyEquiv}{\QRef{eq1}}
    \alpha\Tysubst'\Tysubst''\TysubstStar
    \ReasonAbove{=}{A \cap \Dom{\Tysubst''} = \emptyset}
    \alpha\Tysubst'(\Tysubst'' \cupdisjoint \TysubstStar)
  \end{igather}
  Hence, $u \in \Inst{\Gen{\alpha\Tysubst'}}$, so
  $\Gen{\alpha\Tysubst'} \MoreGeneral \Gen{\alpha\Tysubst}$. \qed
\end{proof}

\begin{lemma}[Constraint rewriting with more general environments]
  \label{lem:rew-more-general}
  Assume $\ConstrRewFull{\Senv}{\Venv}{C}{d}{}$ and
  $\Senv' \MoreGeneral \Senv$ for some $\Senv'$.
  Then $\ConstrRewFull{\Senv'}{\Venv}{C}{d'}{A}$ for some $d'$.
  Further, if $\SubstSolves{\Tysubst}{d}$ for some $\Tysubst$ then
  $\SubstSolves{\Tysubst \cupdisjoint \Tysubst'}{d'}$
  for some $\Tysubst'$ with $\Dom{\Tysubst'} = A$.
\end{lemma}

\begin{proof}
  By induction on the derivation of
  $\ConstrRewFull{\Senv}{\Venv}{C}{d}{}$.
  \begin{CaseDistinction}{on the last rule of the derivation}
    \CdCase{Rule \Rule{rc-and}} Follows from the IH{}
    \CdCase{Rule \Rule{rc-subty}} Obvious.
    \CdCase{Rule \Rule{rc-var}} Simple because $\Senv' \MoreGeneral \Senv$ implies
    $\Dom{\Senv'} = \Dom{\Senv}$.
    \CdCase{Rule \Rule{rc-var-poly}}
    We have by inverting the rule:
    \begin{igather}
      x \notin \Dom{\Venv}\\
      \Senv(x) = \forall B.t\\
      C = \SubtyConstr{x}{t'}\\
      d = \SubtyConstr{t}{t'}
    \end{igather}
    We have $\Senv' \MoreGeneral \Senv$, so $\Senv'(x) = \forall A . u$.
    Then by rule \Rule{rc-var-poly}
    \begin{igather}
      \ConstrRewFull{\Senv'}{\Venv}{C}{\BraceBelow{\SubtyConstr{u}{t'}}{{}=: d'}}{A}
    \end{igather}
    Now assume $\SubstSolves{\Tysubst}{d}$. Then
    \begin{igather}
      t\Tysubst \IsSubty t'\Tysubst \QLabel{t-sub-tp}
    \end{igather}
    Noting that
    $t \in \Inst{\Senv(x)}$, so with $\Senv' \MoreGeneral \Senv$, we have some
    $\TysubstStar$ with $\Dom{\TysubstStar} = A$ and
    \begin{igather}
      u\TysubstStar \IsSubty t \QLabel{ustar-sub-t}
    \end{igather}
    Assume $A = \{\alpha_i \mid i \in I\}$ and \WlogLower{} $A \cap \Dom{\Tysubst} = \emptyset$.
    Let $\TysubstStar = [ t_i/\alpha_i \mid i \in I ]$ for some $t_i$.
    Define $\Tysubst' = [ t_i\Tysubst/\alpha_i \mid i \in I ]$.
    Then $\Dom{\Tysubst'} = A$ and $\Dom{\Tysubst} \cap \Dom{\Tysubst'} = \emptyset$ and we have
    \begin{igather}
      u(\Tysubst \cupdisjoint \Tysubst') = u\TysubstStar\Tysubst
      \ReasonAbove{\IsSubty}{\QRef{ustar-sub-t} ~\textrm{and}~ \textrm{\Cref{prop:subty}}} t\Tysubst
      \ReasonAbove{\IsSubty}{\QRef{t-sub-tp}} t'\Tysubst
      \ReasonAbove{=}{A~\textrm{fresh}} t'(\Tysubst \cupdisjoint \Tysubst')
    \end{igather}
    Hence, $\SubstSolves{\Tysubst \cupdisjoint \Tysubst'}{d'}$.
    \CdCase{Rule \Rule{rc-def}} Straightforward from the IH{}
    \CdCase{Rule \Rule{rc-case}}
    We get by inverting the rule:
    \begin{igather}
      \inferrule{
        \ConstrRewV{\Venv}{C'}{c'}{}\\
        \Forall{i \in I}~ \ConstrRewV{\Venv}{\hat{C_i}}{\hat{c_i}}{}\\
        \ConstrRewV{\Venv \InterEnv \Venv_i}{C_i}{c_i}{}
      }{
        \ConstrRewV{\Venv}{
          \BraceBelow{\CaseConstr{C'}{\Multi{(\InConstr{\Venv_i}{C_i}{\hat{C_i}})}}}{{}= C}
        }{
          \BraceBelow{c' \ConstrAnd \MedConstrAnd\nolimits_{i \in I}(\hat{c_i} \ConstrOr c_i)}{{}= d}
        }{}
      }
    \end{igather}
    We get from the \IH{}:
    \begin{igather}
      \ConstrRewFull{\Senv'}{\Venv}{C'}{d''}{B} \\
      \ConstrRewFull{\Senv'}{\Venv}{\hat{C_i}}{\hat{d_i}}{\hat{B_i}} \\
      \ConstrRewFull{\Senv'}{\Venv \InterEnv \Venv_i}{C_i}{d_i}{B_i}
    \end{igather}
    We then have by rule \Rule{rc-case}
    \begin{igather}
      \ConstrRewFull{\Senv'}{\Venv}{C}{d'}{A}\\
      d' = d'' \ConstrAnd \MedConstrAnd\nolimits_{i \in I}(\hat{d_i} \ConstrOr d_i)\\
      A = B \cupdisjoint \medcupdisjoint\nolimits_{i \in I}(B_i \cupdisjoint \hat{B_i})
    \end{igather}
    If we additionally assume $\SubstSolves{\Tysubst}{d}$, then
    $\SubstSolves{\Tysubst}{c'}$ and for all $i \in I$ either
    $\SubstSolves{\Tysubst}{\hat{c_i}}$ or $\SubstSolves{\Tysubst}{c_i}$.
    Then we have also from the \IH{} for some $\TysubstStar$ that
    \begin{igather}
      \SubstSolves{\Tysubst \cupdisjoint \TysubstStar}{d''}\\
      \Dom{\TysubstStar} = B
    \end{igather}
    For $i \in I$ we have from the \IH{} for some $\Tysubst_i$ either
    $\SubstSolves{\Tysubst_i}{\hat{d_i}}$ with $\Dom{\Tysubst_i} = \hat{B_i}$
    or $\SubstSolves{\Tysubst_i}{d_i}$ with $\Dom{\Tysubst_i} = B_i$.
    \Wlog{}, we can extend the domain of $\Tysubst_i$ such that
    $\Dom{\Tysubst_i} = B_i \cupdisjoint \hat{B_i}$.
    Define $\Tysubst' = \TysubstStar \cupdisjoint \medcupdisjoint\nolimits_{i \in I}\Tysubst_i$.
    Then $\Dom{\Tysubst'} = A$ and $\SubstSolves{\Tysubst \cupdisjoint \Tysubst'}{d'}$ as required.
  \end{CaseDistinction} \qed
\end{proof}

\begin{lemma}
  \label{lem:solves-equiv}
  If $\Dom{\Tysubst} \cap \FreeTyVars{\Tysubst} = \emptyset$ then
  $\SubstSolves{\Tysubst}{\Equiv{\Tysubst}}$.
\end{lemma}

\begin{proof}
  Assume $\alpha \in \Dom{\Tysubst}$.
  From $\Dom{\Tysubst} \cap \FreeTyVars{\Tysubst} = \emptyset$ we get
  $\alpha\Tysubst\Tysubst = \alpha\Tysubst$, from which
  $\SubstSolves{\Tysubst}{\Equiv{\Tysubst}}$ follows directly. \qed
\end{proof}

\begin{lemma}[Free type variables in constraint generation and rewriting]\label{lem:free-tyvars-constr}~
  \begin{EnumThm}
  \item
    If $\ConstrRewV{\Venv}{C}{d}{A}$ then
    $\FreeTyVars{d} \subseteq \FreeTyVars{\Senv} \cup \FreeTyVars{\Venv} \cup A \cup \FreeTyVars{C}$
  \item
    If $\ConstrGenV{e}{t}{C}{A}$ then $\FreeTyVars{C} \subseteq A \cup \FreeTyVars{t}$.
  \item
    If $\DefConstrGenV{\DefSym}{C}{\Senv}{A}$ then $\FreeTyVars{C} \subseteq A$.
  \end{EnumThm}
\end{lemma}

\begin{proof}
  The proofs of (i) and (ii) are by induction on the respective derivations. Claim (iii) follows
  with (ii).
\end{proof}

We now restate and prove the completeness theorem for algorithmic typing from
\Cref{sec:algor-typing-rules}.

\paragraph{Theorem~\ref*{lem:completeness-algo-prog}
  \textnormal{(Completeness of algorithmic typing for programs)}}
\emph{
    Given program $\Prog$ and some type $t$.
    If $\ProgTy{\Prog}{t}$ then there exist program constraint $P$, simple constraint $d$,
    and a type substitution $\Tysubst$ such that
    $\ConstrGenV{\Prog}{t}{P}{}$ and $\ConstrRewProgV{P}{d}{}$ and $\SubstSolves{\Tysubst}{d}$.
}

\begin{proof}
  \renewcommand\LabelQualifier{lem:completeness-algo-prog}
  Let $\Prog = \Letrec{\DefSym_{i \in I}}{e}$. We get by inverting rule \Rule{d-prog}:
  \begin{igather}
    \Forall{i \in I}~\Envs{\DefSym_i} = \MetaPair{\Senv_i}{\Senv_i'} \QLabel{envs-defi}\\
    \DefOk{\Senv_{i \in I}}{\DefSym_i} \QLabel{defi-ok}\\
    \ExpSchemaTy{\Senv'_{i \in I}}{\EmptyEnv}{e}{t} \QLabel{type-e}
  \end{igather}
  Define $\Senv = \Senv_{i \in I}$ and $\Senv' = \Senv'_{i \in I}$. Further,
  assume $\Senv'' = \Senv''_{i \in I}$ such that
  \begin{igather}
    \DefEnvGen{\DefSym_{i \in I}}{\Senv_i''}\quad(\forall i \in I) \QLabel{defi-gen}
  \end{igather}
  Note that for $\DefSym_i = (x : \sigma = e)$ the premise $\FreeTyVars{\sigma} = \emptyset$
  required for $\DefEnvGen{\DefSym_{i \in I}}{\Senv_i''}$ follows from \QRef{envs-defi} by
  inverting rule \Rule{d-def-annot}.

  \Wlog{}, we may assume for $n = |I|$
  \[
    \begin{array}{r@{{}={}}l@{~\ldots,~}l@{~}c@{~\ldots,~}l}
      \DefSym_{i \in I} &    (x_1 = e_1),           & (x_k = e_k),            & (x_{k+1} : \sigma_{k+1} = e_{k+1}),   & (x_n: \sigma_n = e_n) \\
      \Senv''           & \{ x_1 = \alpha_1,        & x_k = \alpha_k,         & x_{k+1} = \sigma_{k+1},               & x_n = \sigma_n \} \\
      \Senv             & \{ x_1 = t_1,             & x_k = t_k,              & x_{k+1} = \sigma_{k+1},               & x_n = \sigma_n \} \\
      \Senv'            & \{ x_1 = \forall A_1.t_1, & x_k = \forall A_k.t_k,  & x_{k+1} = \sigma_{k+1},               & x_n = \sigma_n \}
    \end{array}
  \]
  with $\sigma_i = \forall A_i . t_i$ for $i = k+1,\ldots,n$. With \QRef{envs-defi} we have
  by inverting rules \Rule{d-envs-no-annot} and \Rule{d-envs-annot} that
  $A_i = \FreeTyVars{t_i}$ for $i = 1,\ldots,n$.
  Define $\TysubstTild := [ t_1/\alpha_1, \ldots, t_k/\alpha_k ]$.
  Then  $\Senv = \Senv'' \TysubstTild$, so with \QRef{defi-ok} then
  \begin{igather}
    \DefOk{\Senv''\TysubstTild}{\DefSym_i} \quad (\forall i \in I)
  \end{igather}
  With \QRef{defi-gen} and \Cref{lem:completeness-algo-def-multi} then
  \begin{igather}
    \DefConstrGenV{\DefSym_i}{C_i}{\Senv_i''}{B_i} \quad (\forall i \in I) \QLabel{gen-Ci-senvi}\\
    \ConstrRewFull{\Senv''}{\EmptyEnv}{\MedConstrAnd\nolimits_{i \in I} C_i}{\tilde{d}}{B}
    \QLabel{rew-Ci}\\
    \SubstSolves{\TysubstTild \cupdisjoint \TysubstStar}{\tilde{d}} \QLabel{solve-di}\\
    \Dom{\TysubstStar} = B \cupdisjoint \medcupdisjoint\nolimits_{i \in I} B_i \setminus
    \FreeTyVars{\Senv''} \QLabel{dom-theta-star}
  \end{igather}
  for some $\TysubstStar$ and $B, B_i$ fresh.
  We have $\ExpSchemaTy{\Senv'}{\EmptyEnv}{e}{t}$ with \QRef{type-e}, so with \Cref{lem:completeness-algo-exp}
  \begin{igather}
    \ConstrGenV{e}{t}{C'}{B'} \QLabel{gen-Cp} \\
    \ConstrRewFull{\Senv'}{\EmptyEnv}{C'}{d'}{B''} \QLabel{rew-Cp} \\
    \SubstSolves{\Tysubst'}{d'} \QLabel{solve-dp}
  \end{igather}
  for some $\Tysubst'$ with $\Dom{\Tysubst'} = B' \cupdisjoint B''$ and $B', B''$ fresh.
  From \QRef{solve-di} we have with \Cref{lem:props-tally} and \Cref{lem:tally-better}
  that there exists some $\TysubstTop$ such that
  \begin{igather}
    \TysubstTop \in \Tally{\tilde{d}} \QLabel{in-tally}\\
    \SubstSolves{\TysubstTop}{\tilde{d}} \\
    \Dom{\TysubstTop} \cap \FreeTyVars{\TysubstTop} = \emptyset \QLabel{dom-disjoint-free} \\
    \Gen{\alpha_i \TysubstTop} \MoreGeneral
    \Gen{\alpha_i (\TysubstTild \cupdisjoint \TysubstStar)}
    \ReasonAbove{=}{\Dom{\TysubstStar}~\textrm{fresh}}
    \Gen{\alpha_i \TysubstTild} \quad\textrm{for all}~i = 1,\ldots,k, 
  \end{igather}  
    Further, $\Senv' = \Gen{\Senv} = \Gen{\Senv''\TysubstTild}$.
  To show $\Gen{\Senv''\TysubstTop} \MoreGeneral \Senv'$, we note that both environments have the
  same domain $\{x_1,\ldots,x_n\}$ and
  $\Gen{\Senv''\TysubstTop}(x_j) = \sigma_j = \Senv'(x_j)$ for $j \in \{k+1,\ldots,n\}$ (note that $\sigma_j$ is closed) and
  $\Gen{\Senv''\TysubstTop}(x_i) = \Gen{\alpha_i \TysubstTop} \MoreGeneral \Gen{\alpha_i \TysubstTild} = \Senv'(x_i)$
  for $i \in \{1,\ldots,k\}$.

  Then with \Cref{lem:rew-more-general} and \QRef{rew-Cp}, \QRef{solve-dp}
  \begin{igather}
    \ConstrRewFull{\Gen{\Senv''\TysubstTop}}{\EmptyEnv}{C'}{d''}{B'''} \QLabel{rew-Cp2}\\
    \SubstSolves{\Tysubst' \cupdisjoint \Tysubst''}{d''} \QLabel{solve-dpp}\\
    \Dom{\Tysubst''} = B'''\QLabel{dom-thetapp}
  \end{igather}
  for some $\Tysubst''$.
  With \QRef{gen-Ci-senvi} and \QRef{gen-Cp} and rule \Rule{c-prog}:
  \begin{igather}
    \ConstrGenV{\Prog}{t}{P}{}\\
    P = \LetConstr{\MedConstrAnd_{i \in I}C_i}{\Senv''}{C'}
  \end{igather}
  With rule \Rule{rc-prog}, \QRef{rew-Ci}, \QRef{in-tally}, and \QRef{rew-Cp2} also
  \begin{igather}
    \ConstrRewProgV{P}{d}{} \\
    d = \Equiv{\TysubstTop} \ConstrAnd d''
  \end{igather}
  Define $\Tysubst = \TysubstTop \cupdisjoint \Tysubst' \cupdisjoint \Tysubst''$.
  For disjointness, we reason as follows:
  \begin{itemize}
  \item $\Dom{\TysubstTop}$ is disjoint from $\Dom{\Tysubst'}$:
    By \QRef{in-tally} and \Cref{lem:props-tally}(iii)
    \begin{igather}
      \Dom{\TysubstTop} = \FreeTyVars{\tilde{d}}
    \end{igather}
    From \QRef{rew-Ci} and \Cref{lem:free-tyvars-constr}(i) we have
    \begin{igather}
      \FreeTyVars{\tilde{d}} \subseteq \FreeTyVars{\Senv''} \cup B \medcup_{i \in I} \FreeTyVars{C_i}
    \end{igather}
    With \QRef{gen-Ci-senvi} and \Cref{lem:free-tyvars-constr}(iii) $\FreeTyVars{C_i} \subseteq B_i$ for
    all $i \in I$. Hence
    \begin{igather}
      \Dom{\TysubstTop} \subseteq \{\alpha_1,\ldots,\alpha_k\} \cup B \cup \medcup_{i \in I} B_i
    \end{igather}
    Noting that $\Dom{\Tysubst'} = B' \cupdisjoint B''$, where $B'$ and $B''$ are fresh
    with respect to
    $\alpha_1,\ldots,\alpha_k$ and $B$ and all $B_i$, proves that
    $\Dom{\TysubstTop} \cap \Dom{\Tysubst'} = \emptyset$.
  \item $\Dom{\TysubstTop}$ is disjoint from $\Dom{\Tysubst''}$ because $\Dom{\Tysubst''} = B'''$ is
    also fresh with respect to $\alpha_1,\ldots,\alpha_k$ and $B$ and all $B_i$.
  \item $\Dom{\Tysubst'}$ is disjoint from $\Dom{\Tysubst''}$ because of \QRef{solve-dpp}.
  \end{itemize}

  From \QRef{solve-dpp} we get $\SubstSolves{\Tysubst}{d''}$, and from
  \QRef{dom-disjoint-free} and \Cref{lem:solves-equiv} we have
  $\SubstSolves{\Tysubst}{\Equiv{\TysubstTop}}$. Hence
  $\SubstSolves{\Tysubst}{d}$ as required. \qed
\end{proof}

\FloatBarrier
\section{Empirical Analysis of Pattern Matching}
\label{sec:empirical-patterns}

\begin{table}[tpb]
  \center

\begin{tabular}{|l|r|r|r|r|rr|rr|}
\hline
\multicolumn{1}{|c|}{\textbf{project}} & \multicolumn{1}{c|}{\textbf{files}} & \multicolumn{1}{c|}{\textbf{\begin{tabular}[c]{@{}c@{}}pattern \\ matches\end{tabular}}} & \multicolumn{1}{c|}{\textbf{branches}} & \multicolumn{1}{c|}{\textbf{\begin{tabular}[c]{@{}c@{}}type \\ tests\end{tabular}}} & \multicolumn{2}{c|}{\textbf{\begin{tabular}[c]{@{}c@{}}type tests on\\ variables\end{tabular}}} & \multicolumn{2}{c|}{\textbf{\begin{tabular}[c]{@{}c@{}}type tests on\\ bound variables\end{tabular}}} \\
                                       & \multicolumn{1}{l|}{}               & \multicolumn{1}{l|}{}                                                                    & \multicolumn{1}{l|}{}                  & \multicolumn{1}{l|}{}                                                               & \multicolumn{1}{c|}{\textit{abs.}}              & \multicolumn{1}{c|}{\textit{\%}}              & \multicolumn{1}{c|}{\textit{abs.}}                 & \multicolumn{1}{c|}{\textit{\%}}                 \\ \hline\hline
antidote                               & 55                                  & 1635                                                                                     & 2310                                   & 4                                                                                   & \multicolumn{1}{r|}{4}                          & 100.0                                         & \multicolumn{1}{r|}{4}                             & 100.0                                            \\ \hline
etylizer                               & 45                                  & 1792                                                                                     & 3084                                   & 17                                                                                  & \multicolumn{1}{r|}{17}                         & 100.0                                         & \multicolumn{1}{r|}{15}                            & 88.2                                             \\ \hline
rabbitmq                               & 296                                 & 18886                                                                                    & 31305                                  & 1759                                                                                & \multicolumn{1}{r|}{1749}                       & 99.4                                          & \multicolumn{1}{r|}{1719}                          & 98.3                                             \\ \hline
riak\_core                             & 104                                 & 5572                                                                                     & 8034                                   & 268                                                                                 & \multicolumn{1}{r|}{268}                        & 100.0                                         & \multicolumn{1}{r|}{263}                           & 98.1                                             \\ \hline
stdlib                                 & 93                                  & 18482                                                                                    & 53841                                  & 3383                                                                                & \multicolumn{1}{r|}{3358}                       & 99.3                                          & \multicolumn{1}{r|}{3307}                          & 98.5                                             \\ \hline
\textbf{Total}                         & \textbf{593}                        & \textbf{46367}                                                                           & \textbf{98574}                         & \textbf{5431}                                                                       & \multicolumn{1}{r|}{\textbf{5396}}              & \textbf{99.4}                                 & \multicolumn{1}{r|}{\textbf{5308}}                 & \textbf{98.4}                                    \\ \hline
\end{tabular}
\vspace{2mm}
\caption{Empirical study on the use of type tests in pattern matching. \\
  The following versions of the projects were examined:
  antidote git commit 1482b81 from 2022-09-14;
  etylizer git commit 4822984 from 2025-05-28,
  rabbitmq git commit ee652cb from 2025-07-02;
  riak\_core git commit 7b04e2c from 2023-06-12;
  stdlib version 27.1.1.}
  \label{table:pattern-analysis}
\end{table}

To guide our design choices, we conducted an empirical study on the
use of type tests in pattern matching across six representative Erlang
projects. The results are shown
in \Cref{table:pattern-analysis}.

In total, we analyzed 593
source files containing 46,367 pattern matches with 5,431 type tests.
Our results show that type tests are almost
exclusively applied to variables (99.4\%), and predominantly to
variables bound within the pattern itself (98.4\%).
Consequently, our type system focuses on providing precise support for this class of type tests.


\FloatBarrier
\section{Empirical Analysis of Dynamic Functions}
\label{sec:empirical-dynamic-functions}

To assess the practical need for full value-dependent typing, we analyzed the usage of the \lstinline|erlang:element/2|, \lstinline|erlang:setelement/3| and functions 12 in the lists module that access a tuple dynamically (e.g. \lstinline|lists:ukeysort/2|) across a corpus of six Erlang projects. 
The goal was to determine how often the index argument is a statically known literal constant versus a dynamically computed value. 
Our hypothesis was that the majority of uses in functional code are with constant indices, which can be effectively handled by the type overlays in conjunction with singleton integer types.

The results, summarized in \Cref{table:dynamic-analysis} support this hypothesis.
We found 1820 total function calls to \lstinline|element/2|, \lstinline|setelement/3|, and the list module functions. 
A significant majority of these calls, 1386 (76\%), used a literal integer (e.g., \lstinline|element(2, Tuple)|) as the index. 
These cases are suited for our overloaded type approach, which provides a precise return type for each constant index. 
The remaining 434 calls (24\%) used a non-constant expression as the index.

It is noteworthy that the standard library accounts for nearly all of the non-constant index cases. 
This is expected, as stdlib contains generic utility functions designed to handle arbitrary, user-supplied inputs. 
When stdlib is excluded from the analysis, of the 743 non-constant index calls in the other five projects, 645 are verifiable with our current overlays. 
The amount of code that can be captured with our type system increases to 87\%.

\begin{table}[tpb]
  \centering
  \begin{minipage}{0.48\textwidth}
    \centering
    \begin{tabular}{|l|r|r|}
    \hline
    Function & Literal & Non-Lit \\
    \hline
    erlang:element & 741 & 326 \\
    \hline
    erlang:setelement & 30 & 68 \\
    \hline
    lists:keydelete & 58 & 3 \\
    \hline
    lists:keyfind & 268 & 12 \\
    \hline
    lists:keymap & 0 & 1 \\
    \hline
    lists:keymember & 59 & 1 \\
    \hline
    lists:keymerge & 2 & 0 \\
    \hline
    lists:keyreplace & 28 & 4 \\
    \hline
    lists:keysearch & 59 & 3 \\
    \hline
    lists:keysort & 59 & 11 \\
    \hline
    lists:keystore & 33 & 0 \\
    \hline
    lists:keytake & 24 & 3 \\
    \hline
    lists:ukeymerge & 5 & 0 \\
    \hline
    lists:ukeysort & 20 & 2 \\
    \hline
        \textbf{total} & \textbf{1386} & \textbf{434} \\
        \hline
    \end{tabular}
  \end{minipage}
  \hfill
  \begin{minipage}{0.48\textwidth}
    \centering
    \begin{tabular}{|c|c|c|c|}
    \hline
    Project & Files & Literal & Non-lit\\
    \hline
    antidote & 55 & 27 & 4 \\
    \hline
    etylizer & 45 & 6 & 20 \\
    \hline
    rabbitmq & 296 & 546 & 62 \\
    \hline
    riak\_core & 104 & 134 & 25 \\
    \hline
    stdlib & 93 & 673 & 321 \\
    \hline
    jsone & 4 & 0 & 2 \\
    \hline
    \textbf{total} & \textbf{597} & \textbf{1386} & \textbf{434} \\
    \hline
    \end{tabular}
  \end{minipage}
  
  \vspace{2mm}
  \caption{Empirical study on the use of dependent types in dynamic functions. \\
    The following versions of the projects were examined:
    antidote git commit 1482b81 from 2022-09-14;
    etylizer git commit 4822984 from 2025-05-28,
    rabbitmq git commit ee652cb from 2025-07-02;
    riak\_core git commit 7b04e2c from 2023-06-12;
    stdlib version 27.1.1;
    jsone git commit a1c481f from 2024-28-11.}
  \label{table:dynamic-analysis}
\end{table}

\end{document}